\documentclass[twocolumn,amssymb, amsmath, showpacs, pra, notitlepage]{revtex4-1}
\usepackage{}

\usepackage{graphicx,graphics,epsfig,subfigure,times,bm,bbm,amssymb,amsmath,amsfonts,amsthm,mathrsfs,MnSymbol}
\usepackage[matrix,frame,arrow]{xypic}
\usepackage[pdfstartview=FitH]{hyperref}
\usepackage[pdftex]{color}
\usepackage{enumerate}
\usepackage{epstopdf}
\usepackage{pdfsync}
\usepackage{relsize}

\input{Qcircuit.tex}

\newtheorem{definitionenv}{Definition}
\newtheorem{remarkenv}[definitionenv]{Remark}
\newenvironment{remark}{\begin{remarkenv}\rm}{\end{remarkenv}}

\newtheorem{mydef}{Definition}

\newtheorem{mytheorem}{Theorem}
\newtheorem{mylemma}{Lemma}
\newtheorem{mycorollary}{Corollary}

\newcommand{\bes} {\begin{subequations}}
\newcommand{\ees} {\end{subequations}}
\newcommand{\bea} {\begin{eqnarray}}
\newcommand{\eea} {\end{eqnarray}}

\newcommand{\beq}{\begin{equation}}

\newcommand{\eeq}{\end{equation}}

\newcommand{\ignore}[1]{}

\def\G{\Gamma}

\def\>{\rangle}
\def\<{\langle}
\def\Tr{\mathrm{Tr}}

\begin{document}
\title{Fault-tolerant Holonomic Quantum Computation in Surface Codes}

\author{Yi-Cong Zheng}
\email{
yicongzh@usc.edu
}
\author{Todd A. Brun}
\email{tbrun@usc.edu}
\affiliation{Ming Hsieh Department of Electrical Engineering, Center for Quantum Information Science \& Technology, University of Southern California, Los Angeles, California, 90089\\}

%\date{\today}

\begin{abstract}
We show that universal holonomic quantum computation (HQC) can be achieved fault-tolerantly by adiabatically deforming the gapped stabilizer Hamiltonian of the surface code, where quantum information is encoded in the degenerate ground space of the system Hamiltonian. We explicitly propose procedures to perform each logical operation, including logical state initialization, logical state measurement, logical CNOT, state injection and distillation,etc. In particular, adiabatic braiding of different types of holes on the surface leads to a topologically protected, non-Abelian geometric logical CNOT. Throughout the computation, quantum information is protected from both small perturbations and low weight thermal excitations by a constant energy gap, and is independent of the system size. Also the Hamiltonian terms have weight at most four during the whole process. The effect of thermal error propagation is considered during the adiabatic code deformation. With the help of active error correction, this scheme is fault-tolerant, in the sense that the computation time can be arbitrarily long for large enough lattice size. It is shown that the frequency of error correction and the physical resources needed can be greatly reduced by the constant energy gap.
\end{abstract}
\maketitle

%\tableofcontents

\section{Introduction}
Quantum computers (QCs) provide the means to solve certain problems that cannot be handled classically; yet they are extremely vulnerable to errors during the computation~\cite{Nielsen:2000:CambridgeUniversityPress}. The threshold theorem indicates that if errors are all local and their rates are below a certain threshold, it is possible to implement large scale quantum computation with arbitrarily small error~\cite{Aharonov:1997:176, Gottesman:9705052, DivencenzoFTPhysRevLett.77.3260, KnillFTNature, QECbook:2013} based on active quantum error correction (QEC). However the threshold is difficult to achieve, and tremendous physical resources are required, making QCs very difficult to build in practice.

In addition to protecting QCs by active QEC, much work has been done on providing inherent robustness through the hardware design, such as
holonomic quantum computation (HQC)~\cite{Zanardi:1999:94}, adiabatic quantum computing (AQC)~\cite{Farhi:0001106,FarhiScience}, topological quantum computation (TQC)~\cite{Kitaev:2003:2,freedman2002TQFsimulation,Nayak:2008:1083}. However, these methods all have advantages and disadvantages, which are detailed below. In this paper, we will combine the good features of these architectures and avoid their weakness by proposing the scheme of fault-tolerant HQC in surface codes.

Holonomic QC uses the non-Abelian generalization of Berry phase~\cite{Wilczek:1984:2111} induced by deforming the Hamiltonian adiabatically and cyclic (closed-loop) to obtain unitary gates in the ground space. These unitary gates depend only on the geometry of the paths in the control manifold. This approach has been shown to be robust against various types of errors during the process~\cite{Solinas:2004:042316,sarandy2006abelian-nonabelian_geometric,solinas2012stability_HQC} and could in principle be done in several different systems~\cite{Duan:2001:1695, Yicongzheng2012geometric, renes2013holonomic_symmetryprotected}. Both closed-loop and open-loop HQCs can be compatible with active QEC~\cite{OgyanHolonomicQCPhysRevLett.102.070502,
OgyanHolonomicQcPhysRevA.80.022325,
Bacon-AGTPhysRevLett.103.120504,
bacon2010adiabaticcluster,Yi-Cong_PhysRevA.89.032317}, and can achieve fault-tolerant QC. However, for small quantum systems, it is difficult to maintain the degeneracy of the ground space, which is easily broken by even small perturbations, causing unavoidable phase errors.

Another method is to use adiabatic quantum computing
(AQC) by slowly changing the Hamiltonian to
a special final Hamiltonian, whose ground state encodes the
solution of the problem to be solved~\cite{Farhi:0001106,FarhiScience}. This method completely drops the standard circuit model. AQC can suppress thermal noise when the evolution is very slow~\cite{TameemNJP:adiabaticMarkovianME}, because of the the non-zero  energy gap between the ground state and the other excited states. While considerable work has been done in this direction, such as in Ref.~\cite{Jordan:2005:052322,LidarTowardFTAdqcPhysRevLett.100.160506},
a fault-tolerance theorem for AQC is still lacking. Typically, the minimum energy gap of the system scales as an inverse polynomial in the problem size~\cite{StewartEquivalenceADCPhysRevA.71.062314,
AriEquivalenceADCPhysRevLett.99.070502}, so that the temperature must be arbitrarily low to prevent thermal excitation.

A third method is the beautiful idea of topological quantum computation (TQC) first introduced by Kitaev~\cite{Kitaev:2003:2}, where excited states of system Hamiltonian behave like particles with exotic statistics, called anyons. By adiabatically braiding anyons around one another in space-time, it induces the unitary operation that depends only on the topology of the anyon world lines. Remarkably, some systems can support non-Abelian anyons, perform universal quantum computation on information encoded in the label space of the anyons~\cite{koenig2010quantumUTQC}, while being protected by an energy gap independent of the system size.
Unlike HQC, TQC is immune to the effect of small perturbations, since quantum information is stored and processed nonlocally, so that the splitting of the degenerate ground space will decrease exponentially with the system size~\cite{Bravyi:2010:093512}. However, this topological protection does not completely eliminate the need for active error correction. The energy gap can protect information only to a certain extent, and unwanted anyons could be created
if the computation time is long enough. Besides, unwanted anyons may be generated during the process of creation, fusion and imperfect adiabatic motion of anyons, and they may not be detectable.
One must measure anyon occupations to
determine when and where unwanted
anyons are created~\cite{wootton2014error_decoding_non_abelianTQC}, but this is usually difficult in most TQC models (like fractional quantum Hall systems).

%However, for schemes based on non-Ablian anyons the requirement for active error correction is often ignored. The presence of an energy gap that suppress the creation of unwanted anyons may seem to replace the need for error correction. However, an constant gap will only suppress anyon creation to a limited extent. Once the computation scale becomes sufficiently large, in terms of either the size of computer or its runtime, the presence of unwanted anyons becomes almost certain. The lack of error correction then means the lack of scalability which is a requirement for fault-tolerant QC.  Research into the corresponding methods required for non-Abelian error correction is long overdue, only recently

On the other hand, a combination of ideas from TQC and QEC gives schemes of active error correction architecture based on topological QEC codes, especially the surface codes~\cite{Raussendorf:2007:190504,Raussendorf:2007:199} and color codes~\cite{landahl2011fault_color}, using code deformation~\cite{, Bombin:2009:095302,bombin2011clifford}. In this approach, one works directly with the quantum error correcting code used in TQC, without introducing a Hamiltonian to protect quantum information with energy gap~\cite{Dennis:2002:4452}. In the case of surfaces code, one truncates it by turning off some stabilizer generators in a region to create a hole or defect.
Rather than encoding information in the label space of anyons in TQC,  each hole can be viewed as an encoded qubit. Via a sequence of measurements, the boundary of holes can be deformed. One can then braid holes by using suitable deformations to perform logical operations between logical qubits associated with the holes. Because of its tolerance of local errors~\cite{Dennis:2002:4452}, scalable structure and high threshold ($0.57\%$)~\cite{Fowler:2009:052312, Folwer2012PhysRevA.86.032324}, surface codes have attracted a great deal of attention, and impressive experimental progress in this direction has been made recently with superconducting qubits~\cite{barends_Martinis2014superconducting}.

In this paper, we try to combine the best features of all the architectures mentioned above, and avoid their weakness. We focus on surface codes with a stabilizer Hamiltonian turned on to form a topological quantum memory~\cite{Dennis:2002:4452,Wotton:2012QMemoryReview} on a \emph{single} 2D lattice, to protect quantum information encoded in the degenerate ground space from both thermal errors and perturbations. We explicitly construct all processes needed to do universal holonomic quantum computation (HQC) based on the surface code, by adiabatically deforming this gapped Hamiltonian. By adiabatically braiding different types of holes on the surface, one performs a topologically protected non-Abelian geometric logical CNOT gate. Throughout the entire information processing procedure, including logical state initialization, logical state measurement, logical gates, state injection and distillation, quantum information is protected from local thermal excitations by a constant energy gap, and the weight of the Hamiltonian terms is bounded by 4 during the whole adiabatic code deformation process. To deal with unwanted excitations caused by errors (creation of anyons) during the adiabatic code deformation, we analyze errors propagation, and give conditions when turning off the stabilizer Hamiltonian is needed to do syndrome measurement and error correction. It can be shown that with gap protection the frequency of error correction and the physical resources needed can be greatly reduced. We conclude that the computation procedures are scalable, and that the scheme is fault tolerant.

\section{Preliminary}

\subsection{Surface Code}\label{sec:surface_code}
A good introduction to the surface code can be found in Refs.~\cite{Fowler:2009:052312, Folwer2012PhysRevA.86.032324}. In this section, we follow Ref.~\cite{Folwer2012PhysRevA.86.032324} and give a brief review to establish our notation.
Surface codes can be viewed as a special kind of stabilizer codes defined on a 2D square lattice. In this paper,
we implement the surface code on a two-dimensional $L\times L$ lattice, with qubits on the edges of the lattice, as shown in Fig.~\ref{Fig:surface_code_illustration}
for $L=8$. The stabilizer generators of surface codes are two different kinds of operators:
%\beq
%\begin{split}
%I&=\left(
%    \begin{array}{cc}
%      1 & 0 \\
%      0 & 1 \\
%    \end{array}
%  \right), \quad
%\sigma_x=\left(
%    \begin{array}{cc}
%      0 & 1 \\
%      1 & 0 \\
%    \end{array}
%  \right),\\
%\sigma_y&=\left(
%    \begin{array}{cc}
%      0 & -i \\
%      i & 0 \\
%    \end{array}
%  \right), \quad
%\sigma_z=\left(
%    \begin{array}{cc}
%      1 & 0 \\
%      0 & -1 \\
%    \end{array}
%  \right)
%\end{split}
%\eeq
\beq
X_s=\mathlarger{\prod_{i\in s}\sigma_{x_i}}, \ \ \ Z_p = \mathlarger{\prod_{i\in p}\sigma_{z_i}},
\eeq
that represents vertices ($X_s$) and plaquette operators ($Z_p$) on the square lattice.
\begin{figure}[!ht]
\centering\includegraphics[width=90mm]{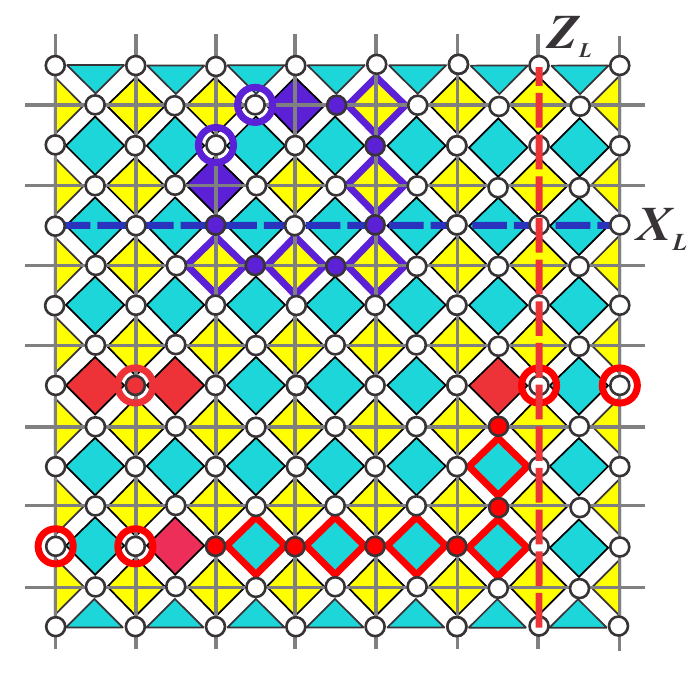}
\caption{\label{Fig:surface_code_illustration} (Color Online) A surface code based on an $8\times8$ lattice with 113 physical qubits on the edges. This code contains 1 logical qubit and has distance $d=L=8$, where $d$ is the distance of the code. The four-body (or three-body) plaquette stabilizer generator ($Z_p$) and vertex stabilizer generator ($X_s$) are indicated as cyan and yellow plaquettes, respectively inside the lattice (or on the boundary). A particular choice of logical operators $X_L$ and $Z_L$ is shown. A number of qubits are affected by $\sigma_x$ (red dots) or $\sigma_z$ (purple dots) errors, leading to excited $Z_p$ operators (or $m$ anyons) and $X_s$ operators (or $e$ anyons). Measuring these operators yields the positions of the excited vertices and plaquettes but reveals no information about the actual physical errors which cause them. A minimum-weight matching error correction procedure applies $\sigma_x$ and $\sigma_z$ to the qubits marked by the larger red and purple circles. While the $\sigma_z$ errors are annihilated properly (up to a trivial loop of multiplication of $Z_p$ operators), the red pair underneath is connected by a topologically non-trivial path across the surface. This introduces a logical error in the state to be protected.}
\end{figure}

\begin{figure}[!ht]
\centering\includegraphics[width=60mm]{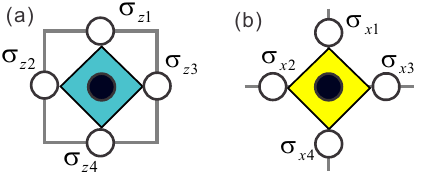}
\caption{\label{Fig:stabilizer_generator} (Color online) Four-body plaquette operator $Z_p$ (a) and vertex operator $X_s$ (b) as stabilizer generators of surface code inside the lattice. The black dot in the center of the plaquettes are syndrome qubits used to do stabilizer measurement.}
\end{figure}

Besides stabilizer generators inside the lattice, there are also ones on the boundaries for each lattice. Typically, for each surface code, there are two kinds of boundaries: $X$ boundaries and $Z$ boundaries. $X$ boundaries comprise three-body $X_s$ operators on the boundary of lattice, while $Z$ boundaries comprise three-body $Z_p$ operators, as in the boundaries shown in Fig.~\ref{Fig:surface_code_illustration}. In general, a lattice with two $X$ boundaries and two $Z$ boundaries has $2L^2-2L+1$ qubits and $2L^2-2L$ stabilizer generators, and encodes 2 degrees of freedom to form a logical qubit. The corresponding logical operators are given by $Z_L=\prod_{k\in l_z}\sigma_{z_k}$ and  $X_L=\prod_{k\in l_x}\sigma_{x_k}$ where $l_z$ and $l_x$ are chains of qubits that support $\sigma_z$ and $\sigma_x$ operators all the way across the lattice (see Fig.~\ref{Fig:surface_code_illustration} for an example).

Not shown in Fig.~\ref{Fig:surface_code_illustration} are additional \emph{syndrome qubits} for each plaquette and vertex, that enable one to check the sign of the associated stabilizer generator, as shown in Fig.~\ref{Fig:stabilizer_generator}. Inside the surface, each syndrome qubit contacts four data qubits and performs four-qubit joint measurement. On the boundaries, each syndrome qubit contacts only three data qubits and performs a three-qubit joint measurement. The corresponding quantum circuit for one stabilizer generator measurement of the $Z_p$ and $X_s$ operators are
\begin{center}$
\Qcircuit @C=1em @R=1em  {
\lstick{\ket{0}} & \targ &\targ & \targ &\targ &\measureD{M_{Z}}\\
\lstick{\text{1}}  & \ctrl{-1}&\qw & \qw &\qw &\qw\\
\lstick{\text{2}} & \qw &\ctrl{-2} &\qw &\qw & \qw \\
\lstick{\text{3}}  & \qw &\qw &\ctrl{-3} &\qw & \qw\\
\lstick{\text{4}}  & \qw &\qw & \qw& \ctrl{-4}&  \qw
}$
\end{center}
and
\begin{center}$
\Qcircuit @C=1em @R=0.6em {
\lstick{\ket{0}}&\gate{H} & \ctrl{1} &\ctrl{2} & \ctrl{3} &\ctrl{4} & \gate{H}&\measureD{M_{Z}}\\
\lstick{\text{1}} &\qw & \targ &\qw & \qw &\qw &\qw &\qw\\
\lstick{\text{2}}&\qw & \qw &\targ &\qw &\qw &\qw  &\qw \\
\lstick{\text{3}} &\qw & \qw &\qw &\targ &\qw &\qw  &\qw\\
\lstick{\text{4}} &\qw & \qw &\qw & \qw& \targ &\qw  &\qw
}$
\end{center}
%\centerline{
%\Qcircuit @C=1em @R=0.6em {
%\lstick{\ket{0}}&\gate{H} & \ctrl{1} &\ctrl{2} & \ctrl{3} &\ctrl{4} & \gate{H}&\measuretab{M_{Z}}\\
%\lstick{\text{1}} &\qw & \targ &\qw & \qw &\qw &\qw &\qw\\
%\lstick{\text{2}}&\qw & \qw &\targ &\qw &\qw &\qw  &\qw \\
%\lstick{\text{3}} &\qw & \qw &\qw &\targ &\qw &\qw  &\qw\\
%\lstick{\text{4}} &\qw & \qw &\qw & \qw& \targ &\qw  &\qw
%}}
respectively. The syndrome qubits are always
initialized to $|0\>$ before the measurement.
%Note that the index of qubits measurement shown in the circuits above should be identical to the index of qubits shown in Fig.~\ref{Fig:stabilizer_generator} as argued in Ref.~\cite{Folwer2012PhysRevA.86.032324}.

If no errors of any kind occur, the code remains in the simultaneous $+1$ eigenstate of all stabilizer generators.
We will restrict our attention to $\sigma_x$ bit-flip errors and $\sigma_z$ phase-flip errors, since very general noise can be tolerated with just the ability to correct these two types of error. If $\sigma_x$ or $\sigma_z$ errors occur, the value of the stabilizer generators anticommute with errors will be flipped to $-1$.
Fig.~\ref{Fig:surface_code_illustration} shows the effect of
$\sigma_x$ and $\sigma_z$ errors on the surface.
If we can reliably detect when stabilizer generators
become negative, it is possible for us to detect the errors
and correct them by finding paths that connect the
flipped syndromes of same kind such that the total number of path edges is minimized.  Note that $\sigma_x$ errors can also be matched to
$X$ boundaries and $\sigma_z$ errors can be matched to $Z$ boundaries
of the surfaces. An example of decoding failure is also shown
in Fig.~\ref{Fig:surface_code_illustration}.

However, the syndrome measurement processes are not necessarily perfect. It is possible for the reported measurement outcome to be wrong because of the imperfect CNOT gates and measurement errors. To get around this problem, one needs to keep track of every time the reported eigenvalue of each stabilizer generator changes. Pairs of flipped syndromes are then connected by paths in both space and time, such that total number of edges connected in space-time used to decode the errors is minimal. Polynomial time minimum weight matching algorithms exists \cite{Edmonds:1965:449}, and hence this can be done efficiently.
\begin{figure}[!ht]
\centering\includegraphics[width=70mm]{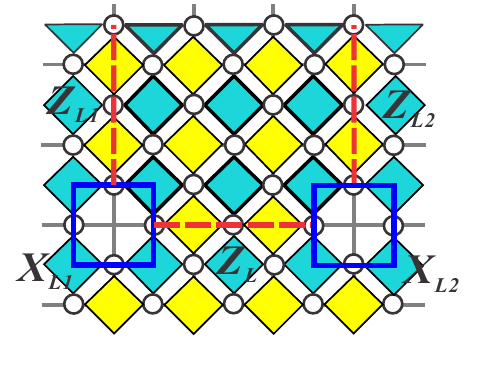}
\caption{\label{Fig:double_X} (Color online). An example of double $X$-cut qubit, with $X_s$ operators turned off. Each $X$-cut hole forms a single $X$-cut logical qubit and there are two kinds of logical operators. $Z_{L_1}$ ($Z_{L_2}$) connects left (right) hole with the $Z$ boundary on the top of lattice, while $X_{L_1}$ ($X_{L_2}$) are any loops encircling left (right) hole. For double $X$-cut qubit, there is a more convenient way to define the logical operators is to set $X_L=X_{L_1}$ and $Z_L=Z_{L_1}Z_{L_2}$. Note that $Z_L$ is equivalent to $Z_{L_1}Z_{L_2}$ up to multiplication by $Z_p$ operators inside the loop and has effect of flip phases for both qubit holes.}
\end{figure}
%In practice, decoding may be delayed after rough $d$ cycles of syndrome measurement as shown in Ref.~\cite{Folwer2012PhysRevA.86.032324}.

For the single logical qubit encoded in surface code, its logical operators $X_L$ and $Z_L$ compose chains of $\sigma_x$ and $\sigma_z$ operators crossing the entire lattice. So this way of encoding is not suitable for larger lattices. Besides, no matter how large the lattice, only a single logical qubit can be stored, since the dimension of code space is fixed. A more flexible approach of encoding is to create holes, or defects, inside the lattice to build extra boundaries on lattice. This can be done by turning off one or more of the $X_s$ and $Z_p$ stabilizer generators inside lattice to form a hole. Here, ``turn off" means that syndrome measurement is no longer performed for this operator (see example in Fig.~\ref{Fig:double_X}) in subsequent computation, so that extra degrees of freedom can be obtained to form a logical qubit. We call the logical qubit obtained this way an $X$-cut ($Z$-cut) \emph{single} logical qubit when an $X_s$ ($Z_p$) stabilizer generator is turned off. For the case of Fig.~\ref{Fig:double_X}, any chain of $\sigma_z$ operators connecting this hole to an $X$-boundary on the top of lattice and any chain of $\sigma_x$ operators encircling the $X$-cut hole can be used to manipulate these extra degrees of freedom. We call any such $\sigma_z$ chain $Z_L$, and any $\sigma_x$ ring $X_L$. If the eigenvalue of $X_s$ is $+1\ (-1)$ before it is turned off, the logical qubit is initialized to the $|+\>\ (|-\>)$ state of $X$-cut single logical qubit, we represent it as $|+^X_{SL}\>\ (|-^X_{SL}\>)$.

One can go further by making logical operators for qubit not rely on operator chains that reach the boundary of lattice. In particular, we can use a pair of $X\ (Z)$-cut holes to form a single logical qubits and manipulate them in a correlated way. This pair of holes are called \emph{double} $X \ (Z)$-cut logical qubit. Fig.~\ref{Fig:double_X} shows an example of a double $X$-cut qubit. Four additional degrees of freedom will be added to the lattice when two $X_s$ are turned off, which can be represented as:
\beq
\begin{split}
&|+^X_{SL}\>_1|+^X_{SL}\>_2, \ \ \ \ |-^X_{SL}\>_1|-^X_{SL}\>_2, \\
&|-^X_{SL}\>_1|+^X_{SL}\>_2, \ \ \ \ |+^X_{SL}\>_1|-^X_{SL}\>_2,
\end{split}
\eeq
where 1 denotes the single $X$-cut qubit on the left and 2 denotes the one on the right. Each single $X$-cut qubit can be manipulated by defining $X_{L_1}$ and $Z_{L_1}$ for the left $X$-cut hole and $X_{L_2}$ and $Z_{L_2}$ for the right $X$-cut hole. The effect of each logical operator pair is:
\beq\label{eq:state_transform}
X_{L_j}|\pm^X_{SL}\>_j=\pm|\pm^X_{SL}\>_j, \ \ \ \ Z_{L_j}|\pm^X_{SL}\>_j=|\mp^X_{SL}\>_j.
\eeq
Manipulating the two qubit holes of a double cut logical qubit in a correlated way can greatly simply the forms of logical operators and increase the number of logical qubits encoded on a single lattice. We can define the $|+\>$ and $|-\>$ states for double $X$-cut logical qubits as:
\beq
|+^X_{DL}\>=|+^X_{SL}\>_1|+^X_{SL}\>_2, \ \ \ |-^X_{DL}\>=|-^X_{SL}\>_1|-^X_{SL}\>_2.
\eeq
A chain of $\sigma_z$ operators connecting the two holes is then used as the definition of the $Z_L$ operator for double $X$-cut qubit, as shown in Fig.~\ref{Fig:double_X}. The $X_L$ operator can be defined as any ring of $\sigma_x$ operators around either hole, as can be seen from Eq.~(\ref{eq:state_transform}). We can then find the $|0\>$ state for the double $X$-cut qubit:
\beq
\begin{split}
|0^X_{DL}\>=&\frac{1}{\sqrt{2}}\big(|+^X_{SL}\>_1|+^X_{SL}\>_2 +|-^X_{SL}\>_1|-^X_{SL}\>_2 \big),\\
|1^X_{DL}\>=&\frac{1}{\sqrt{2}}\big(|+^X_{SL}\>_1|+^X_{SL}\>_2 -|-^X_{SL}\>_1|-^X_{SL}\>_2 \big).
\end{split}
\eeq
Similarly, the $|0\>$ and $|1\>$ states of double $Z$-cut qubits can be defined as:
\beq
|0^Z_{DL}\>=|0^Z_{SL}\>_1|0^Z_{SL}\>_2, \ \ \ |1^Z_{DL}\>=|1^Z_{SL}\>_1|1^Z_{SL}\>_2
\eeq
and the corresponding $|+\>$ and $|-\>$ states of double $Z$-cut qubits are
\beq
\begin{split}
|+^Z_{DL}\>&=\frac{1}{\sqrt{2}}\big(|0^Z_{SL}\>_1|0^Z_{SL}\>_2 +|1^Z_{SL}\>_1|1^Z_{SL}\>_2 \big),\\
|-^Z_{DL}\>&=\frac{1}{\sqrt{2}}\big(|0^Z_{SL}\>_1|0^Z_{SL}\>_2 -|1^Z_{SL}\>_1|1^Z_{SL}\>_2 \big).
\end{split}
\eeq
Note that for the logical qubits described here, the distance of the codes is bounded by 4, no matter how far two holes are separated, because the perimeter of hole created by turning off one stabilizer generator is limited by 4 physical qubits. The error correction ability can be significantly improved if we increase both the size and spacing of the two holes, as this will increase the number of physical qubits involved in $Z_L$ and $X_L$. The details of making larger holes for logical qubits will be discussed in Sec.~\ref{sec:adiabatic_qbit_enlarge_move}.

%\subsection{Logical State}
%$C(\mathcal{S})\backslash\mathcal{S}$

%\subsubsection{Encoded Unitary Operation}
%
%\centerline{
%\Qcircuit @C=1em @R=1em {
%\lstick{\text{Z-cut control in}}&\qw       & \ctrl{1} &\qw      & \rstick{\text{Z-cut target out}}\qw \\
%\lstick{\text{X-cut}\ \ket{0_L}}&\targ     & \targ    &\targ    & \measuretab{M_{Z}}\\
%\lstick{\text{Z-cut}\ \ket{+_L}} &\qw & \qw      &\ctrl{-1}& \rstick{\text{Z-cut target out}}\qw \\
%\lstick{\text{Z-cutP target in}}&\ctrl{-2} & \qw      &\qw      & \measuretab{M_{X}}\\
%&          &          &    \lstick{\text{(a)}}      & \\
%\lstick{\text{X-cut}\ \ket{0_L}}&\qw       & \qw    &\targ      & \rstick{\text{X-cut target out}}\qw \\
%\lstick{\text{X-cut control in}}&\ctrl{-1} &\targ     & \qw    & \measuretab{M_{Z}}\\
%\lstick{\text{Z-cut}\ \ket{+_L}}&\qw &\ctrl{1}      &\ctrl{-2}& \measuretab{M_{X}}\qw \\
%\lstick{\text{X-cut target in}}&\qw       & \targ      &\qw      & \rstick{\text{X-cut target out}}\qw\\
%&          &            & \lstick{\text{(b)}}&
%}
%}
%
%\vspace{3mm}

%\subsubsection{Error as anyons}
%
%

\subsection{Holonomic Quantum Computation}

Consider a Hamiltonian family $\{H_\lambda\}$ on an $N-$dimensional Hilbert space. The point $\lambda$, parametrizing the Hamiltonian, is an element of a manifold $\mathcal {M}$ called the control manifold, and the local coordinates of $\lambda$ are denoted by $\lambda^i\ (1\leq i \leq \textrm{dim}\mathcal{M})$. Assume there are only a fixed number of eigenvalues $\{\varepsilon_k(\lambda)\}$ and suppose the $n$th eigenvalue $\varepsilon_n(\lambda)$ is $K_n$-fold degenerate for any $\lambda$. The degenerate subspace at $\lambda$ is denoted by $\mathcal{H}_n(\lambda)$. The orthonormal basis vectors of $\mathcal{H}_n(\lambda)$ are denoted by $\{|\phi_\alpha^n;\lambda\>\}$, satisfying
\beq
H_\lambda|\phi_\alpha^n;\lambda\>=\varepsilon_n(\lambda)|\phi_\alpha^n;\lambda\>,
\eeq
and
\beq
\<\phi_\alpha^n;\lambda|\phi_\beta^m;\lambda\>=\delta_{nm}\delta_{\alpha\beta}.
\eeq
Assume the parameter $\lambda$ is changed adiabatically, which means that
\beq\label{eq:adiabatic_approx}
\left(\varepsilon_n(\lambda(t))-\varepsilon_{n^\prime}(\lambda(t))\right)T\gg1
\eeq
is satisfied for $n\neq n^\prime$ during $0\leq t\leq T$). Suppose the initial state at $t=0$ is an eigenstate $|\psi^n(0)\>=|\phi_\alpha^n;\lambda(0)\>$. The Schr\"{o}dinger equation is
\beq \label{eq:shrodinger}
i\frac{\text{d}}{\text{d}t}|\psi^n(t)\>=H(\lambda(t))|\psi^n(t)\>,
\eeq
whose solution will have the form
\beq \label{eq:solution}
|\psi^n(t)\>=\sum_{\beta=1}^{K_n}|\phi_\beta^n;\lambda(t)\>U_{\beta\alpha}(t).
\eeq
where we have used the adiabatic approximation from Eq.~(\ref{eq:adiabatic_approx}). Substituting Eq.~(\ref{eq:solution}) into Eq.~(\ref{eq:shrodinger}), one finds that $U_{\beta\alpha}$ satisfies
\beq
\begin{split}
\dot{U}_{\beta\alpha}(t)=&-i\varepsilon_n(\lambda(t))U_{\beta\alpha}(t)\\
&-\sum_{\mu}\<\phi_\beta^n;\lambda(t)|\frac{\text{d}}{\text{d}t}|\phi_\mu^n;\lambda(t)\>U_{\mu\alpha}(t).
\end{split}
\eeq
The solution can be expressed as
\beq
\begin{split}
U(t)=&\exp\left(-i\int_0^t\varepsilon_n(\lambda(s))\text{d}s\right)\times\\
&\mathcal{T}\exp\left(-\int_{0}^{t} A^n(\tau)\text{d}\tau\right),
\end{split}
\eeq
where $\mathcal {T}$ is the time-ordering operator and
\beq\label{eq:WZ_connection}
A^n_{\beta\alpha}(t)=\<\phi^n_{\beta};\lambda(t)|\frac{\text{d}}{\text{d}t}|\phi^n_\alpha;\lambda(t)\>
\eeq
is the Wilczek-Zee (WZ) connection~\cite{Wilczek:1984:2111}. Define the connection
\beq
\mathcal{A}^n_{i,\beta\alpha}(t)=\<\phi^n_{\beta};\lambda(t)|\frac{\partial}{\partial \lambda^i}|\phi^n_\alpha;\lambda(t)\>,
\eeq
through which $U(t)$ can be expressed as
\beq\label{eq:general_adiabatic_evolution}
\begin{split}
U(t)=&\exp\left(-i\int_0^t\varepsilon_n(\lambda(s))\text{d}s\right)\times\\
&\mathcal{P}\exp\left(-\int_{\lambda(0)}^{\lambda(t)}\sum_i\mathcal{A}^n_i\text{d}\lambda^i\right),
\end{split}
\eeq
where $\mathcal{P}$ is the path-ordering operator. Eq.~(\ref{eq:general_adiabatic_evolution}) is a general description of both open loop and closed loop adiabatic state evolution. Both are useful for our scheme as will be shown in Sec.~\ref{sec:sketch_scheme} and Sec.~\ref{sec:HQC_surface}. In particular, suppose the path $\lambda(t)$ is a loop $\lambda$ in $\mathcal{M}$ such that $\lambda(0)=\lambda(T)=\lambda_0$ (closed loop). Then after transporting through $\lambda$, states are transformed to
\beq
|\psi^n(T)\>=\sum_{\beta=1}^{K_n}|\psi_\beta^n(0)\>U_{\beta\alpha}(T).
\eeq
The unitary matrix
\beq
\Gamma_\lambda
=\mathcal{P}\exp\left(-\oint_\lambda\sum_i\mathcal{A}^n_i\text{d}\lambda^i\right)
\eeq
is called the holonomy associated with the loop $\lambda(t)$. $\Gamma_\lambda$ is a purely geometric object, and is independent of the parametrization of the path. Note that for a given $\Gamma_\lambda$, there exist infinitely many paths $\lambda$. One of the main objects of the paper to find the proper path in $\mathcal{M}$ that will give us the desired state transformation in the code space of the surface code under adiabatic transformation of the stabilizer Hamiltonian. A geometric formulation of the holonomic problem, which gives an alternative description as shown in Refs.~\cite{Tanimura2004199, tanimura:022101}, is also given in Appendix~\ref{sec:geometric_holonomic}, which is useful in improving the results of the next section.

%Given a path $\lambda$, to find the holonomy is easy. However, the inverse problem---given a holonomy, to find the the proper path $\lambda$---is in general not trivial at all. In the rest of the paper, we will discuss how to find a proper path $\lambda$ to realize a certain holonomy in the code space, and thus perform an encoded quantum gate operation.

\section{Sketch of the Scheme}\label{sec:sketch_scheme}
In this scheme, we always regard all physical qubits on the lattice as a single big stabilizer code. We assume that the qubits independently and weakly interact with a thermal bath in the Markovian approximation. The corresponding thermal errors are local and low-weight during a certain period of evolution. Those low-weight thermal excitations will cause transitions from the ground space to excited spaces. Their rate should decrease as $\delta_{\text{thermal}}\sim \exp\left(-c\beta\Delta_{\min}\right)$, where $\Delta_{\text{min}}$ is the minimum spectral gap of the system, $\beta$ is the inverse of temperature, and $c$ is a constant depending on the coupling strength between system and thermal bath~\cite{alicki2009thermalizationKitaev}. This is true even when the Hamiltonian is not static and changes slowly, so long as the system is weakly coupled to the thermal bath ~\cite{TameemNJP:adiabaticMarkovianME}. The goal is to do the whole quantum computation fault-tolerantly, while the code space is protected by an energy gap of the stabilizer Hamiltonian that exponentially suppresses errors at low temperature throughout the information processing procedure.

To analyze the error performance of the architecture, we must first define a fault-tolerant procedure:
\begin{mydef}\label{def:fault-tolerant}
A procedure is \emph{fault-tolerant} if it has the property that if only one component (or more generally, a small number of components) in the procedure fails, the errors produced by this failure are not transformed into an uncorrectable error by the procedure, before error correction is applied.
\end{mydef}
\noindent With this definition, the fault-tolerance of a procedure can be regarded as a property of the procedure itself regardless of the error model of the system.
Before we go deeper, we must impose some requirements to follow in the rest of the paper:
\begin{enumerate}
  \item Procedures like logical state preparation, logical state measurement, encoded gate operations, state injection and state distillation should be done when the system Hamiltonian is ``turned on", so that a constant energy gap protects the information and the error rate for each procedure is low.
  \item All procedures should be done fault-tolerantly according to Def.~\ref{def:fault-tolerant}, whether adiabatic or not.
  \item Syndrome measurements and error correction should be done before uncorrectable errors happen.
  \item Syndrome measurements and error correction should be done as seldom as possible, since they are in general not compatible with the system Hamiltonian and we must turn off the Hamiltonian before doing them. Besides, the syndrome measurement procedure itself is quite expensive. The frequency of error correction is expected to be low if all procedures are gap-protected.
  \item A threshold theorem should exist, in the sense that if the error rate is below the threshold, the computation can be made arbitrarily long by suitably increasing the lattice size.
  \item It is possible to measure $\sigma_x$ and $\sigma_z$ of single physical qubits in certain circumstances even when the Hamiltonian is turned on.
  \item Maximum weight of the Hamiltonian terms should be low, and the Hamiltonian should be geometrically local.
  \item All procedures should be done in a single lattice.
\end{enumerate}

%Most of the time, the system Hamiltonian exists so that we can not do computation in usual measurement based way in Ref.~\cite{,Fowler:2009:052312,Folwer2012PhysRevA.86.032324}. And we hope to do syndrome measurement as less as possible, since in general, they are not compatible with system Hamiltonian and we need to turn off the Hamiltonian, and syndrome measurement procedure itself is quite expensive. However, as we well see later, we do allow to do single qubit measurement.
%Principal: during the information processing procedure, we hope to use measurement on system as less as possible. We also hope that the terms in Hamiltonian is low weight and geometrically local.
Requirements $1-4$ are crucial to our main objective of reducing the physical resources and 5 guarantees that arbitrarily large-scale computation can be done.
Requirement 6 is physically reasonable, and we will see its importance in Sec.~\ref{sec:HQC_surface}. Requirement 7 comes from the fact that in real experiments, high weight and nonlocal Hamiltonians are difficult or impossible. Requirement 8 is technical rather than fundamental, since it simplifies the the computation architecture.

\subsection{Adiabatic processes}
In most cases, adiabatic processes can be used to simultaneously fulfill most of the requirements above. For the purpose of encoding and measuring logical qubits, we will show that these can be done by open-loop adiabatic processes (for logical measurement, we also need qubit measurement), while the logical CNOT can be done by a closed-loop adiabatic processes to get a holonomy on the code space. Both such processes can be described by Eq.~(\ref{eq:general_adiabatic_evolution}). In this and the following sections, we will focus on a special kind of adiabatic evolution that turns out to be particularly useful. In addition, we will discuss how it can be used to analyze propagation of potential errors and parallelism of the processes.

Assume the total number of qubits on the lattice is $n=2L^2-2L+1$, so the dimension of the Hilbert space is $N=2^n$. The number of logical qubits in our scheme may change over time, since we can create defects on the lattice to create logical qubits.  However, we assume that when an adiabatic process is applied, the dimension of code space is fixed. This can be realized by isospectral deformation of the Hamiltonian.
Denote the number of logical qubits encoded in the ground space by $k$. Assume that at time $t_0$, the initial Hamiltonian can be written as
\beq
H(t_0) = -\sum_{j=1}^{n-k} J S_j,
\eeq
where the $\{S_j\}$ are a set of stabilizer generators of the surface code at time $t_0$ and that $\<S_j\>$ forms the stabilizer group $\mathcal{S}$. Consider the following way to adiabatically deform the Hamiltonian isospectrally:
\beq\label{eq:Hamiltonian_change}
\begin{split}
H(t)=&-\sum_{j=1}^{n-k}J S_j(t)\\
=&-\sum_{j=1}^{n-k}J U(t,t_0)S_j(t_0)U^\dag(t,t_0),
\end{split}
\eeq
%
% Assuming at each time $t_q$,  Hamiltonian can be represented as:
%\beq
%H(t_q) = -\sum_{j=1}^{n-k}J S_j(t_q),
%\eeq
with $S_j(t)=U(t,t_0)S_jU^\dag(t,t_0)$ and $[S_i(t), S_j(t)]=0$ for all $i$, $j$. The $\{S_j(t)\}$ can be viewed as a set of generators of
an Abelian group, like the stabilizer group.
%with $\{S_j(t_q)\in G_n\}$ the stabilizer generators of the codes defined at time $t_q$ with $k_q$ logical operators.
%for $t\in[t_{q},t_{q+1}]$
%\beq
%H(t)=\sum_{j=1}^{n-k_q} C_j(t)S_j(t)
%\eeq
%where $C_j(t)\in[-1,0]$ and $\{S_j(t)\}$ can be viewed as generators of an Abelian group, such as stabilizer group defined at time $t$.
The Hamiltonian also has a spectral decomposition:
\beq
H(t) = \sum_{\textbf{s}}\varepsilon_{\textbf{s}}P_{\textbf{s}}(t).
\eeq
Here, the $\{P_{\textbf{s}}(t)\}$ are projectors onto the simultaneous eigenspaces of all the $S_j(t)$, with eigenvalues:
\beq
\varepsilon_{\textbf{s}} = -J\sum_j s_j,
\eeq
where the labels $s_j=\pm 1$ form a vector:
\beq
\textbf{s}=\{s_1,s_2,\ldots s_{n-k}\}.
\eeq
The ground space evolves with the system Hamiltonian. This defines a time-dependent code space $\mathcal{C}_t$. Let $P_0(t)=U(t,t_0)P_0(t_0)U^\dag(t,t_0)$ be the projector onto the ground space of $H(t)$, such that $s_j=1$ for all $j$. We emphasize that $U(t+\tau,t)$ should be chosen such that
\beq\label{eq:evolution_condition}
\left[\frac{\partial}{\partial \tau} U(t+\tau,t)|_{\tau=0}, \ P_{\textbf{s}}(t)\right]\neq 0\ \ \ \ \text{for all \textbf{s}},
\eeq
for any time $t$, so that the deformation procedure is nontrivial for all eigenspaces. In other word, $U(t+\tau,t)$ should not belong to the isotropy group of $P_{\textbf{s}}(t)$ for small values of $\tau$.

The adiabatic condition must hold for each eigenspace $P_\textbf{s}$, so that each eigenspace undergoes nontrivial evolution under the adiabatic process, in case an error excites the system to $P_\textbf{s}$ during the process. The standard adiabatic condition~\cite{Messiah:1965:North} for any eigenspace $\{P_{\textbf{s}_\alpha}\}$ can be reformulated as:
\beq\label{eq:adiabatic_condition_general}
\frac{\parallel P_{\textbf{s}_\alpha}(t)\frac{\partial}{\partial t}H(t)P_{\textbf{s}_\beta}(t)\parallel_1}{K\left(\varepsilon_{\textbf{s}_\alpha}(t)
-\varepsilon_{\textbf{s}_\beta}(t)\right)^2}\approx  0, \ \text{for any}\ \alpha\neq\beta.
\eeq
Here, $K$ is the degeneracy of each $P_{\textbf{s}}$. This must hold for all $t\in[t_0,t_p]$, where $\parallel\cdot\parallel_1$ is the trace norm $\left(\parallel A\parallel_1=\Tr\sqrt{A^\dag A}\right)$. It is very likely that for a Hamiltonian of the form Eq.~(\ref{eq:Hamiltonian_change}), several $P_{\textbf{s}}(t)'$s will share the same eigenenergy, so that the adiabatic condition cannot be directly satisfied. Fortunately, for the surface code, we will show later that there is a natural way to cope with this problem, so that each $P_\textbf{s}(t)$ can satisfy the adiabatic condition during the adiabatic code deformation.

As shown in Ref.~\cite{Yi-Cong_PhysRevA.89.032317}, a closed loop adiabatic logical gate operation can be built from a fault-tolerant circuit of the corresponding stabilizer code. However, for the surface code, we in general don't know the exact fault-tolerant circuit for encoded gate operations. Moreover, we wish to do encoding and logical state measurement with gap protection, so the result in Ref.~\cite{Yi-Cong_PhysRevA.89.032317} cannot be applied here directly. Instead, in this paper, we consider a special kind of quantum circuit $\mathcal{G}$ composed of a sequence of gate operations $\{g_1, g_2 \ldots g_p \}$ giving the unitary operation $\Omega_p=\prod_{l=1}^p g_l$. Here, $g_l = \exp\left(i\frac{\pi}{4}Q_l\right)$ for some Hermitian operator $Q_q\in G_n$, where $G_n$ is the Pauli group acting on $n$ qubits. For simplicity, when we talk about a ``circuit" in the rest of paper, we means the circuit of this type. We divide the information processing time $[t_0, t_p]$ into $p$ small steps and represent the $q$th time segment as $[t_{q-1},t_{q}]$. Now, set the unitary operator
\beq
U_q(t,t_{q-1})=\exp\big(i f_q(t) Q_q\big),
\eeq
for $t\in[t_{q-1},t_q]$ and let
$f_q:[t_{q-1}, t_q]\rightarrow[0,\pi/4]$ be a monotonic smooth function with boundary conditions $f_q(t_{q-1})=0$ and $f_q(t_q)=\pi/4$.
For each time segment $[t_{q-1},t_q]$, we adiabatically deform the Hamiltonian:
\beq\label{eq:Hamiltonian_change_each_step}
H(t,t_{q-1}) = U_q(t,t_{q-1})H(t_{q-1})U_q^\dag(t,t_{q-1}), \ \ \ t\in[t_{q-1},t_q]
\eeq
and assume $[Q_q,H(t_{q-1})]\neq 0$ so that Eq.~(\ref{eq:evolution_condition}) is satisfied. A state in the ground space will evolve as described by the following lemma:
\begin{mylemma}\label{lemma:state_evolution}
(State Evolution) Consider a circuit composed of gates $\{g_q\}$ and an initial state $|\psi(t_0)\>\in \mathcal{C}(t_0)$, with $H(t_0)=-\sum_j J S_j$. We apply a sequence of Hamiltonian deformations as in Eq.~(\ref{eq:Hamiltonian_change_each_step}), for $1\leq q\leq p$. Then,  under the adiabatic approximation, the final state will be:
\beq
\begin{split}
|\psi(t_p)\>&=e^{-i\varepsilon_0(t_p-t_0)}\left(\prod_{l=1}^{p}g_{l}\right)|\psi(t_0)\>\\ &=e^{-i\varepsilon_0(t_p-t_0)}\Omega_p |\psi(t_0)\>.
\end{split}
\eeq
\end{mylemma}
\begin{proof}
See Appendix~\ref{sec:proof_lemma_state_evolution}.
\end{proof}
%\begin{col}
%Final Hamiltonian is $H(t_p) = \Omega_g H(t_0)\Omega_g^\dag$, which does not necessarily equal to the initial one.
%\end{col}

%\begin{thm}
%(Adiabatic condition) If $U(t)$ can be represented as $\exp\left(-if(t)Q\right)$ for $Q\in G_n$, then $h(t,0)=h(0,0)$. $P(0)=\prod_i\frac{I+g_i}{2}$, $\forall i$, $g_i\in G_n$, and $[U(t),P(0)]\neq0,\forall t$
%\end{thm}
%\begin{proof}
%From Eq.~(\ref{eq:horizontal_equation_group}), we have $v_0\dot{h}(t,0)v_0^\dag = -iP(0)QP(0)$. Since $P(0)$ can be written as $P(0)=\prod_i\frac{I+g_i}{2}$ for  $\forall i$, $g_i\in G_n$, $[U,P(0)]\neq0$, so there exists at least one $g_i$ such that $\{g_i,Q\}=0$. So $-iP(0)QP(0)=0$ and $h(t,0)=h(0,0)$.
%\end{proof}
In the case of a many-body system like the surface code, it is difficult to follow the change of the state in code space since it is hard to represent the state. One normally uses the stabilizer formalism (Heisenberg picture) to track the change of the logical $Z_L$ and $X_L$ operators during the process. The following theorem is a direct consequence of Lemma~\ref{lemma:state_evolution}:
\begin{mytheorem}\label{them:them1}
Suppose the initial state $|\psi(t_0)\>$ is in the code space of a stabilizer code with generators $\{S_j\}$ and logical operators $\{X^i_L, Z^i_L\}$, and that $H(t_0)=-\sum_j J S_j$.  Under the adiabatic Hamiltonian deformation described in Eq.~(\ref{eq:Hamiltonian_change_each_step}) for $1\leq q \leq p$, the logical operators will map to $X^i_L\rightarrow \Omega_pX^i_L\Omega^\dag_p, Z^i_L\rightarrow\Omega_pZ^i_L\Omega^\dag_p$, and the system Hamiltonian will become $H(t_p)=-\sum_j J S_j^\prime=-\sum_jJ\Omega_p S_j\Omega^\dag_p$.
\end{mytheorem}
If the process is cyclic for the ground space, which means $\Omega_p P_0(t_0)\Omega_p^\dag=P_0(t_0)$, then $\Omega_p$ can be viewed as an encoded gate operation, and we have following conclusion:
\begin{mycorollary}
If $\Omega_p\in N(\mathcal{S})\backslash\mathcal{S}$, where $N(\mathcal{S})$ is the normalizer of $\mathcal{S}$ in $\text{U}(N)$, then $\Omega_p$ is a closed-loop holonomic operation under the adiabatic process.
\end{mycorollary}
\begin{remark}
These results build a relationship between the special kind of circuits $\mathcal{G}$ we are interested in and the corresponding adiabatic process. If we can find a circuit in $\mathcal{G}$ giving a particular unitary, then
we can translate it to an adiabatic process.
%If the process is cyclic for $P_0(t)$, which means $\Omega_p P_0(t) \Omega_p^\dag = P_0(t)$, then  $\Omega_p\in N(\mathcal{S}) \backslash {S}$, where  $N(\mathcal{S})$ is the normalizer of $\mathcal{S}$ in $G_n$, then we obtain an logical operation on the code space.
However, in general, the weight of the Hamiltonian terms changes with time, and it is quite possible that during the adiabatic process, the Hamiltonian terms will become both nonlocal and high weight. Fortunately, as we will see, in the case of surface codes this can be avoided.
\end{remark}

\subsection{Error propagation}
Although in the process described by Eq.~(\ref{eq:Hamiltonian_change}), the ground space is protected by a constant energy gap $2J$, the lifetime is about $e^{2c\beta J}$ in the presence of a thermal bath. This lifetime doesn't grow with the lattice size $L$, so the the thermal gap does not guarantee fault-tolerance. We still need to do active error correction to make the computation time arbitrarily long. We must analyze how an error caused by thermal excitation will propagate during the adiabatic process to choose the proper circuit from $\mathcal{G}$ and design the subsequent error correction procedure.

Without loss of generality, we assume that an error $E_{t_q}$ happens at time $t_q\ (q\leq l)$. Since any error operator $E_{t_q}$ on an $n$-qubit system can be decomposed into a sum of Pauli operators  $E_{t_q}=\sum_\alpha c_\alpha F_\alpha$, it is sufficient to analyze Pauli errors. We have following lemma:
\begin{mylemma}\label{lemma:error_propagation}
(Error Propagation) If an error $E_{t_q}=\sum_\alpha c_\alpha F^q_\alpha$  $(F^q_\alpha \in G_n)$ happens at time $t_q$ in the procedure described by Eq.~(\ref{eq:Hamiltonian_change_each_step}), and there is an \emph{odd} number of stabilizer generators $S_j(t_r)$ such that $[Q_{r},S_{j}(t_{r-1})]\neq0$ for all times $1 \leq r \leq p$, then
\beq
|\psi(t_p)\>= \mathlarger{\sum}_\alpha c_\alpha e^{-i\varepsilon_{\textbf{s}_\alpha}(t_p-t_q)} F^{pq}_\alpha\left(\prod_{l=1}^{p}g_{l}\right)|\psi(t_0)\>
\eeq
where $F_\alpha^{pq}=\mathscr{U}^{pq} F_\alpha^q \left(\mathscr{U}^{pq}\right)^\dag$ with $\mathscr{U}^{pq}=\prod_{l=q+1}^{p}g_{l}$.
\end{mylemma}
\begin{proof}
See
Appendix~\ref{sec:proof_error_propagation}.
\end{proof}
Lemma~\ref{lemma:error_propagation} gives the condition that the error will just propagate to some other error under the expected unitary evolution.  The condition that at each step $r$ the number of $S_j(t_{r-1})$ such that $[Q_r, S_j(t_{r-1})]\neq 0$ should be $\emph{odd}$ is crucial. In general, an error will excite the ground space to another eigenspace $P_{\textbf{s}_\alpha}$, which will usually share the same energy with some other eigenspaces, so that the adiabatic condition will not hold. This condition guarantees that even when this is the case, the degenerate eigenspaces will still satisfy the adiabatic condition Eq.~(\ref{eq:adiabatic_condition_general}) and adiabatic evolution will not fail.

Also, note that if $\Omega_p$ is a logical gate operator, although for the ground space $P_0$ the process is a cyclic evolution, e.g, $\Omega_p P_0(t_0)\Omega_p^\dag = P_0(t_0)$, this is not true for the other eigenspaces. In general, $\Omega_p P_{\textbf{s}}(t_0)\Omega_p^\dag \neq P_{\textbf{s}}(t_0)$ for $\textbf{s}\neq 0$. This means that after an error excites the ground space $P_0$ to $P_{\textbf{s}}$, the adiabatic process becomes open loop for $P_{\textbf{s}}$.

\subsection{Parallelism of adiabatic operation}
The method described in the previous sections is basically a serial operation, meaning that we need to adiabatically deform the Hamiltonian according to the gates in the circuit $\mathcal{G}$ step by step. However, for a large scale QC on a lattice (not only the surface code), we expect that many operations can be done in parallel, so that operations which commute with each other can be done simultaneously. Here we give the condition for those operations to parallelize.
\begin{mylemma}\label{lemma:parrallelism}
(Parallelism) Suppose that at time $t_q$, $|\psi(t_q)\>$ is in the ground space $\mathcal{C}(t_q)$. Define $\mathscr{C}_{Q_r}=\{j\ |\ 1\leq j\leq n-k , \{S_j(t_q),Q_r\}=0\}$. Suppose the set of operators $\mathscr{P}_q=\{Q_r|q+1 \leq r\leq q+M\}$ satisfies the following conditions:
\begin{enumerate}
  \item $[Q_r,Q_m]=0$, for any $Q_r,Q_m \in \mathscr{P}_q$,
  \item $\mathscr{C}_{Q_{r}} \bigcap \mathscr{C}_{Q_{m}}=\emptyset$ for any $Q_r,Q_m \in \mathscr{P}_q$,
  \item $|\mathscr{C}_{Q_r}|$ is odd for all $Q_r\in \mathscr{P}_q$.
\end{enumerate}
and set $U_{q+1}(t,t_q)=\prod_{r=q+1}^{q+M}\exp\left(if(t)Q_r\right)$ with $f(t)=f_{q+1}(t)$ for $t\in[t_q,t_{q+1}]$. Assume the Hamiltonian changes adiabatically as $H(t)=U_{q+1}(t,t_q)H(t_q)U_{q+1}^\dag(t,t_q)$. Then we have:
\begin{enumerate}
  \item The state at time $t_{q+1}$ will be:
  \beq
  |\psi(t_{q+1})\>=e^{-i\varepsilon_0(t_{q+1}-t_{q})}\left(\mathlarger{\prod}_{l=q+1}^{q+M}g_{l}\right) |\psi(t_q)\>.
  \eeq
  \item If an error $E_{t_q}=\sum_\alpha c_\alpha F^q_\alpha$ $(F^q_\alpha\in G_n)$ occurs at time $t_q$, then the state at time $t_{q+1}$ will be:
      \beq
      |\psi(t_{q+1})\>=\mathlarger{\sum}_\alpha c_\alpha e^{-i\varepsilon_{\textbf{s}_\alpha}(t_{q+1}-t_q)}F_\alpha^{q+1,q}
      \left(\mathlarger{\prod}_{l=q+1}^{q+M}g_{l}\right) |\psi(t_q)\>,
      \eeq
      where $F_\alpha^{q+1,q}=\mathscr{U}^{q+1,q} F_\alpha^q \left(\mathscr{U}^{q+1,q}\right)^\dag$ with $\mathscr{U}^{q+1,q}=\prod_{r=q+1}^{q+M}g_{r}$.
\end{enumerate}
\end{mylemma}
\begin{proof}
See Appendix~\ref{sec:proof_parallesim}.
\end{proof}
Lemma~\ref{lemma:parrallelism} suggests that it is possible to do $M$ steps of the adiabatic transformation described in Lemma~\ref{lemma:state_evolution} in one step, and gives the conditions for the adiabatic evolution to still be valid when errors occur. This property is extremely important. Since we need to apply our scheme to surface codes of large size, operations applied simultaneously on different parts of the surface can greatly improve the efficiency of computation.

%\vspace{2mm}

\section{HQC in Surface Codes}\label{sec:HQC_surface}
We are ready to show how to do QC fault-tolerantly by adiabatically deforming the stabilizer Hamiltonian of the surface code. As mentioned in the previous section, our goal is that all the procedures, including state preparation, ancilla preparation, logical gate operations and logical state measurements, be implemented fault-tolerantly with constant energy gap protection. In the next few subsections, we discuss how to construct these procedures, and discuss error propagation and error detection in detail.

% We will ignore the error accounting for the discussion at first to keep it simple, but we return to this important issue when necessary.
State measurement is a special case worth more discussion here. At the end in the computation, when we want to read all of the data in the logical qubits, we can just turn off the Hamiltonian and measure everything. However, during the computation, when the stabilizer Hamiltonian exists, we still may need to measure logical qubits from time to time, so that actions conditioned on those classical measurement outcomes of logical qubit can be applied. We must put some restrictions on the kinds of measurements we can do that are compatible with the existence of the stabilizer Hamiltonian. The first requirement is that the observable $\mathcal{O}$ we want to measure should commute with the Hamiltonian:
\beq\label{eq:measurement_req}
\left[H, \mathcal{O}  \right]=0.
\eeq
This requirement guarantees that if a state encoding quantum information is in one of the eigenspaces $P_{\textbf{s}}$ before the measurement, then after the projective measurement, the state will still be in $P_{\textbf{s}}$. If Eq.~(\ref{eq:measurement_req}) is not satisfied, the measurement will lead to excitations out of the eigenspace. The second requirement is that the observable should be geometrically local, so that the measurement procedure will not introduce non-local interactions. Note that when the Hamiltonian is turned on, we do not do $X_s$ or $Z_p$ stabilizer measurements even though they commute with the system Hamiltonian and are local. The reason for this is that to projectively measure these many-body observables, we would need to introduce CNOT gates and syndrome qubits, which are not compatible with the system Hamiltonian. So in our scheme, syndrome measurements are always done when the system Hamiltonian is turned off. However, as stated in requirement 6 in the previous section, we do allow single physical qubit measurements as long as they commute with the system Hamiltonian.

Errors can happen during the single qubit measurement process. There are two kinds of measurement errors. The first kind is that, instead of an ideal measurement, some quantum process occurs during the measurement process which is equivalent to one of the following circuits:
\begin{center}$
\Qcircuit @C=1em @R=1em {
\lstick{\ket{\psi}} & \gate{\sigma_x} &\measureD{M_{Z}} \gategroup{1}{2}{1}{3}{.7em}{--} & & & &\lstick{\ket{\psi}} & \gate{\sigma_z} &\measureD{M_{X}}\gategroup{1}{8}{1}{9}{.7em}{--}
}$
\end{center}
for $\sigma_z$ measurement and $\sigma_x$ measurement, respectively. The second kind of error can be regarded as a software error: even though the measurement is perfect, some classical noise corrupts the measurement result and we get the wrong outcome. This can be modeled by the circuits
\begin{center}$
\Qcircuit @C=1em @R=1em {
\lstick{\ket{\psi}} &\measureD{M_{Z}}& \cw &\push{X}   &  & & &\lstick{\ket{\psi}} &\measureD{M_{X}}& \cw &\push{X}\gategroup{1}{2}{1}{4}{.7em}{--}\gategroup{1}{9}{1}{11}{.7em}{--}
}$
\end{center}
\vspace{3mm}
In this paper, we assume we can completely overcome errors of the second kind, and focus only on the first kind of errors.

Finally, note that in the process of computation, we are frequently required to do logical $X_L$ and logical $Z_L$ gates. We do not  necessarily implement these gates physically; rather, we can simply keep a record of it, and apply $X_L$ and $Z_L$ to that logical qubit in ``software", as described in Secs. IX and XVI.A of Ref.~\cite{Folwer2012PhysRevA.86.032324}.

\subsection{Creation of $|+\> \ (|0\>)$ state for $X\ (Z)$-cut double qubit }
Before computation begins, we assume the system is already prepared with the eigenvalues of all stabilizer generators equal to $+1$. This can be done by several methods. One of them is preparing all qubits in the $|0\>$ state and then measuring all $X_s$ stabilizer generators and resetting their eigenvalues to $+1$. After that, we turn on the stabilizer Hamiltonian:
\beq
H(t_0)=-J\mathlarger{\sum}_{i}X_{s_i}-J\mathlarger{\sum}_{j}Z_{p_j}.
\eeq
\begin{figure}[!ht]
\centering\includegraphics[width=85mm]{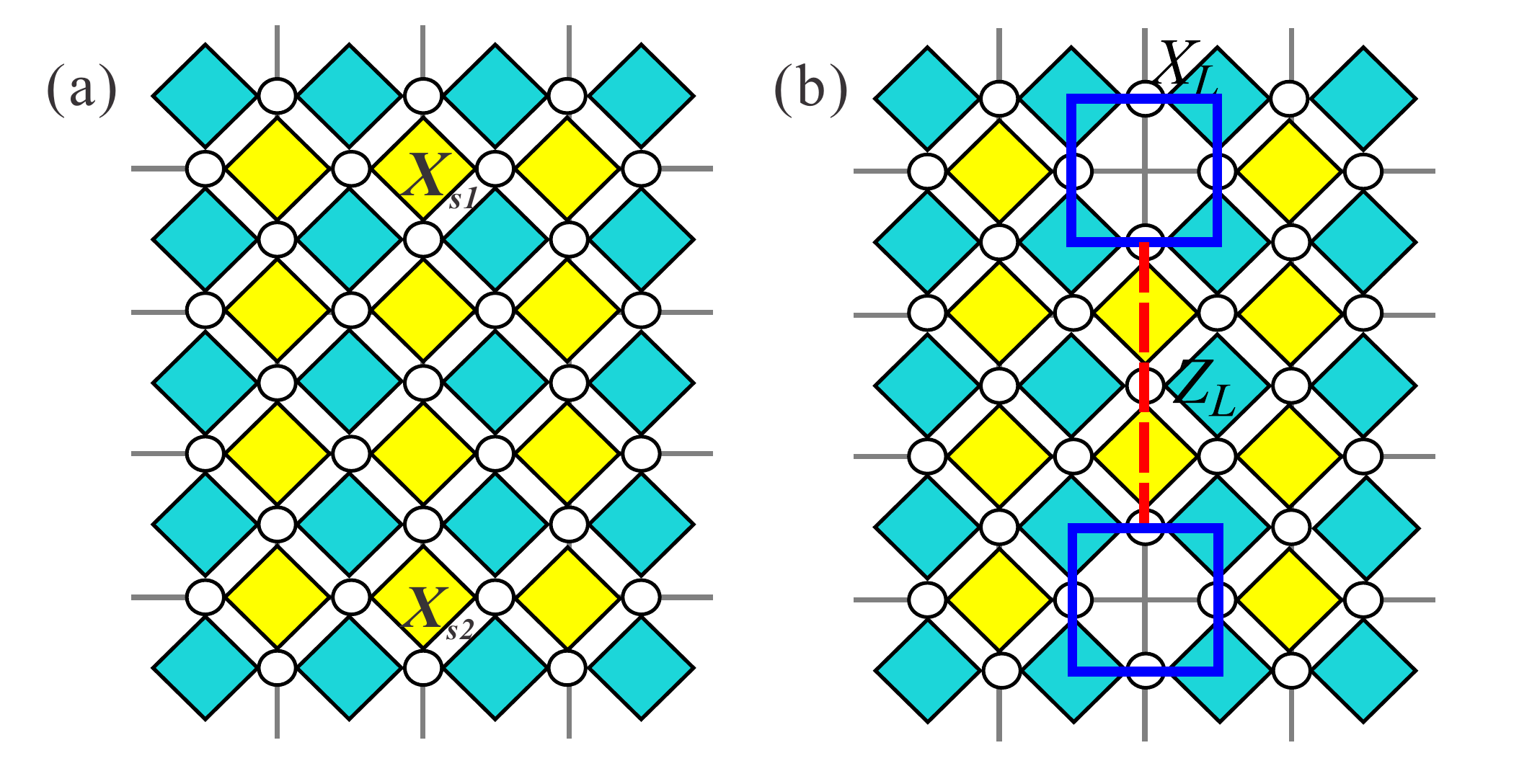}
\caption{\label{fig:easy_initialization} (Color online) Creation of $|+\>$ for $X$-cut double qubit.  System Hamiltonians before and after are shown in (a) and (b) respectively. Colored squares indicate that the corresponding $X_s$ (yellow) and $Z_p$ (cyan) stabilizer generators are turned on. }
\end{figure}

There are two types of initialization procedures. The first is the creation of a $|+\>\ (|0\>)$ state for a $X\ (Z)$-cut and second is the creation a $|+\>\ (|0\>)$ state for $Z\ (X)$-cut qubit. Here, we give an example of preparing a $|+\>$ state for $X$-cut double logical qubit; the $Z$-cut case is similar. We will see that if we can do the first type of preparation fault-tolerantly, we can do the second type fault-tolerantly as well, as will be shown in Sec.~\ref{sec:difficult_initialization}. Suppose initially the state of the system is shown in panel (a) of Fig.~\ref{fig:easy_initialization} with a fully stabilized array, and the stabilizer Hamiltonian terms in this area are all turned on. Turning off the $X_{s_1}$ and $X_{s_2}$ terms and makes the Hamiltonian:
\beq
H(t_1)=-J\mathlarger{\sum}_{i\neq 1,2}X_{s_i}-J\mathlarger{\sum}_{j}Z_{p_j}.
\eeq
This will make the state $|+_{DL}^X\>=|+_{SL}^X\>_1|+_{SL}^X\>_2$. This process can be done either adiabatically or instantaneously. If errors occur, they will leave nonzero syndromes for future correction, and no errors will be propagated when the $X_{s_1}$ and $X_{s_2}$ terms are turned off.

We can see that the distance for $\sigma_x$ errors is restricted by 4, no matter how far the pair of holes are separated. To increase the error protection ability of $\sigma_x$ errors, we need to enlarge the size of the holes. We will describe in detail the adiabatic procedure to enlarge the holes with gap protection in Sec.~\ref{sec:adiabatic_qbit_enlarge_move}.

Also note that all state preparations of this type are done right after the initialization of the whole surface, such that $X_{s_1}$ and $X_{s_2}$ are known to be $+1$ for certain. During the computation, $X_{s_1}$ and $X_{s_2}$ can be flipped to $-1$ before they are turned off, and we have no way to know their values except by doing syndrome measurement, which we try to avoid. So all qubits needed in the computation are prepared at the beginning.

\subsection{Enlarging the hole}\label{sec:adiabatic_qbit_enlarge_move}
%(Need to discuss in great detail)
After holes are created, we need to enlarge the size of the hole to improve the ability to correct $\sigma_z \ (\sigma_x)$ errors for $Z \ (X)$-cut double qubits. In this section, we will show how to enlarge the hole adiabatically with gap protection. First, we will assume that no error  occurs on any qubits during the process. Then we will analyze how errors propagate, and the fault-tolerance of the process. Since this is the first example where we apply the results of Sec.~\ref{sec:sketch_scheme}, we will follow the state transformations based on stabilizer formalism in detail.
\subsubsection{Scheme}
\begin{figure}[!ht]
\begin{center}
\centering\includegraphics[width=75mm]{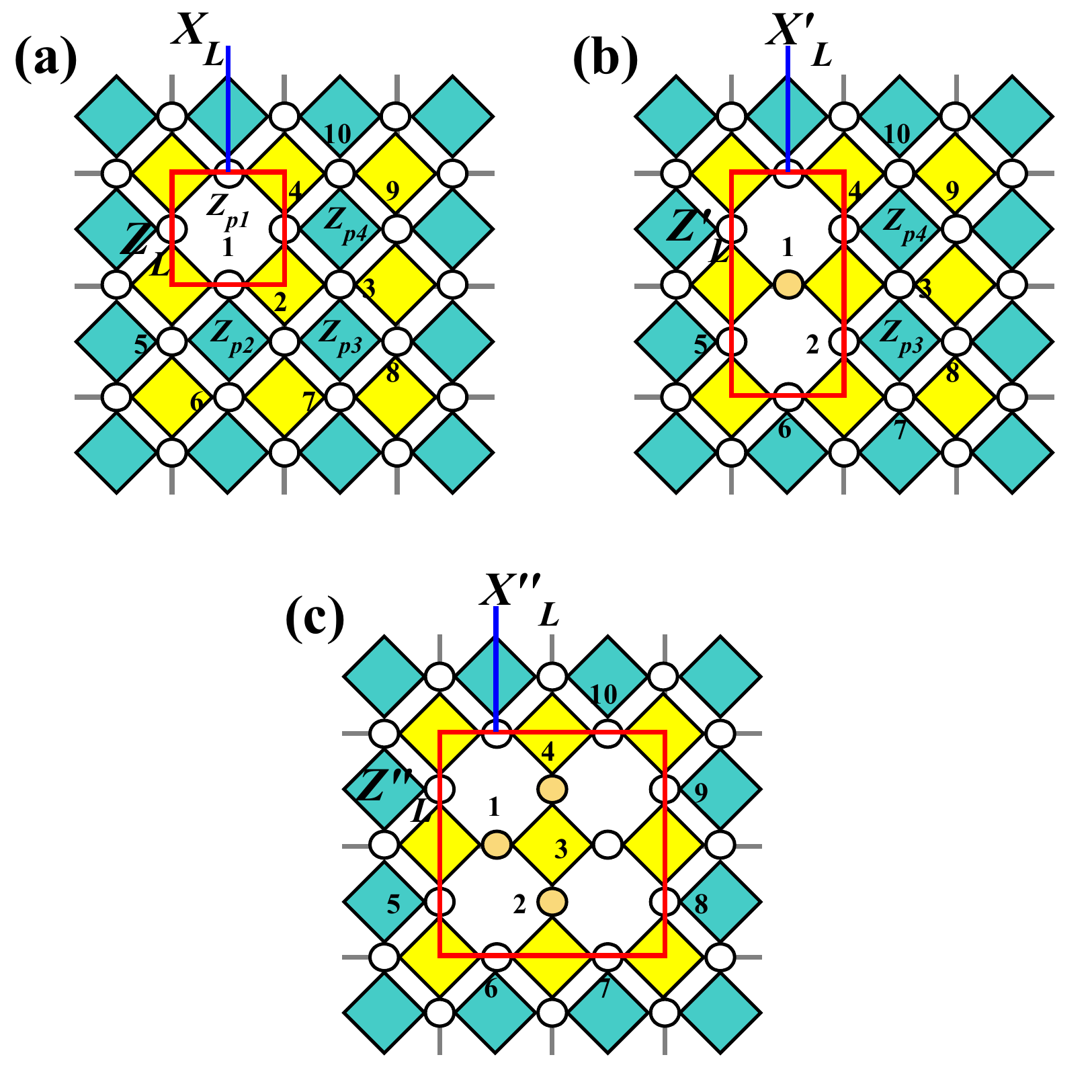}
\caption{\label{fig:hole_enlargement} (Color online) Enlarging a hole of an $X$-cut logical double qubit adiabatically. Colored squares indicate that the corresponding $X_s$ (yellow) and $Z_p$ (cyan) stabilizer generators are turned on. The yellow qubits in (b) and (c) indicate that $\sigma_x$ for that qubit is turned on in the Hamiltonian. Adiabatic evolution between (a) and (b) maps $X_L$ to $X_L'$ and $Z_L$ to $Z_L'$. Similarly, adiabatic evolution between (b) and (c) maps $X_L'$ to $X_L''$ and $Z_L'$ to $Z_L''$. }
\end{center}
\end{figure}
Consider the case of a $Z$-cut qubit, the situation for $X$-cut qubits is similar. Right after the creation of the pair of holes, we first expand one of the two holes vertically down and then horizontally right, as shown in Fig.~\ref{fig:hole_enlargement}. Following the spirit of Sec.~\ref{sec:sketch_scheme}, consider a circuit $\mathcal{G}$ composed of three gates of the form
$g_l=\exp\left(i\frac{\pi}{4}Q_l\right)$,
where $Q_l$ are defined as:
\beq
\begin{split}
Q_1 =&\sigma_{y_1}\sigma_{z_2}\sigma_{z_5}\sigma_{z_6},\\
Q_2 =&\sigma_{y_2}\sigma_{z_3}\sigma_{z_7}\sigma_{z_8},\\
Q_3 =&\sigma_{y_4}\sigma_{z_3}\sigma_{z_9}\sigma_{z_{10}}.
\end{split}
\eeq
In this case , the state is stabilized by $Z_{s_2}$, $Z_{s_3}$, $Z_{s_4}$
(and other stabilizer generators) with logical operators $X_L$
and $Z_{s_1}$. The transformation of the stabilizer generators and logical operators under $\mathcal{G}$ is listed in Table.~\ref{tab:hole_enlargement}. We can see that the circuit $\mathcal{G}$ maps logical operator $X_L$ and $Z_L$ to $X_L''$ and $Z_L''$ in panel (c) of Fig.~\ref{fig:hole_enlargement}, and also maps the system Hamiltonian in panel (a) to the ones shown in panel (c).
\begin{table}
\begin{tabular}{cccc}
\begin{tabular}{ c|c   }
  \hline
  \hline
  % after \\: \hline or \cline{col1-col2} \cline{col3-col4} ...
  $\mathbb{L}_1(t_0)$      & $Z_{s_1}$ \\
  $\mathbb{L}_2(t_0)$      & $X_{L}$\\
  \hline
  $S_1(t_0)$ & $Z_{s_2}$  \\
  $S_2(t_0)$ & $Z_{s_3}$  \\
  $S_3(t_0)$ & $Z_{s_4}$  \\
  \hline
  \hline
\end{tabular}
&{\Large$\substack{g_1 \\ \Rightarrow}$}&
\begin{tabular}{ c|c  }
  \hline
  \hline
  % after \\: \hline or \cline{col1-col2} \cline{col3-col4} ...
  $ \mathbb{L}_1(t_1) $ & $Z_{s_1}Z_{s_2}$ \\
  $\mathbb{L}_2(t_1)$      & $X_{L}$\\
  \hline
  $S_1(t_1)$ & $\sigma_{x_1}$  \\
  $S_2(t_1)$ & $Z_{s_3}$  \\
  $S_3(t_1)$ & $Z_{s_4}$  \\
  \hline
  \hline
\end{tabular}
&{\Large$\substack{g_2 \\ \Rightarrow}$}\\
& & &\\
\begin{tabular}{ c|c   }
  \hline
  \hline
  % after \\: \hline or \cline{col1-col2} \cline{col3-col4} ...
  $\mathbb{L}_1(t_2)$   & $Z_{s_1}Z_{s_2}Z_{s_3}$ \\
  $\mathbb{L}_2(t_2)$      & $X_{L}$\\
  \hline
  $S_1(t_2)$ & $\sigma_{x_1}$  \\
  $S_2(t_2)$ & $\sigma_{x_2}$ \\
  $S_3(t_2)$ & $Z_{s_4}$  \\
  \hline
  \hline
\end{tabular}
&{\Large$\substack{g_3 \\ \Rightarrow}$}&
\begin{tabular}{ c|c  }
  \hline
  \hline
  % after \\: \hline or \cline{col1-col2} \cline{col3-col4} ...
  $ \mathbb{L}_1(t_3) $ & $Z_{s_1}Z_{s_2}Z_{s_3}Z_{s_4}$ \\
  $\mathbb{L}_2(t_3)$      & $X_{L}$\\
  \hline
  $S_1(t_3)$ & $\sigma_{x_1}$  \\
  $S_2(t_3)$ & $\sigma_{x_2}$  \\
  $S_3(t_3)$ & $\sigma_{x_3}$  \\
  \hline
  \hline
\end{tabular}
\end{tabular}
\caption{The related transformation of stabilizer generators $\{S_i\}$ and logical operators $\{\mathbb{L}_i\}$ of a $Z$-cut qubit in Fig.~\ref{fig:hole_enlargement} is shown under gate operation $\{g_i\}$.}\label{tab:hole_enlargement}
\end{table}

Now we transform this procedure to an adiabatic one that gives the same state evolution following Theorem ~\ref{them:them1}. Set
$U_l(t,t_{l-1})=\exp\big(i\pi/4 f_l(t) Q_l\big)$ for time segment $t\in[t_{l-1},t_l]$, and adiabatically deform the Hamiltonian as in Eq.~(\ref{eq:Hamiltonian_change_each_step}). Note that $Q_1$ only anticommutes with the $Z_{p_2}$ term in the system Hamiltonian, which guarantees that even if errors occur, the adiabatic evolution is still valid (Lemma~\ref{lemma:error_propagation}). The situation is the same for $Q_2$ and $Q_3$. We first consider the adiabatic transformation generated by $U_1(t,t_0)$:
\beq
\begin{split}
H(t)=& -J\cos\left[f_1(t)\right]Z_{p_2}-J\sin\left[f_1(t)\right]\sigma_{x_1}\\
&-J\mathlarger{\sum}_{j\neq 1,2} Z_{p_j}-J\mathlarger{\sum}_i X_{s_i},
\end{split}
\eeq
for $t\in[t_0,t_1]$, with
\beq
H(t_1)= -J\sigma_{x_1}-J\mathlarger{\sum}_{j\neq 1,2} Z_{p_j}-J\mathlarger{\sum}_i X_{s_i}.
\eeq
At this time, qubit 1 is in the state $|+\>$.
For $U_2$ and $U_3$, we see that $Q_2$ commutes with $Q_3$, while $Q_2$ only anticommutes with $Z_{p_3}$, and $Q_3$ only anticommutes with $Z_{p_4}$.
According to Lemma~\ref{lemma:parrallelism}, the adiabatic procedures generated by $U_2$ and $U_3$ can be done simultaneously with the same state transformation as if done serially. The corresponding Hamiltonian deformation is
\beq
\begin{split}
H(t)=&-J\sigma_{x_1}-J\cos\left[f_2(t)\right]Z_{p_3}-J\sin\left[f_2(t)\right]\sigma_{x_2}\\
&-J\cos\left[f_2(t)\right]Z_{p_4}-J\sin\left[f_2(t)\right]\sigma_{x_4}\\
&-J\mathlarger{\sum}_{j\neq 1,2,3,4} Z_{p_j}-J\mathlarger{\sum}_i X_{s_i}
\end{split}
\eeq
for $t\in[t_1,t_2]$, with
\beq
H(t_2)=-J\sigma_{x_1}-J\sigma_{x_2}-J\sigma_{x_4}-J\mathlarger{\sum}_{j\neq 1,2,3,4} Z_{p_j}-J\mathlarger{\sum}_i X_{s_i},
\eeq
with qubits 1, 2, 3, and 4 all in the state $|+\>$, while they are all protected from $\sigma_z$ errors by the energy gap.

%Explain in details of this part.
%According to Theorem.~\ref{them:them1}, the state evolves to the one stabilized by the surface with the large hole.
This procedure can be generalized to obtain arbitrarily large square hole with distance equal to the perimeter $d$ (assuming $d$ is a multiple of 4). We first adiabatically expand $d/4$ times vertically down to form a long strip like that in panel (b) of Fig.~\ref{fig:hole_enlargement}, and then adiabatically expanding horizontally right parallel $d/4$ times as in panel (c). In all, we need about $d/2$ time steps of adiabatic evolution.

\subsubsection{Error propagation}
Even though the ground space is protected by an energy gap, there is still a nonzero probability that thermal excitations will occur at finite temperature. In this section, we apply the the result of Lemma~\ref{lemma:error_propagation} to study the propagation of these errors. If errors occur outside the hole or inside the hole, they will not be affected by the adiabatic process at all. However, if errors occur on the boundary of the hole before the adiabatic process, they may potentially propagate during the adiabatic procedure and cause uncorrectable logical errors.
\begin{figure}[!htp]
\begin{center}
\centering\includegraphics[width=70mm]{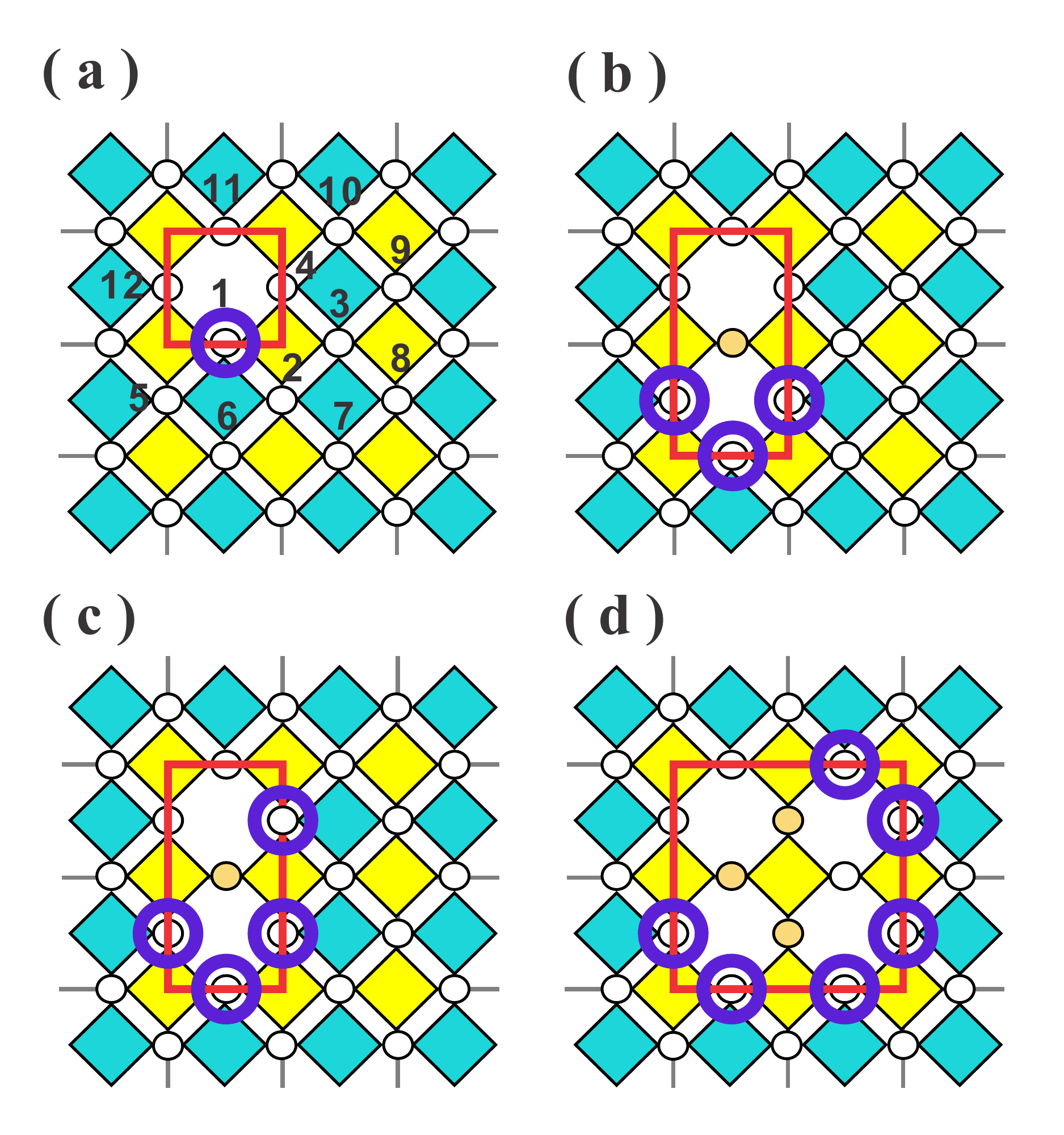}
\caption{\label{fig:hole_enlargement_error} (Color online) Error propagation during an adiabatic process to enlarge a hole of
an $X$-cut logical qubit. Colored squares indicate that
the corresponding $X_s$, $\sigma_x$ (yellow) and $Z_p$, $\sigma_z$ (cyan) operators are turned on. The purple circle around a qubit indicates a $\sigma_z$ error occurs on that qubit. (a) A $\sigma_z$ error occurs on qubit 1. (b) Effective errors after the adiabatic process.
(c) An additional $\sigma_z$ error occurs on qubit $4$.
(d) Effective errors after the adiabatic procedure to enlarge
the hole will cause a logical error after decoding.}
\end{center}
\end{figure}
Consider the case in Fig.~\ref{fig:hole_enlargement_error}. Before expanding the hole vertically down, assume a $\sigma_z$ error occurs on qubit 1. Then according to Lemma~\ref{lemma:error_propagation}, the $\sigma_z$ error will propagate to $\sigma_{x_1}\sigma_{z_2}\sigma_{z_5}\sigma_{z_6}$, as shown in panel (b). The effective errors are $\sigma_{z_2}\sigma_{z_5}\sigma_{z_6}$, since $\sigma_{x_1}$ has no effect because state of qubit 1 is $|+\>$. However, if another $\sigma_{z}$ error occurs on qubit 4, as shown in panel (c), then after
expanding horizontally rightward, we get effective errors $\sigma_{z_5}\sigma_{z_6}\sigma_{z_7}\sigma_{z_8}\sigma_{z_9}\sigma_{z_{10}}$, which occupy majority of the qubits around the hole. If the minimum-weight error correction is taken, it will close the path by applying $\sigma_{z_{11}}\sigma_{z_{12}}$ and cause a logical $Z$ error. So in general, this procedure is not fault-tolerant by the meaning of Def.~\ref{def:fault-tolerant}. However, we can get around this problem by the following observation: if before the hole expansion, the system is prepared in the $|0_{DL}^Z\>$, then a logical $Z$ error has no effect on the state. The situation is the same for the $|+_{DL}^X\>$ state for an$X$-cut double qubit.  Fortunately, as we will see later, in this scheme we only need to expand a $Z$-cut hole after creation a $|0_{DL}^Z\>$ state and $X$-cut hole after creation a $|+_{DL}^X\>$ state, so the non fault-tolerance of this procedure can be overcome.

\subsection{Moving logical qubits}\label{sec:hole_movement}
We now turn to the realization of logical gate operations in surface codes, like logical CNOT, $S$, Hadamard and $T$ gates. An element way to do these logical gates is by adiabatically moving the holes around each other on a single 2D lattice. In this section, we focus on the details of hole movement by adiabatically deforming the system Hamiltonian. We start with a scheme free of errors at first and then discuss the corresponding error propagation and fault-tolerance.

\subsubsection{Scheme}
We focus on the $Z$-cut qubit in this section, the method for the $X$-cut is similar. Consider a $Z$-cut qubit hole as shown in Fig.~\ref{fig:hole_movement}. Initially, the system Hamiltonian is
\beq
H(t_0)=-J\mathlarger{\sum}_{i=5}^{8}\sigma_{x_i}-J\mathlarger{\sum}_{i=12}^{14}\sigma_{x_i}
-J\mathlarger{\sum}_{j=1}^{4}Z_{p_j}+H_{\text{rest}},
\eeq
\begin{figure}[!ht]
\begin{center}
\centering\includegraphics[width=90mm]{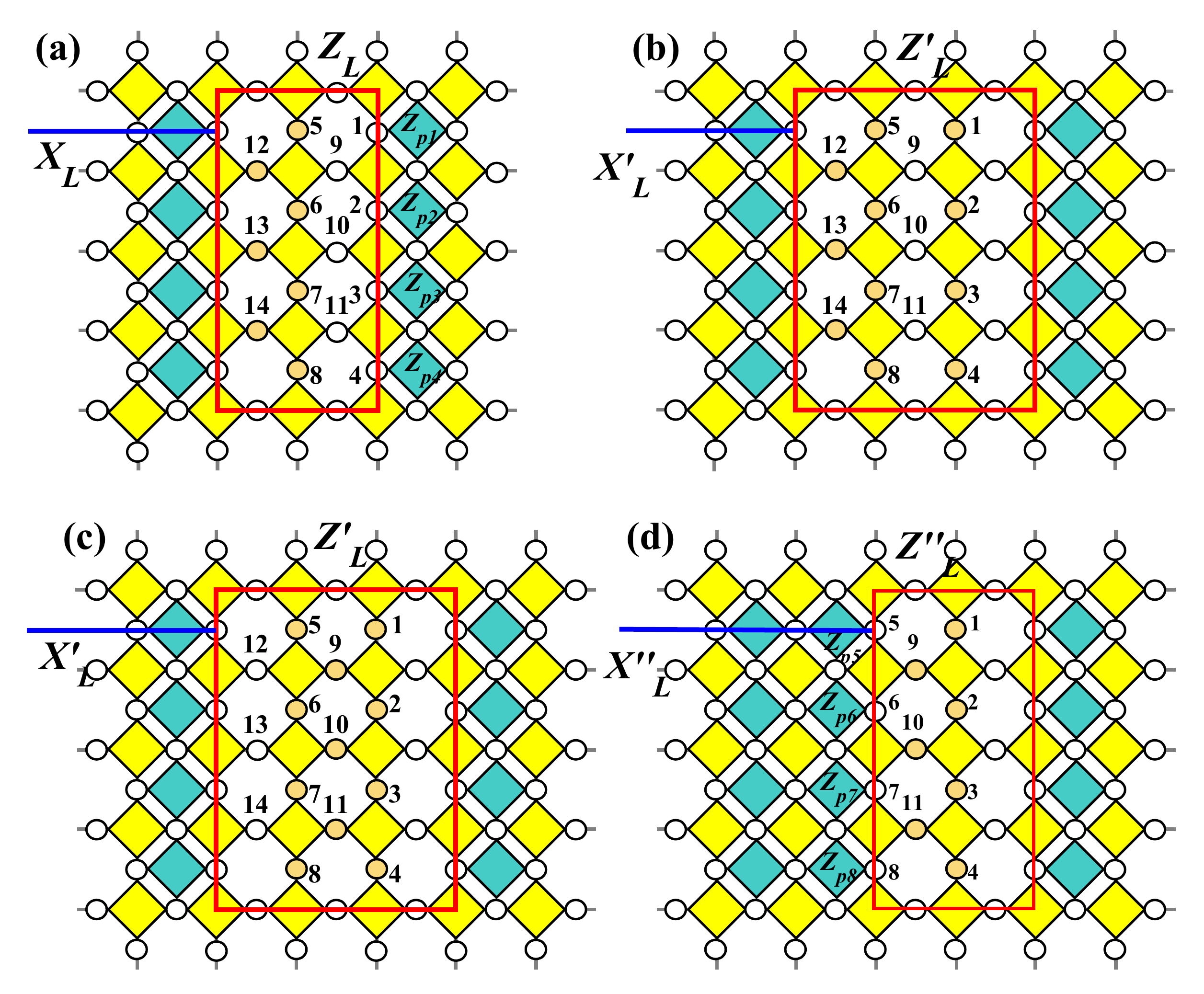}
\caption{\label{fig:hole_movement} (Color online) Adiabatic process for moving a $Z$-cut logical qubit hole horizontally right. Colored squares indicate that the corresponding $X_s$, $\sigma_x$ (yellow) and $Z_p$, $\sigma_z$ (cyan) operator are turned on. Logical operators of the qubit
are $X_L$ and $Z_L$ in (a). An adiabatic process between (a) and (b) maps $X_L$ to $X_L'$ and $Z_L$ to $Z_L'$. Similarly, an adiabatic process between (c) and (d) maps $X_L'$ to $X_L''$ and $Z_L'$ to $Z_L''$.}
\end{center}
\end{figure}
where $H_{\text{rest}}$ represents terms which are not altered in this process but are shown in Fig.~\ref{fig:hole_movement}. We start with a  circuit $\mathcal{G}$ composed of gates $\{g_l\}$ generated by $\{Q_l\}$. For illustration purposes, we divide them into two groups. We first expand the hole horizontally right as shown from panel (a) to panel (b), and then we shrink the hole rightward, as shown from panel (c) to panel (d). Consider the expansion procedure generated by:
\beq
Q_l =i\sigma_{x_l} Z_{p_l}, \ \ \ \ 1\leq l \leq 4,
\eeq
and the corresponding unitary transformations of the Hamiltonian $U_l = \exp\left(if_l(t)Q_l\right)$,
for $l$ from 1 to 4. We can see that each $Q_l$ anticommutes only with $Z_{p_l}$, so we can apply the adiabatic procedures generated by $Q_1$, $Q_2$, $Q_3$, $Q_4$ simultaneously
\beq
\begin{split}
H(t)=&-J\mathlarger{\sum}_{j=1}^{4}\big\{\cos[f_1(t)]Z_{p_j}+\sin[f_1(t)]\sigma_{x_j}\big\}\\
&-J\mathlarger{\sum}_{i=5}^{8}\sigma_{x_i}-J\mathlarger{\sum}_{i=12}^{14}\sigma_{x_i}+H_{\text{rest}},
\end{split}
\eeq
for $t\in[t_0,t_1]$, and obtain
\beq
H(t_1)=-J\mathlarger{\sum}_{i=1}^{4}\sigma_{x_i}-J\mathlarger{\sum}_{i=5}^{8}\sigma_{x_i}
-J\mathlarger{\sum}_{i=12}^{14}\sigma_{x_i}+H_{\text{rest}}
\eeq
at time $t_1$ as shown in panel (b). At this time, all qubits inside the hole are set to the $|+\>$ state. To contract the hole, rightward, we follow the circuit generated by $Q_l$,
\beq
Q_l =i Z_{p_l}\sigma_{x_l}, \ \ \ \ 5 \leq l \leq 8,
\eeq
and the corresponding unitary transformation of the Hamiltonian
$U_l = \exp\left(if_l(t)Q_l\right)$.
We need to be a little careful here, since $Q_l$ here anticommutes with two terms in the Hamiltonian. For example, $i\sigma_{x_{5}}Z_{p_5}$ anticommutes with both $\sigma_{x_5}$ and $\sigma_{x_{12}}$. To get around this, we turn off the terms $-J\sigma_{x_{12}}$, $-J\sigma_{x_{13}}$, $-J\sigma_{x_{14}}$ in the above equation, and turn on $-J\sigma_{x_9}$, $-J\sigma_{x_{10}}$, $-J\sigma_{x_{11}}$ instead. We can see that this procedure doesn't change the state of the system and can be done either adiabatically or instantaneously, making the Hamiltonian to be:
\beq
H^\prime(t_1)=-J\mathlarger{\sum}_{i=1}^{11}\sigma_{x_i}+H_{\text{rest}}.
\eeq
$Q_l$ now anticommutes with just one stabilizer generator (which is $\sigma_{x_l}$). Like the expansion process, we can adiabatically deform the Hamiltonian:
\beq
\begin{split}
H(t)=&-J\mathlarger{\sum_{j=5}^8}\big\{\cos[f_2(t)]\sigma_{x_j}+\sin[f_2(t)]Z_{p_j}\big\}\\
&-J\mathlarger{\sum}_{i=1}^{4}\sigma_{x_i}-J\mathlarger{\sum}_{i=9}^{11}\sigma_{x_i}+H_{\text{rest}},
\end{split}
\eeq
for $t\in[t_1,t_2]$, and obtain
\beq
H(t)=-J\mathlarger{\sum_{j=5}^8}Z_{p_j}
-J\mathlarger{\sum}_{i=1}^{4}\sigma_{x_i}-J\mathlarger{\sum}_{i=9}^{11}\sigma_{x_i}+H_{\text{rest}},
\eeq
which completes a full cycle of hole movement and leaves us ready for the next cycle of Hamiltonian deformation. The original ground space will be mapped to the one with a hole sitting one unit rightward of the original one (see panel (d)), and $X_L$ and $Z_L$ will be mapped to $X_L^{\prime\prime}$ and $Z_L^{\prime\prime}$ following Theorem.~\ref{them:them1}.
\begin{remark}
Note that the two steps of the adiabatic expansion and contraction of the hole can be combined into one step, if we turn off $-\sigma_{x_{12}}$, $-\sigma_{x_{13}}$, $-\sigma_{x_{14}}$ while turning on $-\sigma_{x_9}$, $-\sigma_{x_{10}}$ and $-\sigma_{x_{11}}$ at the beginning. So we need just one time step to adiabatically deform the system Hamiltonian to move a hole by one unit.
\end{remark}

\subsubsection{Error propagation and fault tolerance}
\begin{figure}[!ht]
\begin{center}
\centering\includegraphics[width=90mm]{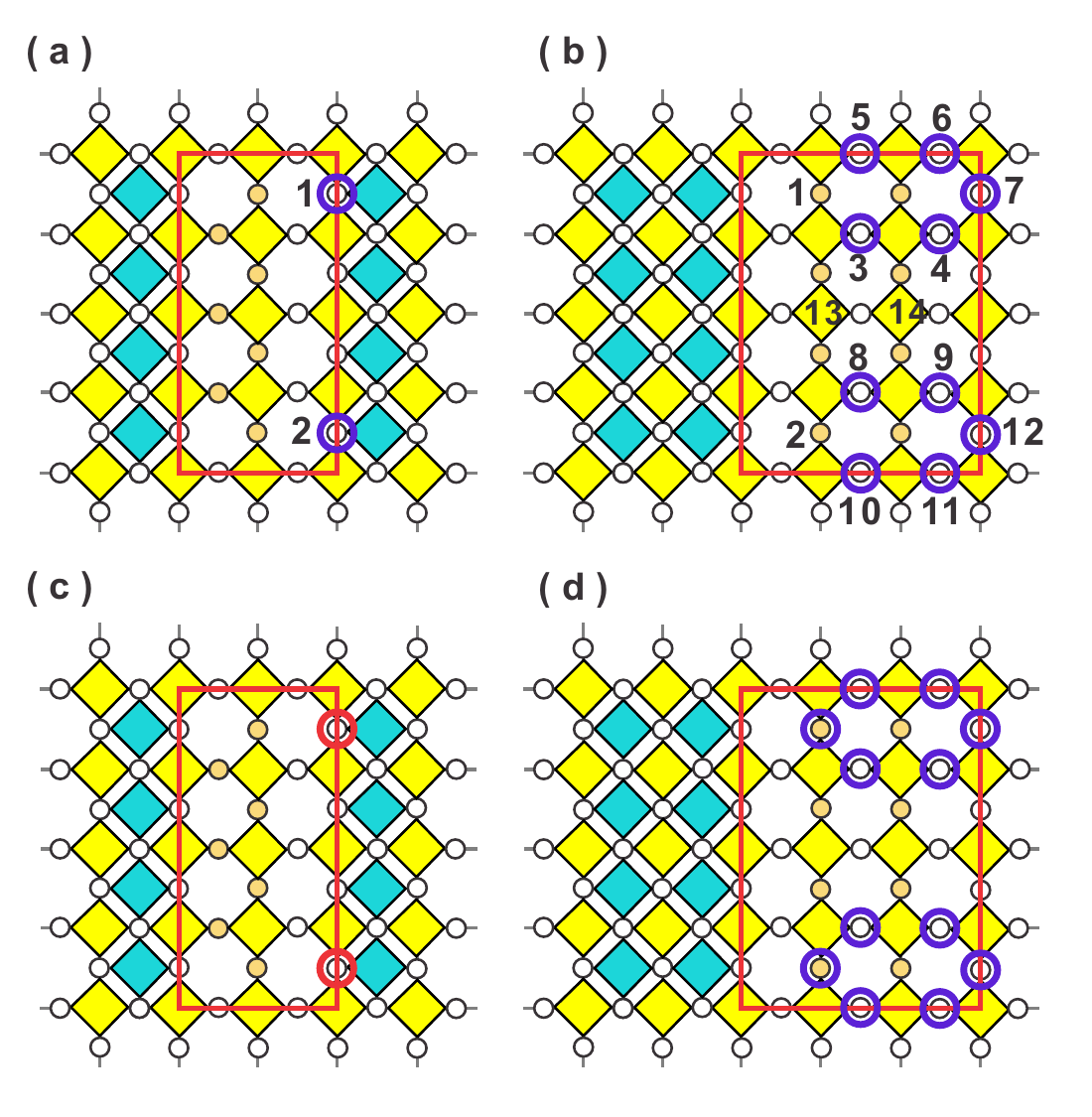}
\caption{\label{fig:hole_movement_error} (Color online) Error propagation during adiabatic process to move a hole of an $X$-cut logical double qubit horizontally right. Colored squares and qubits indicate that the corresponding $X_s$, $\sigma_x$ (yellow) and $Z_p$, $\sigma_z$ (cyan) operators are turned on. The purple circle around qubit indicates a $\sigma_z$ error occurs on that qubit and a red one indicates a $\sigma_x$ error occurs. (a) $\sigma_z$ errors occur on qubit 1 and 2. (b) Effective errors caused by $\sigma_{z_1}$ and $\sigma_{z_2}$ after adiabatic process. (c) $\sigma_x$ errors occur on qubit 1 and 2. (d) Effective errors caused by $\sigma_{x_1}$ and $\sigma_{x_2}$ after adiabatic process.}
\end{center}
\end{figure}
Like the case of hole enlargement, there's chance that thermal errors will cause an excitation. Errors outside or inside the holes will not be propagated by the process. However, if errors occur on the boundary of the hole before moving, they may potentially propagate to uncorrectable logical errors. Consider the case in Fig.~\ref{fig:hole_movement_error} for a 2 units movement rightward. Before expanding the hole horizontally right, assume $\sigma_z$ errors occurs on qubit 1 and qubit 2, as shown in panel (a). They will be propagated to:
\beq \sigma_{z_1}\sigma_{z_2}\mapsto\sigma_{z_3}\sigma_{z_4}\sigma_{z_5}\sigma_{z_6}\sigma_{z_7}\sigma_{z_8}
\sigma_{z_9}\sigma_{z_{10}}\sigma_{z_{11}}\sigma_{z_{12}},
\eeq
by the subsequent adiabatic operation, as shown in panel (b).
If we keep expanding the hole rightward, the errors will occupy more than half of the qubits on the perimeter of the hole, and cause a logical $Z$ error after later decoding. Similarly, if $\sigma_{x}$ errors occur on qubit 1 and qubit 2, the effective errors after the adiabatic procedure will be
\beq
\sigma_{x_1}\sigma_{x_2}\mapsto\sigma_{z_1}\sigma_{z_2}\sigma_{z_3}\sigma_{z_4}\sigma_{z_5}\sigma_{z_6}\sigma_{z_7}\sigma_{z_8}
\sigma_{z_9}\sigma_{z_{10}}\sigma_{z_{11}}\sigma_{z_{12}},
\eeq
as shown in panel (d). In general, the adiabatic procedure to move the hole on its own is not fault-tolerant, since the circuit $\mathcal{G}$ we follow to build the adiabatic procedure is not a fault-tolerant one, and the results from Ref.~\cite{Yi-Cong_PhysRevA.89.032317} cannot be
used here directly.

Fortunately, we can still make this process fault-tolerant. Errors that occur on the boundary of the hole, like qubit 1 and qubit 2 in this example, can be detected after each step of hole movement by measuring the qubits inside the hole after the expansion, since they are correlated, as shown in Fig.~\ref{fig:hole_movement_error}. In this case, we will do $\sigma_x$ measurement on qubit 3, 4, 8, 9, 13 and 14, when we are in panel (b). If any of these measurements give $-1$, it indicates that errors (which could be $\sigma_x$ or $\sigma_z$) occurred on the boundary's right side before the hole expansion, and we need to turn off the system Hamiltonian and do a full cycle of syndrome measurement and error correction before they become uncorrectable. A $\sigma_z$ error happens on the boundary with probability about $\exp\left(-4c\beta J\right)$ per time step, while $\sigma_x$ happens on the boundary with probability about $\exp(-2c\beta J)$, so the probability that we must do a full cycle of error correction during hole movement is low.

In practice, measurements themselves involve errors whose effect was discussed earlier in this section. Here, we need to check the probability that the measurement outcomes cause us to make a wrong decision about error correction. As an example, if a $\sigma_z$ error occurs on qubit 1 in panel (a), qubit 3 and 4 in panel (b) will not be protected by an energy gap, and we assume that the probability of a wrong measurement outcome in these cases is $p$ each time step.
Meanwhile, if a $\sigma_x$ error occurs on qubit 1 in panel (c),
qubit 3 and 4 in panel (d) are protected by an energy gap $4J$.
\begin{figure}[!ht]
\centering\includegraphics[width=85mm]{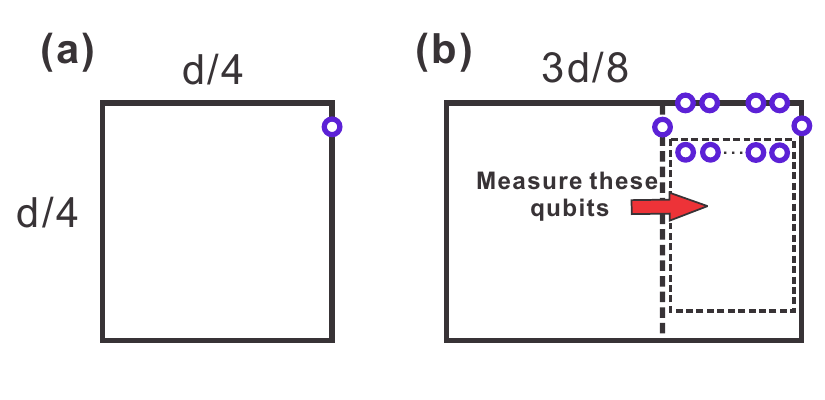}
\caption{\label{fig:FT_measurement_movement} (Color online) Scheme to fault-tolerantly detect errors occurring on the boundary. (a) Before the movement, an error occurs on the boundary. (b) After expanding the hole $d/8$ units rightward, the error propagates to a strip of errors. We measure all qubits in the dashed box and determine if the corresponding error happened on the boundary based on the majority vote of the measurement outcomes of each row of qubits.}
\end{figure}
%Although error rate may be small, still, the there's non-zero probability that error will propagate to uncorrectable error due the imperfect measurement.
Fortunately, we can make the uncorrectable error rate arbitrarily small by growing the lattice size and hole size, using majority vote. Consider a square hole with perimeter $d$ as shown in panel (a) of Fig.~\ref{fig:FT_measurement_movement}. Now we expand the hole $d/8$ units rightward and measure $\sigma_x$ in the dashed area of panel (b). The number $d/8$ is chosen so that error detection can be applied before an error can propagate to an uncorrectable error.
If an error occurs on the boundary of the hole before moving, it will corrupt an entire row of qubits in the dashed area of panel (b) in Fig.~\ref{fig:FT_measurement_movement}. So, for each row of qubits, we do a majority vote based on the measurement outcomes to determine whether an error happened on the boundary. For any row, if more than half of the measurement outcomes are $-1$, we infer that an corresponding error occurred at the boundary of the hole before moving, and therefore error correction must be applied. Let $E$ be the event that errors happened on the boundary before movement, and let $D$ be the event that we decide to do decoding and error correction based on the majority vote. Then the probability that such errors occurred on the boundary and is not detected is roughly
\beq
P_L=P(\bar{D},E)=P(\bar{D}|E)P(E)\sim O\left(\frac{d}{4}p^{\lfloor\frac{d}{16}\rfloor+1}e^{-4c\beta J}\right).
\eeq
Here $d/4$ indicates that misidentification can occur on any of $d/4$ rows. This gives a rough bound on the probability of logical errors during the $d/8$ unit hole movement.

On the other hand, the probability that no error occurred on the boundary, but we do an unnecessary decoding can be estimated as
\beq
P_{U}=P(D,\bar{E})=P(D|\bar{E})P(\bar{E})\sim O\left(\frac{d}{4}p^{\lfloor\frac{d}{16}\rfloor}e^{-4c \beta J}\right).
\eeq
We can see that both $P_L$ and $P_U$ can be made arbitrarily small
with the growth of hole size, and thus the adiabatic movement process
can be rendered fault-tolerant.
%\beq
%P_{e} \sim O\left(p^{\lfloor \frac{d}{16} \rfloor + 1}e^{-2c\beta J}\right).
%\eeq
%This gives a upper bound of probability of logical error during the $d/8$ units of hole movement:
%\beq
%P_{L} \sim O\left(p^{\lfloor\frac{d}{16}\rfloor+1}e^{-4c\beta J}\right).
%\eeq
%We can see the logical error can be made arbitrarily small with growth with the hole size and the movement procedure can be regarded as fault-tolerant.
\begin{remark}
We only analyzed the error propagation for the case of hole expansion.
It is worth noting that for the procedure to adiabatically contract
the hole, errors occurring on the boundary of the hole will
\emph{not} accumulate to uncorrectable logical errors, and thus
can be left for future error correction.
\end{remark}

%\begin{remark}
%%These are correlated error.
% On the other hand compare to traditional measurement based method. (by-product operator, correctable at each step). measurement requirement can be relaxed.
%\end{remark}

\subsection{Creation of $|0\> \ (|+\>)$ state for $X \ (Z)$-cut double qubit}\label{sec:difficult_initialization}
The second type of logical state initialization is to prepare the $|0\>$ state for an $X$-cut qubit or $|+\>$ state for a $Z$-cut qubit. We show an example for an $X$-cut qubit in detail. For a $Z$-cut qubit, the procedure is similar.

This can be done using a logical Hadamard after initializing the $|+\>$ state for $X$-cut qubit. However, we have not shown how to perform a logical Hadamard yet, and it is also extremely useful to directly initialize the $|0\>$ state for an $X$-cut qubit, as we will see in next few sections.

Suppose we have created a $|+\>$ state for an $X$-cut qubit with two holes attached to each other, as shown in panel (a) of Fig.~\ref{fig:difficult_initialization}.
\begin{figure}[!ht]
\centering\includegraphics[width=90mm]{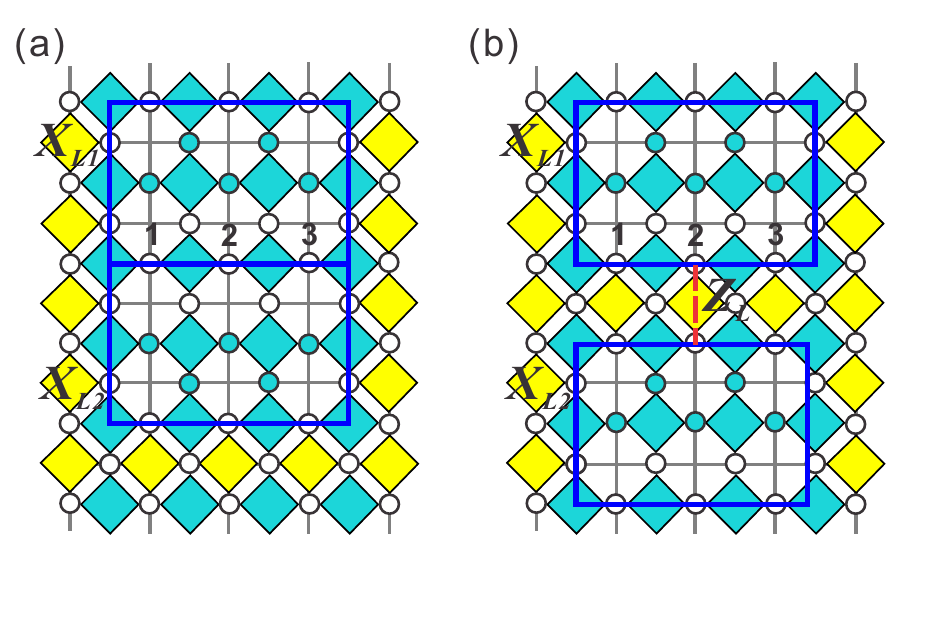}
\caption{\label{fig:difficult_initialization} (Color online) Creation of $|0\>$ state for $X$-cut double qubit. (a) Create a $|+_{DL}^X\>$ with two holes attached to each other. Measure $\sigma_{z_1}$, $\sigma_{z_2}$ and $\sigma_{z_3}$ and do majority vote to determine whether $|0_{DL}^X\>$ or $|1_{DL}^X\>$ is prepared. (b) Move two holes apart to increase error correction ability of $\sigma_z$ errors of the logical qubit. Note that both $\sigma_z$ and $\sigma_x$ errors on qubits between two holes during the adiabatic movement have no uncorrectable effect on logical state $|0_{DL}^X\>$ or $|1_{DL}^X\>$ and can be left for future error correction.}
\end{figure}
The logical $Z$ operator in this case can be $\sigma_{z_1}$, $\sigma_{z_2}$ or $\sigma_{z_3}$, and they all commute with the system Hamiltonian. If we measure any one of them, we can prepare the logical state $|0\>$ or $|1\>$. Either one is useful as long as we know which state it is for certain. If any $\sigma_z$ errors occur on these qubits, it will have no effect, and any single $\sigma_x$ errors on these qubits suffers an energy penalty of $4J$ and leaves $Z_p$ operators nearby flipped and correctable by a future error correction procedure. However, when a $\sigma_x$ happens on these qubits, it will give an incorrect measurement outcome, and will affect any future operations conditioned on whether the state is $|0\>$ or $|1\>$. This can also be resolved by measuring $\sigma_z$ on all three qubits and taking the majority vote to determine the measurement outcome. This procedure can be extended to the square hole with perimeter $d$, where there are $d/4$ qubits shared by two holes. Note that the first measurement error is suppressed by the energy penalty, and occurs with probability $\exp(-4c\beta J)$, while the subsequent measurement errors may not suffer an energy penalty. We assume that the probability to obtain a wrong measurement result is $p$ each time step. The probability that we prepare a $|0_{DL}^X\>\ (|1_{DL}^X\>)$ state with an erroneous measurement $-1\ (+1)$ can be estimated to be:
\beq
P_L\sim O(p^{\lfloor \frac{d}{8}\rfloor+1}e^{-4c\beta J}),
\eeq
which decreases rapidly with the growth of the hole size, and can be made arbitrarily small. After the measurement, we separate the two holes by distance $d$, as illustrated in panel (b) of Fig.~\ref{fig:difficult_initialization} for a single time step of movement. It takes about $d/2$ time steps in total to move the pairs of holes apart by distance $d$ if the two holes
move simultaneously. Any $\sigma_x$ and $\sigma_z$ errors on qubit 1, 2, 3 will not propagate to uncorrectable errors during the movement. The hole movement process can be done adiabatically and fault-tolerantly with gap protection, as described in the previous section. Thus, the whole state preparation process can be made fault-tolerant.

\subsection{Logical $Z\ (X)$ measurement for $X\ (Z)$-cut double qubit}\label{sec:easy_measurement}
%\begin{figure}[!ht]
%\centering\includegraphics[width=60mm]{measure_z_xcut.pdf}
%\caption{\label{fig:difficult_initialization} (Color online) Adiabatic Z eigenstate (difficult) initialization of an $X$-cut logical double qubit. Colored square indicates that the corresponding $X$ (yellow) and $Z$ (blue) stabilizer is turned on. The blue dots on qubits in (b) and (c) indicate that $-\sigma_z$ of that qubit is turned on in Hamiltonian.}
%\end{figure}
%\begin{figure}[!ht]
%\centering\includegraphics[width=85mm]{measurement_logical_X.pdf}
%\caption{\label{fig:difficult_initialization} (Color online) Adiabatic Z eigenstate (difficult) initialization of an $X$-cut logical double qubit. Colored square indicates that the corresponding $X$ (yellow) and $Z$ (blue) stabilizer is turned on. The blue dots on qubits in (b) and (c) indicate that $-\sigma_z$ of that qubit is turned on in Hamiltonian.}
%\end{figure}
Like the case of initialization, there are two types of measurement procedures. The first is measuring in the $Z\ (X)$ basis for an $X\ (Z)$-cut qubit while the second is measuring in the $Z\ (X)$ basis for a $Z\ (X)$-cut qubit.

The first type of measurement is essentially the reverse process of creating the state $|0\>\ (|+\>)$ for an $X\ (Z)$-cut qubit. For an $X$-cut qubit shown in Fig.~\ref{fig:difficult_initialization}, we first move two holes that are initially $d$ units apart together to contact each other, and then measure $\sigma_z$ on all qubits shared by the two holes and take a majority vote of the outcomes. After that, we separate the two holes back to their original positions. Note that unlike traditional measurement-based QC on the surface code, this measurement is non-destructive and we do not annihilate the holes. The measurement procedure can also be viewed as a logical state preparation that will be used in the future computation. The second type of measurement procedure will be discussed in Sec.~\ref{sec:difficult_measurement}.

\subsection{Holonomic Logical CNOT}
The logical CNOT gate is one of the most important logical operations in the surface code HQC scheme. Based on our results on adiabatic hole movement, we can realize the logical CNOT gate. In this section, we show that by adiabatically braiding one hole around a different type of hole, we can get a closed loop holonomy which can be recognized as a logical CNOT. Starting from panel (a) of Figs.~\ref{fig:braiding_logical_X} and ~\ref{fig:braiding_logical_Z}, the adiabatic movement procedure is shown in details from panel (b) to panel (f). In Fig.~\ref{fig:braiding_logical_X}, following the discussion in Sec.~\ref{sec:hole_movement} and Theorem.~\ref{them:them1}, $X_{L_1}\otimes I_{L_2}$ transforms to $X_{L_1}\otimes X_{L_2}$ up to a multiplication by $X_s$ stabilizer generators inside the dashed square. We can conclude that $X_L$ operators transform in the following way:
\begin{figure}[!htp]
\centering\includegraphics[width=83mm]{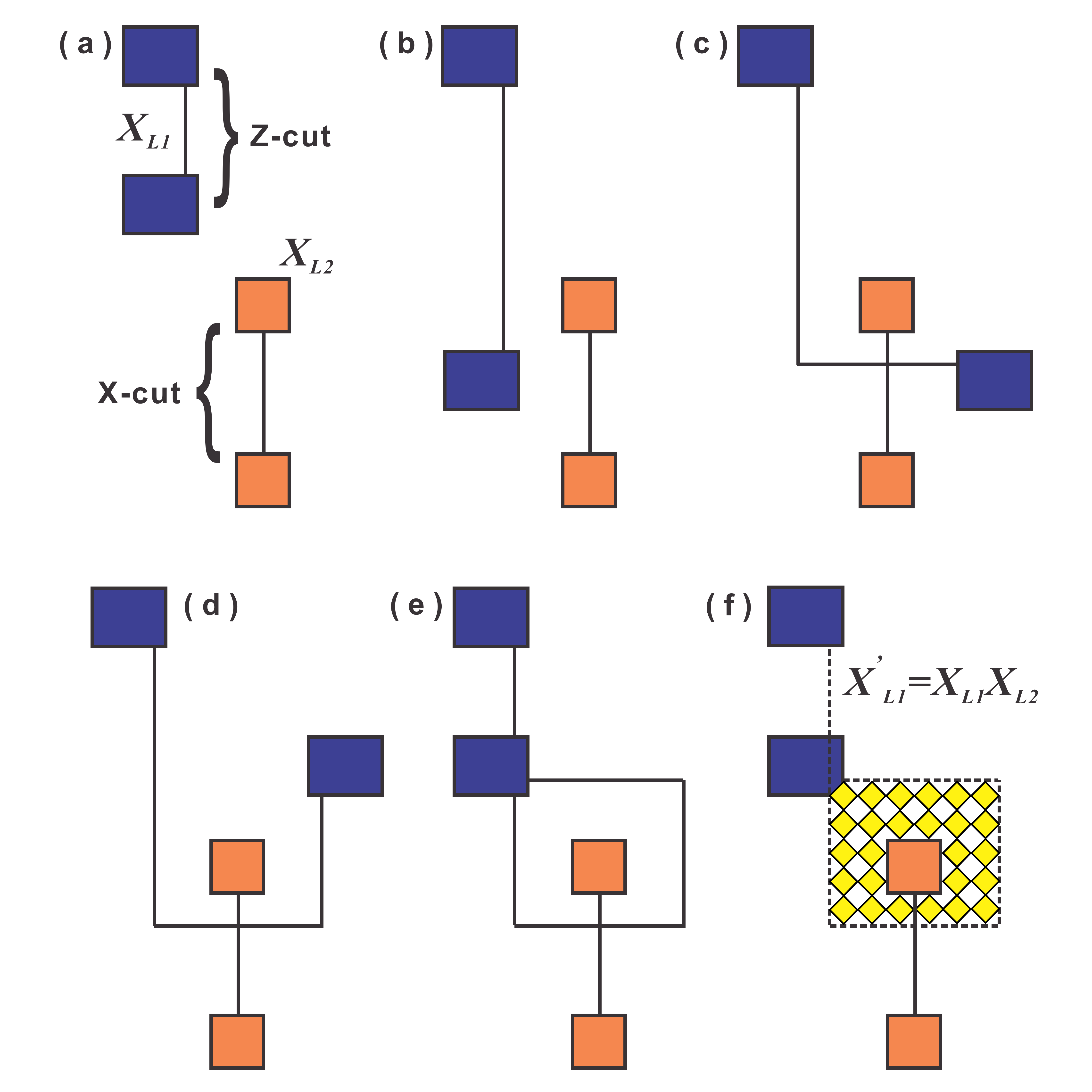}
\caption{\label{fig:braiding_logical_X} (Color online) Adiabatic braiding process of a $Z$-cut hole (dark blue) around an $X$-cut hole (orange) . The operator $X_{L_1}$ has been stretched to multiply a loop of $\sigma_x$ operators which is equivalent to $X_{L_2}$ up to multiplication by $X_s$ stabilizer generators (yellow) inside the loop, while $X_{L_2}$ remains the same under the transformation.}
\end{figure}
\begin{figure}[!htp]
\centering\includegraphics[width=83mm]{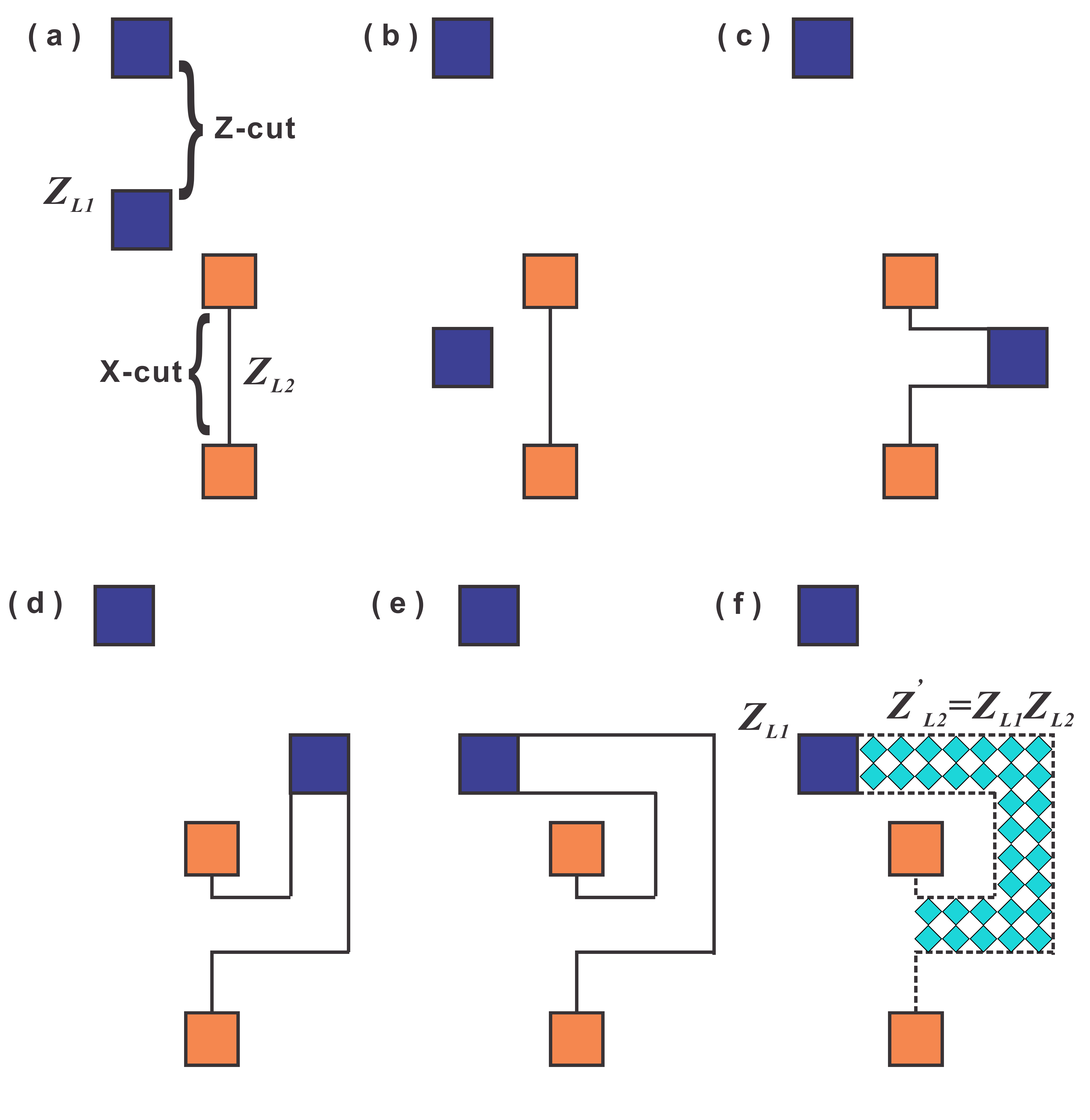}
\caption{\label{fig:braiding_logical_Z} (Color online) Adiabatic braiding process of a $Z$-cut hole (dark blue) around an $X$-cut hole (orange). The operator $Z_{L_2}$ has been stretched to form a strip of $\sigma_z$ operators, which is equivalent to $Z_{L_1}$ up to multiplication by $Z_p$ stabilizer generators (cyan) inside the strip, while $Z_{L_1}$ remains the same under the transformation.}
\end{figure}
\beq
\begin{split}
X_{L_1}\otimes I_{L_2}&\rightarrow X_{L_1}\otimes X_{L_2},\\
I_{L_1}\otimes X_{L_2}&\rightarrow I_{_L1}\otimes X_{L_2}.
\end{split}
\eeq
Similarly, from Fig.~\ref{fig:braiding_logical_Z}, we can see that $I_{L_1}\otimes Z_{L_2}$ transforms to $Z_{L_1}\otimes Z_{L_2}$ up to  multiplication by $Z_p$ stabilizer generators inside the strip. The $Z_L$ operators transform as:
\beq
\begin{split}
Z_{L_1}\otimes I_{L_2}&\rightarrow Z_{L_1}\otimes I_{L_2},\\
I_{L_1}\otimes Z_{L_2}&\rightarrow Z_{L_1}\otimes Z_{L_2}.
\end{split}
\eeq
The closed loop adiabatic evolution can be recognized as a closed loop holonomy which gives a logical CNOT with a $Z$-cut qubit as the control and an $X$-qubit as the target. It also reflects the topological property of braiding on 2D lattice since local deformation of movement path does not have effects on the state. Note that the fault-tolerance of this operation is guaranteed by the fault-tolerance of adiabatic hole movement.

CNOTs from $Z$-cut qubits to $X$-cut qubits are not enough. We need to extend to CNOTs between logical qubits of the same type. For $Z$-cut qubits, we have the following circuit:
%\begin{mytheorem}
%The adiabatic procedure described in xxx gives holonomic logical CNOT between $X$-cut qubit and $Z$-cut qubit
%\end{mytheorem}
\begin{center}
$\Qcircuit @C=1em @R=1em {
\lstick{\text{Z-cut control in}}&\qw       & \ctrl{1} &\qw      & \rstick{\text{Z-cut control out}}\qw \\
\lstick{\ket{0^X_{DL}}}&\targ     & \targ    &\targ    & \measureD{M_{Z}}\\
\lstick{|+^Z_{DL}\>} &\qw &\qw &\ctrl{-1} &\rstick{\text{Z-cut target out}} \qw \\
\lstick{\text{Z-cut target in}}&\ctrl{-2} & \qw      &\qw      & \measureD{M_{X}}\\
}$
\end{center}
which is equivalent to $Z_L^{(1-M_X)/2}$ on the target qubit followed by a CNOT, then followed by $X_L^{(1-M_Z)/2}$ on the target qubit. Similarly, the CNOT between two $X$-cut logical qubits can be built from following circuit:
\begin{center}
$\Qcircuit @C=1em @R=1em {
\lstick{\ket{0^X_{DL}}}&\qw       & \qw    &\targ     & \rstick{\text{X-cut control out}}\qw \\
\lstick{\text{X-cut control in}}&\targ &\qw     & \qw    & \measureD{M_{Z}}\\
\lstick{\ket{+^Z_{DL}}}&\ctrl{-1} &\ctrl{1}      &\ctrl{-2}& \measureD{M_{X}} \\
\lstick{\text{X-cut target in}}&\qw       & \targ  &\qw    & \rstick{\text{X-cut target out}}\qw\\
}$
\end{center}
up to a correction of logical $X$s and $Z$s. The last kind of CNOT, with an $X$-cut qubit as control and a $Z$-cut as target, can be obtained from the circuit realizing CNOT between $Z$-cut qubits:
\begin{center}
$\Qcircuit @C=1em @R=1em {
\lstick{\ket{0^X_{DL}}}&\qw       & \qw    &\targ     & \rstick{\text{X-cut control out}}\qw \\
\lstick{\text{X-cut control in}}&\targ &\qw     & \qw    & \measureD{M_{Z}}\\
\lstick{\ket{+^Z_{DL}}}&\ctrl{-1} &\ctrl{1}      &\ctrl{-2}& \measureD{M_{X}} \\
\lstick{\text{Z-cut target in}}&\qw       & \targ  &\qw    & \rstick{\text{Z-cut target out}}\qw\\
}$
\end{center}
Note that for all four different logical CNOTs, the building block is the CNOT from $Z$-cut to $X$-cut. In addition, we also need to prepare ancillas in logical $|0_{DL}^X\>$ and $|+^Z_{DL}\>$ (which is shown in Sec.~\ref{sec:difficult_initialization}), and to do $Z$ measurements of $X$-cut qubits and $X$ measurements of $Z$-cut qubit (as discussed in Sec.~\ref{sec:easy_measurement}). All of these procedures can be done fault-tolerantly, and thus make all kinds of
logical CNOT fault-tolerant.

\subsection{Measurement of $Z\ (X)$ basis for $Z\ (X)$-cut double qubit}\label{sec:difficult_measurement}
This type of measurement is necessary when doing state distillation (discussed later). Naively, this process can be done by contracting the size of the hole and doing stabilizer measurements. However, stabilizer measurement is not compatible with the system Hamiltonian. What is worse, we close the hole after the measurement to destroy the logical qubit, and we cannot reuse it later. To avoid these problems, we can use the following circuits for $Z$ and $X$ measurement of $Z$-cut and $X$-cut qubits, respectively:
\begin{center}$
\Qcircuit @C=1em @R=1em {
\lstick{\ket{\psi^Z_{DL}}} & \ctrl{1} & \qw  &\qw  &  &  &  & &\lstick{\ket{\psi^X_{DL}}} & \targ & \qw &\qw \\
\lstick{\ket{0_{DL}^{X}}} & \targ & \qw &\measureD{M_{Z}} & & & &  & \lstick{\ket{+_{DL}^{Z}}} & \ctrl{-1} & \qw &\measureD{M_{X}}
}$
\end{center}
These circuits take an ancilla state $|0_{DL}^X\>$ or $|+_{DL}^Z\>$, and a logical CNOT with a $Z$-cut qubit as the control and an $X$-cut qubit as the target, which can both be realized fault-tolerantly. Thus, this type of measurement procedure is fault-tolerant.
Note that, like the measurement of the first type in Sec.~\ref{sec:easy_measurement}, this measurement procedure doesn't annihilate the hole after measurement. The ancilla qubits
after measurement are effectively prepared to $|0_{DL}^X\>$ (or $|1_{DL}^X\>$) fault-tolerantly, which can be used again as ancillas for future computation.

\subsection{Ancilla recycling}
As we have seen so far, to implement different types of CNOTs, we need to frequently create and measure logical qubits. Moreover, state distillation procedures also need large number of fresh ancilla qubits and logical state measurements. We have discussed two different types of state creation---$|0\>\ (|+\>)$ for $X\ (Z)$-cut and $|0\>\ (|+\>)$ for $Z\ (X)$-cut---and two different types of measurement---$X\ (Z)$ measurement for $Z\ (X)$-cut qubit and $Z\ (X)$ measurement for $Z\ (X)$-cut qubit. All can be done fault-tolerantly with constant gap protection, and both kinds of logical state measurement can be made non-destructive, so states after measurement can be reused as ancillas to avoid having to create a new logical qubits. This is particularly important, as we have seen that to create a logical state we need to turn off some $X_s$ or $Z_p$ operators, whose eigenvalues are uncertain when stabilizer Hamiltonian is turned on. With this ancilla recycling process, we can prepare all logical qubits, data or ancilla, right after we turn on the system Hamiltonian at the very beginning of the computation and never create new logical qubits during the computation.

\subsection{State injection}
As will be seen in Sec.~\ref{sec:S_T} and \ref{sec:Hadamard}, to get the logical $S$, $T$ and Hadamard gates, we need to create particular logical ancilla states $|Y_{DL}\>=\frac{1}{\sqrt{2}}(|0_{DL}\>+i|1_{DL}\>)$ and $|A_{DL}\>=\frac{1}{\sqrt{2}}(|0_{DL}\>+e^{i\pi/4}|1_{DL}\>)$. However, there's no obvious way to perform arbitrary rotation of logical qubit with large distance and local Hamiltonians transformation. To deal with this problem, we need to create a logical qubit in which the logical $Z$ operator is just one $\sigma_z$ on single qubit, with the stabilizer Hamiltonian turned on.
\begin{figure}[!ht]
\centering\includegraphics[width=90mm]{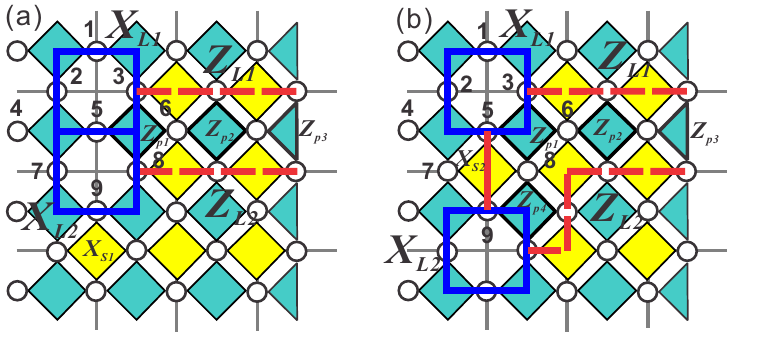}
\caption{\label{fig:state_injection} (Color online) State injection for a $X$-cut qubit. Colored squares indicate that the corresponding $X_s$ (yellow) and $Z_p$ (cyan) operator is turned on.}
\end{figure}
We focus on $X$-cut double qubits. We first put an \emph{existing} $X$-cut qubit into the state $|+^X_{DL}\>$ with the two holes attached to each other, as in panel (a) of Fig.~\ref{fig:state_injection}. This can be done by doing a logical $X$ measurement on an existing $X$-cut qubit (Sec.~\ref{sec:difficult_measurement}) and moving the two holes together. Without loss of generality, assume the state after measurement to be $|+_{DL}^X\>$. Note that $Z_L=\sigma_{z_5}$ is equivalent to $Z_{L_1}Z_{L_2}$ up to multiplication by $Z_p$ operators, as shown in panel (a), which gives:
\beq
\sigma_{z_5}=Z_{p_1}Z_{p_2}Z_{p_3}Z_{L_1}Z_{L_2}.
\eeq
For the $|\pm_{DL}^X\>$ state, the effect of $\sigma_{z_5}$ is
\beq
\begin{split}
\sigma_{z_5}|\pm^X_{DL}\>&=\sigma_{z_5}|\pm_{SL}^X\>_1|\pm_{SL}^X\>_2\\
&=Z_{L_1}Z_{L_2}|\pm_{SL}^X\>_1|\pm_{SL}^X\>_2\\
&=|\mp_{SL}^X\>_1|\mp_{SL}^X\>_2.
\end{split}
\eeq
%\beq
%\begin{split}
%&|+^X_{SL}\>_1|+^{X}_{SL}\>_2\\
%=&\frac{1}{\sqrt{2}}\bigg(\frac{|+_{SL}^X\>_1|+_{SL}^X\>_2+|-_{SL}^X\>_1|-_{SL}^X\>_2}{\sqrt{2}}\\
%+&\frac{|+_{SL}^X\>_1|+_{SL}^X\>_2-|-_{SL}^X\>_1|-_{SL}^X\>_2}{\sqrt{2}}\bigg)
%\end{split}
%\eeq
Applying a pulse $V_c=g\sigma_{z_5}$ for a short time $\tau$, with Hamiltonian
\beq
H=H_{\text{stab}}+V_c,
\eeq
we can see that $[V_c,H_{\text{stab}}]=0$, where  $H_{\text{stab}}$ is the stabilizer Hamiltonian shown in panel (a). The pulse will not cause a transition from the ground space to another eigenspace of $H_{\text{stab}}$. If $\tau$ is chosen such that $g\tau=\theta/2$, we have the state evolution:
\beq
\begin{split}
&\exp\left(-i\frac{\theta}{2}\sigma_{z_5}\right)|+^X_{SL}\>_1|+^{X}_{SL}\>_2\\
=&\frac{e^{-i\frac{\theta}{2}}}{\sqrt{2}}\bigg(\frac{|+_{SL}^X\>_1|+_{SL}^X\>_2+|-_{SL}^X\>_1|-_{SL}^X\>_2}{\sqrt{2}}\\
+&e^{i\frac{\theta}{2}}\frac{|+_{SL}^X\>_1|+_{SL}^X\>_2-|-_{SL}^X\>_1|-_{SL}^X\>_2}{\sqrt{2}}\bigg)\\
=&\frac{e^{-i\frac{\theta}{2}}}{\sqrt{2}}\left(|0^X_{DL}\>+e^{i\theta}|1^{X}_{DL}\>\right),
\end{split}
\eeq
which gives the desired state we want to inject. Note that if a $\sigma_{x_5}$ error occurs, it will suffer from the energy penalty, and cause the $Z_p$s adjacent to it to be flipped, leaving the syndrome for future error correction. On the other hand, the imprecise control of the pulse $V_c$ can affect the state injected and cannot be detected. However, as long as rate of $\sigma_{z_5}$ error is lower than a threshold, logical states $|Y_{DL}\>$ and $|A_{DL}\>$ can be obtained with sufficient precision by state distillation~\cite{reichardt2005quantum}. Then two holes can be adiabatically separated to distance $d$ to better protect against errors, as illustrated in panel (b) of Fig.~\ref{fig:state_injection}.

The process of state injection for a $Z$-cut qubit is slightly more complicated. We first inject state the $|\psi\>=|Y\>$ or $|A\>$ for an $X$-cut qubit and prepare a $Z$-cut qubit in state $|+\>$ and then we swap the state of these two logical qubits using following circuit:
\begin{center}$
\Qcircuit @C=1em @R=1em {
\lstick{\ket{\psi^X_{DL}}}&\qw &\targ &\ctrl{1} & \targ &\qw &\rstick{|+_{DL}^X\>}\qw \\
\lstick{\ket{+_{DL}^{Z}}} &\qw &\ctrl{-1} &\targ & \ctrl{-1}& \qw &\rstick{|\psi_{DL}^Z\>}\qw
}$
\end{center}
Note that the $|+_{DL}^X\>$ is ready to be reused for state injection, and all process included here can be done fault-tolerantly.

%which
%\beq
%H(t)=-J\cos[f_1(t)]X_{s_1}-J\sin[f_1(t)]\sigma_{z_9}+H_{\text{rest}}
%\eeq
%
%\beq
%H(t)=-J\cos[f_2(t)]\sigma_{z_9}-J\sin[f_2(t)]X_{s_2}+H_{\text{rest}}
%\eeq
%
%The $\sigma_x$ error on qubit 5 is detectable and correctable, and will not propagate during the procedure. The $\sigma_z$ error on qubit 5 can be viewed a logical error the logical qubit we want to inject. As long as the error rate of $\sigma_{z_5}$ is low than some threshold, we can do state distillation.

\subsection{State Distillation}
%\cite{Bravyi_Haah_PhysRevA.86.052329}
The logical ancilla states, $|Y\>=|0\>+i|1\>$ and $|A\>=|0\>+e^{i\pi/4}|1\>$ after injection are not good enough in general for the purpose of fault-tolerant QC. Fortunately, they can be distilled to much higher fidelity~\cite{Bravyi:2005:022316}. The reversed encoding circuit for $7$-qubit Steane code can be used to distill the $|Y\>$ state, with seven input logical states approximately equal to $|Y\>$~\cite{Fowler:2009:052312} as shown in Fig.~\ref{fig:Y_distillation}. The output $|\psi\>$ will be closer to the logical $|Y\>$ state.
%\begin{figure}[!ht]
%\begin{center}$
%\Qcircuit @C=1em @R=0.9em{
%\lstick{\ket{\widetilde{Y}_{DL}}} & \ctrl{2}  & \qw    &\qw &\qw &\measureD{M_X}\\
%\lstick{\ket{\widetilde{Y}_{DL}}} & \qw       & \ctrl{1} &\qw  &\qw&\measureD{M_X}\\
%\lstick{\ket{\widetilde{Y}_{DL}}} & \targ\qwx[2] & \targ \qwx[3] &\qw & \ctrl{2}  &\rstick{\ket{Y_{DL}}}\qw\\
%\lstick{\ket{\widetilde{Y}_{DL}}} & \qw   & \qw & \ctrl{1} &\qw&\measureD{M_X}\\
%\lstick{\ket{\widetilde{Y}_{DL}}} & \targ \qwx[2] &\qw & \targ \qwx[1] &\targ \qwx[1] &\measureD{M_Z}\\
%\lstick{\ket{\widetilde{Y}_{DL}}} & \qw    & \targ \qwx[1] &\targ \qwx[1]  & \targ&\measureD{M_Z}\\
%\lstick{\ket{\widetilde{Y}_{DL}}} & \targ  & \targ & \targ & \qw&\measureD{M_Z}
%}$
%\end{center}
\begin{figure}[!ht]
\centering\includegraphics[width=55mm]{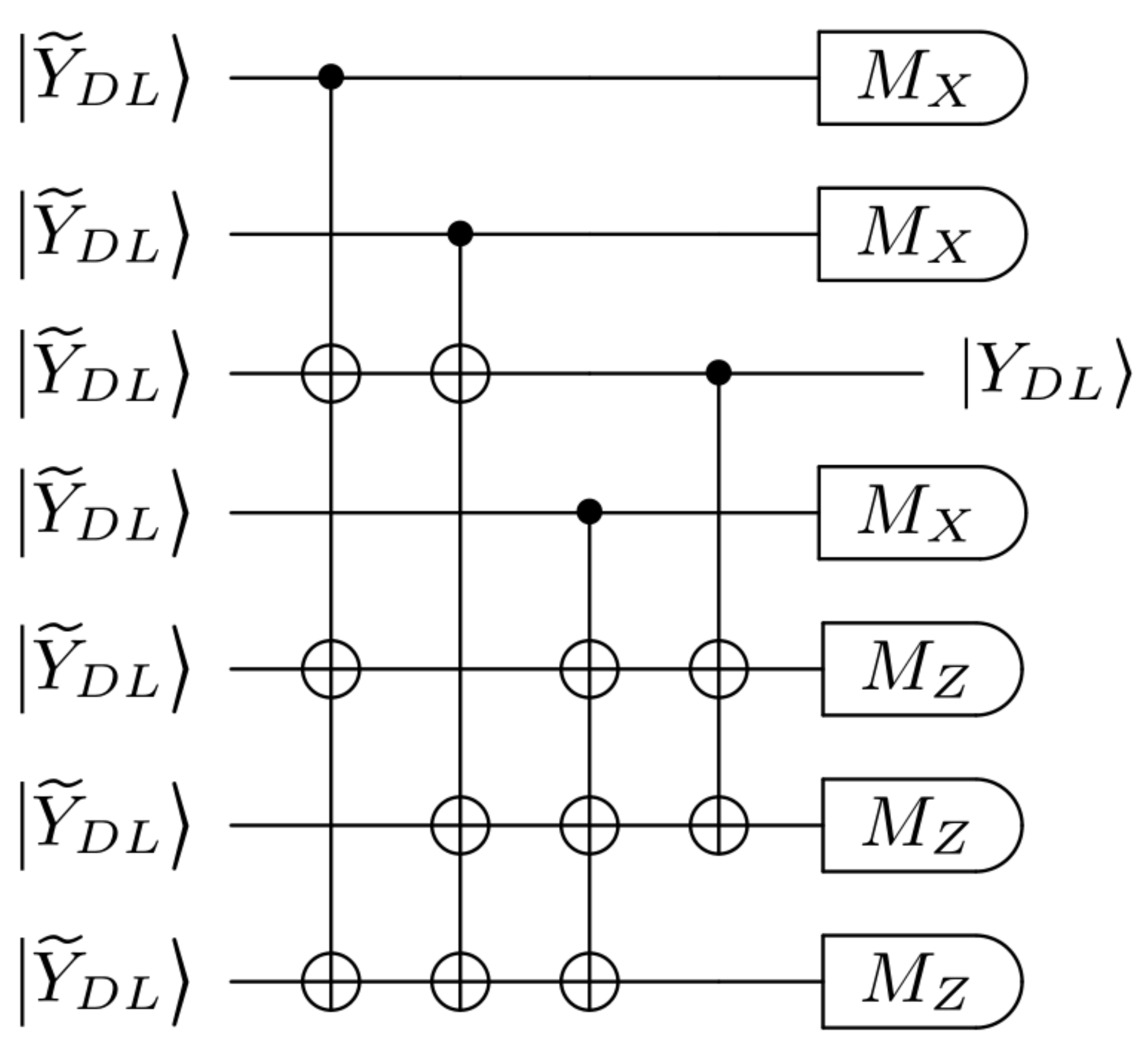}
\caption{\label{fig:Y_distillation} Circuits for logical $|Y\>$ distillation from imperfect $|\widetilde{Y}\>$ states.}
\end{figure}
%\begin{figure}[!ht]
%\begin{center}$
%\Qcircuit @C=1em @R=0.9em{
%\lstick{\ket{\widetilde{A}_{DL}}} & \ctrl{2} & \qw &\qw &\qw &\qw &\measureD{M_X}\\
%\lstick{\ket{\widetilde{A}_{DL}}} & \qw &\ctrl{1} &\qw &\qw &\qw &\measureD{M_X}\\
%\lstick{\ket{\widetilde{A}_{DL}}} & \targ \qwx[2] &\targ \qwx[3] & \qw &\qw &\ctrl{2}&\rstick{\ket{A_{DL}}}\qw\\
%\lstick{\ket{\widetilde{A}_{DL}}} & \qw & \qw & \ctrl{1} &\qw &\qw&\measureD{M_X}\\
%\lstick{\ket{\widetilde{A}_{DL}}} & \targ \qwx[2]  &\qw & \targ\qwx[1] &\qw &\targ\qwx[1]&\measureD{M_Z}\\
%\lstick{\ket{\widetilde{A}_{DL}}} & \qw & \targ\qwx[1] & \targ\qwx[1] &\qw &\targ\qwx[3]&\measureD{M_Z}\\
%\lstick{\ket{\widetilde{A}_{DL}}} & \targ \qwx[2]  &\targ\qwx[3] & \targ\qwx[5] &\qw &\qw&\measureD{M_Z}\\
%\lstick{\ket{\widetilde{A}_{DL}}} & \qw & \qw & \qw &\ctrl{1} &\qw&\measureD{M_X}\\
%\lstick{\ket{\widetilde{A}_{DL}}} & \targ \qwx[2]  &\qw & \qw &\targ\qwx[1] &\targ\qwx[1]&\measureD{M_Z}\\
%\lstick{\ket{\widetilde{A}_{DL}}} & \qw & \targ\qwx[1] & \qw &\targ\qwx[1] &\targ\qwx[2]&\measureD{M_Z}\\
%\lstick{\ket{\widetilde{A}_{DL}}} & \targ\qwx[2] & \targ\qwx[3] & \qw &\targ\qwx[1] & \qw&\measureD{M_Z}\\
%\lstick{\ket{\widetilde{A}_{DL}}} & \qw & \qw & \targ\qwx[1] &\targ\qwx[1] &\targ\qwx[3]&\measureD{M_Z}\\
%\lstick{\ket{\widetilde{A}_{DL}}} & \targ\qwx[2] & \qw &\targ\qwx[1] &\targ\qwx[1] &\qw&\measureD{M_Z}\\
%\lstick{\ket{\widetilde{A}_{DL}}} & \qw & \targ\qwx[1] &\targ\qwx[1] &\targ\qwx[1] &\qw&\measureD{M_Z}\\
%\lstick{\ket{\widetilde{A}_{DL}}} & \targ & \targ &\targ &\targ &\targ &\measureD{M_Z}
%}$
\begin{figure}[!ht]
\centering\includegraphics[width=58mm]{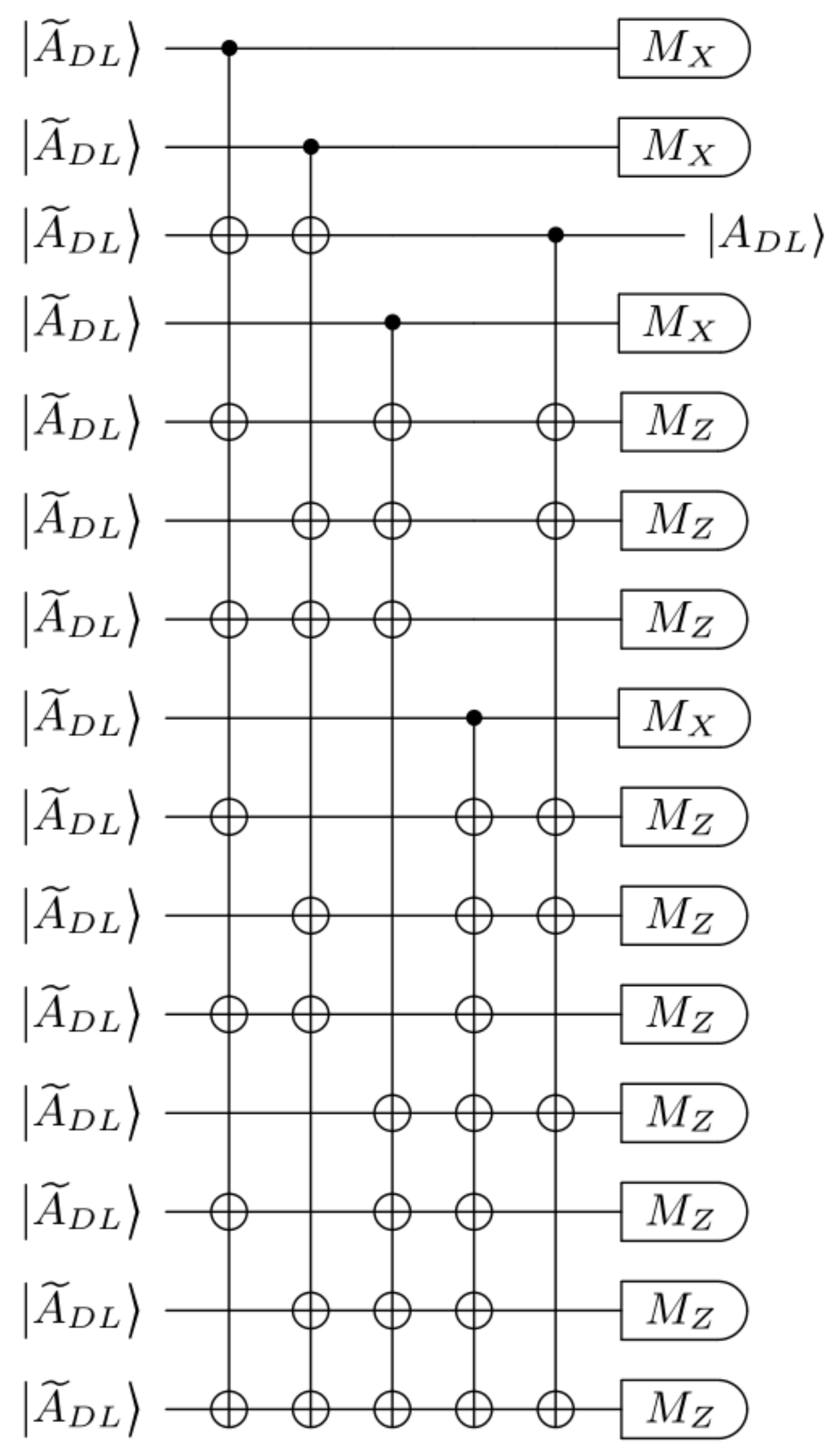}
\caption{\label{fig:A_distillation} Circuits for logical $|A\>$ distillation from imperfect $|\widetilde{A}\>$ states.}
\end{figure}
Repeating this process multiple times, arbitrarily high fidelity $|Y\>$ states can be obtained exponentially quickly if the original fidelity of the input states is higher than some threshold~\cite{reichardt2005quantum}.
A similar distillation circuit exists for the $|A\>$ state, as shown in Fig.~\ref{fig:A_distillation},
which is the reverse of the encoding circuit for the [[$15,1,3$]] truncated Reed-Muller code~\cite{Raussendorf:2006:2242,Fowler:2009:052312}. As before, given a good enough input $|A\>$ state, the convergence is rapid.

Note that these distillation circuits use CNOTs between the same type of qubits, and both types of logical state measurements described in Sec.~\ref{sec:easy_measurement} and ~\ref{sec:difficult_measurement}. If the input states are $X$-cut qubits, then the logical $X$ measurements are of the second kind, and the states after measurement are $|+_{DL}^X\>$ or $|-_{DL}^X\>$, which are ready to be reused to inject $|Y\>$ or $|A\>$ for future state distillation. The logical $Z$ measurements are of the first type, and will prepare logical states $|0_{DL}^X\>$ or $|1_{DL}^X\>$. To recycle these logical qubits to inject new $|Y\>$ or $|A\>$, we need to reset them to $|+_{DL}^X\>$ or $|-_{DL}^X\>$, which can be done by a subsequent logical $X$ measurement:
\begin{center}$
\Qcircuit @C=1em @R=1em {
\lstick{\ket{\psi^X_{DL}}}&\qw &\measureD{M_{Z}} &\targ & \qw &\rstick{|+_{DL}^X\>\text{ or }|-_{DL}^X\>}\qw \\
\lstick{\ket{+_{DL}^{Z}}} &\qw &\qw &\ctrl{-1} & \qw &\measureD{M_{X}}
}$
\end{center}
Note that the ancilla states $|+_{DL}^Z\>$ or $|-_{DL}^Z\>$ introduced here after logical $X$-measurement can also be reused directly as ancilla for another logical $X$-measurement. The recycling process for $Z$-cut qubit inputs is similar.

%Similarly, logical $Z$ measurement is of the second kind and needs ancilla $|0_{DL}^X\>$
%to implement the measurement:
%\begin{center}$
%\Qcircuit @C=1em @R=1em {
%\lstick{\ket{\psi^Z_{DL}}}&\qw & \ctrl{1} & \qw & \qw \\
%\lstick{\ket{0_{DL}^{X}}} &\qw & \targ & \qw &\measureD{M_{Z}}
%}$
%\end{center}
%and prepare $|0_{DL}^Z\>$ or $|1_{DL}^Z\>$ after the measurement. After a round of state distillation, the measured logical qubits prepared in $|0^Z_{DL}\>$ can be directly recycled to inject new logical $|Y_{DL}^Z\>$ and $|A_{DL}^Z\>$ states used for distillation. Similarly the ancilla used to do logical $X$ measurement is prepared in $|0_{DL}^X\>$ and can be recycled to use for next round logical $X$ measurement.

\subsection{Logical Phase and \emph{T} Gates}\label{sec:S_T}
Given the distilled $|Y\>$ state, we can implement high quality logical $S$ gates and logical $R^X_L(\pi/2)=\exp\left(-i\frac{\pi}{4} X_L\right)$ gates using the following circuits~\cite{Fowler:2009:052312}:
\begin{center}$
\Qcircuit @C=1em @R=1em  {
\lstick{\ket{Y_{DL}}}&\qw  & \ctrl{1}   & \qw &\gate{Z_LX_L}\cwx[1] & \rstick{S_L|\psi_{DL}\>}\qw \\
\lstick{\ket{\psi_{DL}}}&\qw  & \targ  & \measureD{M_Z} &\control\cw\cwx
}$
\end{center}
\begin{center}$
\Qcircuit @C=1em @R=1em  {
\lstick{\ket{Y_{DL}}}&\qw  & \targ &\qw   &\gate{Z_LX_L}\cwx[1] & \rstick{R^X_L(\pi/2)|\psi_{DL}\>}\qw \\
\lstick{\ket{\psi_{DL}}}&\qw  & \ctrl{-1}  & \measureD{M_X}&\control\cw\cwx
}$
\end{center}
If the measurement outcome is $+1$, nothing needs to be done; otherwise, do a $ZX$ gate. Note that this $ZX$ gate can be done in ``software" rather than physically.

The non-Clifford gates play a central role in quantum speedup~\cite{Gottesman:9705052}, and are necessary to obtain a universal gate set. For the surface code, the logical $T$ gate is implemented with high quality distilled logical $|A\>$ states using this circuit~\cite{Folwer2012PhysRevA.86.032324}:
\begin{center}
$\Qcircuit @C=1em @R=1em  {
\lstick{\ket{A_{DL}}} &\qw &\ctrl{1}   &\qw &\gate{Z_LX_LS_L} \cwx[1]& \rstick{T_L|\psi_{DL}\>}\qw \\
\lstick{\ket{\psi_{DL}}} &\qw & \targ  & \measureD{M_Z}&\control\cw\cwx
}$
\end{center}
If the logical $Z$ measurement yields a $+1$ outcome, the output state is the desired one. If the measurement yields a $-1$ outcome, the output is $X_LT^\dag_L |\psi_{DL}\>$ and $Z_LX_LS_L$ needs to be applied to get $T_L$.
Again, the logical $X$ and $Z$ gate can be done in classical ``software"  rather than physically. Details of commuting $X_L$, $Z_L$ through $S_L$ and $T_L$ for classical software control were discussed in Sec.XVI.A of Ref.~\cite{Folwer2012PhysRevA.86.032324}. As usual, the states after the measurements in these circuits can all be recycled and used as ancillas for logical CNOT gates, state injection and state distillation in future computational steps.

\subsection{Hadamard}\label{sec:Hadamard}
In the existing, measurement-based QC on the surface code, a logical Hadamard is realized by first digging a ``moat" around the double logical qubits by measuring single qubits around the double hole to create a logical qubit island. On the ``island", a logical Hadamard gate is then realized by a sequence of code deformations through single qubit and stabilizer measurements, and then the ``moat" at last is repaired ~\cite{Folwer2012PhysRevA.86.032324}. This version of logical Hadamard is easy and efficient enough in measurement-based QC, but difficult to implement in our system when the stabilizer Hamiltonian is turned on.
Instead, the logical Hadamard gate can be done directly:
\beq
\text{Had}=S\cdot R^X(\pi/2) \cdot S.
\eeq
Both logical $S$ and $R^X(\pi/2)$ are fault-tolerant but heavily rely on the state distillation of logical $|Y\>$ state.

There is a more efficient way to do a logical Hadamard, as illustrated in Ref.~\cite{Bombin:2010:30403}, by introducing a nontrivial domain wall on the lattice and moving the holes across the wall. The wall can be created by shifting the geometry of the lattice along a line, as shown in Fig.~\ref{fig:twist}. The five body interaction terms terminating the dislocation are called twists~\cite{Bombin:2010:30403}. One can see that the insertion of two twists changes the degeneracy of ground space. This can form an additional logical qubit, which we call gauge qubit $\mathcal{F}$. The corresponding logical operators of this qubit are also shown in Fig.~\ref{fig:twist}.
\begin{figure}[!ht]
\centering\includegraphics[width=92mm]{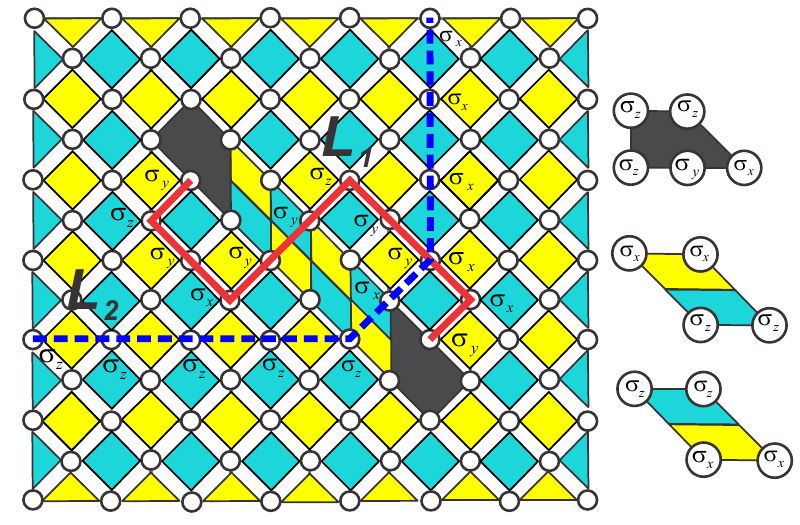}
\caption{\label{fig:twist} (Color online) A dislocation in the geometry of the Hamiltonian produced by shifting the stabilizer generators along a line between two twists. The stabilizer generators corresponding to two different parallelograms (yellow/cyan and cyan/yellow) and a pentagon (dark gray) are shown on the right side. A pair of anticommuting strings of Pauli operators $L_1$ (solid red) and $L_2$ (dashed blue) that commute with all stabilizer generators forms the logical operators of the extra qubit $\mathcal{F}$ attached to the pair of twists.}
\end{figure}

If a single $Z$ ($X$)-cut hole is adiabatically dragged across the wall, it will change to a $X$ ($Z$)-cut hole, as shown in Fig.~\ref{fig:logical_hadamard}. However, note that this process can also change the state of $\mathcal{F}$, since it will change logical operators $L_1$ and $L_2$. This effect in general will yield additional entanglement between data qubit and $\mathcal{F}$.
However, if we drag the second hole of the logical data qubit across the wall, it will reverse the change caused by the first hole and leave the state of $\mathcal{F}$ unchanged.
In summary, adiabatically moving two holes of a logical qubit across the wall will give a state transformation on the data qubit:
\beq
|\psi^Z_{DL}\>\rightarrow \text{Had}\ |\psi^X_{DL}\>,\ \ \ \ \ |\psi^X_{DL}\>\rightarrow \text{Had}\ |\psi^Z_{DL}\>,
\eeq
\begin{figure}[!ht]
\centering\includegraphics[width=90mm]{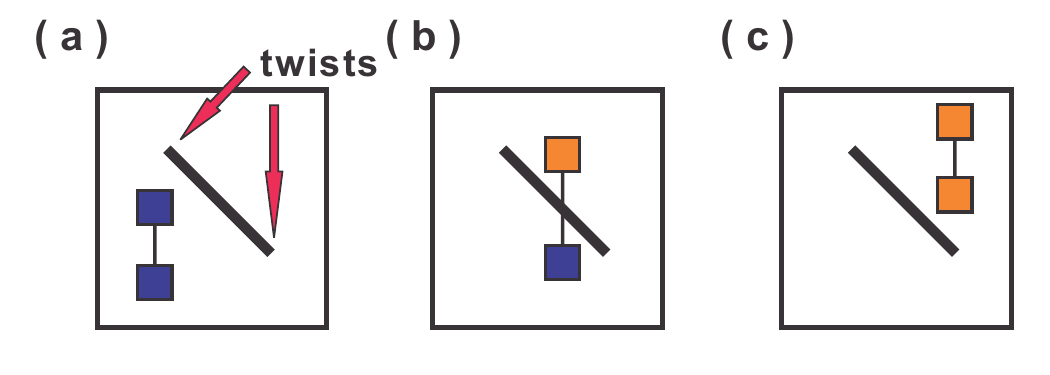}
\caption{\label{fig:logical_hadamard} (Color online) Adiabatically moving a pairs of holes of a $Z$-cut qubit (dark blue holes) across a twist on the surface to get a logical Hadamard gate. This process will transform a $Z$-cut qubit to an $X$-cut qubit (orange holes). }
\end{figure}
\begin{table*}[!ht]
\begin{center}
\begin{tabular}{c|c|c|c|c}
  \hline
  \hline
  % after \\: \hline or \cline{col1-col2} \cline{col3-col4} ...
  {\bf Process} &  {\bf Gap protection} & {\bf Fault-tolerance} & {\bf Dynamics}& {\bf Number of time steps}  \\[2pt]
  \hline
  Creation $|0\>\ (|+\>)$ for $Z\ (X)$-cut qubit& Yes & Yes  & Adiabatic & $\sim d/2$ \\[2pt]
  \hline
  Creation $|0\>\ (|+\>)$ for $X\ (Z)$-cut qubit& Yes & Yes & Adiabatic+Measurement & $\sim d$ \\[2pt]
  \hline
  $Z\ (X)$ measurement for $X\ (Z)$-cut qubit& Yes & Yes  &Adiabatic+Measurement & $\sim d$\\[2pt]
  \hline
  $Z\ (X)$ measurement for $Z\ (X)$-cut qubit& Yes & Yes &Adiabatic+Measurement  & $O(d)$\\[2pt]
  \hline
  Hole enlargement & Yes & No & Adiabatic & $\sim d/2$  \\[2pt]
  \hline
  Hole movement & Yes & Yes & Adiabatic& N/A\\[2pt]
  \hline
  Logical CNOT & Yes & Yes & Adiabatic+Measurement & $O(d)$ \\[2pt]
  \hline
  State injection & Yes & No & Adiabatic+ Pulse control  & $\sim d$ \\[2pt]
  \hline
  State distillation & Yes & Yes & Adiabatic+Measurement & N/A  \\[2pt]
  \hline
  Logical $S$, $T$, Hadamard & Yes & Yes & Adiabatic+Measurement &  N/A\\[2pt]
  \hline
  \hline
\end{tabular}
\caption{Summary.}\label{table_summary}
\end{center}
\end{table*}
for $Z$-cut qubits and $X$-cut qubits. Another problem of this method is that it will change of the type of qubits we are working on. However, we can use an ancilla to swap the data qubit back by the circuit
\begin{center}$
\Qcircuit @C=1em @R=1em {
\lstick{\ket{\psi^Z_{DL}}}&\qw &\ctrl{1} &\targ & \ctrl{1} &\qw &\rstick{|0_{DL}^Z\>\ \text{or}\ \ket{+_{DL}^Z}}\qw \\
\lstick{\ket{0_{DL}^{X}}\ \text{or}\ \ket{+_{DL}^X}} &\qw &\targ &\ctrl{-1} & \targ &\qw &\rstick{|\psi_{DL}^X\>}\qw
}$
\end{center}
for a $Z$-cut qubit, and
\begin{center}$
\Qcircuit @C=1em @R=1em {
\lstick{\ket{\psi^X_{DL}}}&\qw &\targ &\ctrl{1} & \targ &\qw &\rstick{\ket{0_{DL}^X}\ \text{or}\ \ket{+_{DL}^X}}\qw \\
\lstick{\ket{0_{DL}^Z}\ \text{or}\ \ket{+_{DL}^{Z}}} &\qw &\ctrl{-1} &\targ & \ctrl{-1}& \qw &\rstick{|\psi_{DL}^Z\>}\qw
}$
\end{center}
for an $X$-cut qubit. The position of the twists can be fixed on the lattice so that they
can be used repeatedly for Hadamard gates.
%The details of twists on the surface and adiabatically dragging holes across the twists can be
%found in Ref.~\cite{Bombin:2010:30403}, but are beyond the scope of this paper.

%In Ref.~\cite{Bombin:2010:30403}

%The time spent on state distillation is about 90\%. What we do at worst will increase the computational time by a factor of 2, like 95\% of computational time is spent on the state distillation.

\section{Fault-tolerance of the scheme}\label{sec:fault-tolerance}
We have described a way to fault-tolerantly implement QC in surface codes with a constant energy gap to suppress errors in a thermal environment. Table.~\ref{table_summary} lists a summary of each procedure. Note that although adiabatic hole enlargement and state injection are not themselves fault-tolerant, they do not affect the fault-tolerance of the whole QC scheme. In addition to gap protection during the computation, fault-tolerance is guaranteed by performing single qubit and syndrome measurements before errors can propagate to become uncorrectable. We discuss the interval betweens syndrome measurements in Sec.~\ref{sec:period}.

So far, the error models we considered are induced by weak coupling to a thermal bath. We also need to consider other decoherence channels, which may affects qubits collectively or directly act on logical qubits. In this section, we will discuss two of them: local perturbations and adiabatic errors. In the following sections we show that they can both be exponentially bounded.

\subsection{Error correction}\label{sec:period}
%The constant energy gap may not be able to protect the quantum state in ground space for the time period of whole computation. If errors are detected by syndrome measurement and decoding,
%they needs to be corrected and the stabilizer Hamiltonian can be turned on again.
A proper time period to turn off the system Hamiltonian and do error correction, in the case that there are no errors detected during the adiabatic hole movement process, is crucially important. We assume that syndrome measurement is done every $m$ time steps, and
$m\exp(-2c\beta J)$ can be regarded as the error rate on each qubit for every $m$ time steps ($m\exp(-2c\beta J)\ll 1$), since all processes necessary for universal QC are protected by a gap of at least $2J$. Besides thermal errors accumulating on each qubit, the following types of physical errors can occur in a single syndrome measurement cycle in Sec.~\ref{sec:surface_code}~\cite{Folwer2012PhysRevA.86.032324}:
\begin{enumerate}
  \item $\sigma_x$ error occurs when a syndrome qubit is initialized to $|0\>$, with  probability $p$.
  \item The Hardamard gate on syndrome qubit is not perfect. There is extra $\sigma_x$, $\sigma_y$ or $\sigma_z$ error following the gate, each with probability $p/3$.
  \item Error occurs when a syndrome qubit is measured, with probability $p$.
  \item CNOT gate on syndrome qubit-data qubit CNOT is not perfect, but with following erros: $I\otimes\sigma_x$, $I\otimes \sigma_y$, $I\otimes\sigma_z$, $\sigma_x\otimes I$, $\sigma_x\otimes\sigma_x$, $\sigma_x\otimes\sigma_y$, $\sigma_x\otimes\sigma_z$, $\sigma_y\otimes I$, $\sigma_y\otimes\sigma_x$, $\sigma_y\otimes\sigma_y$, $\sigma_y\otimes\sigma_z$, $\sigma_z\otimes I$, $\sigma_z\otimes \sigma_x$, $\sigma_z\otimes\sigma_y$ or $\sigma_z\otimes\sigma_z$, each with probability $p/15$.
\end{enumerate}
Note that one needs several cycles of syndrome measurements to establish values of syndrome before actual decoding. Then, the logical error rate of surface code for $m$ time steps with active error correction can be roughly estimated as~\cite{Folwer2012PhysRevA.86.032324}
\beq
P_L^m\approx d\frac{d!}{(d_e-1)!d_e!}\left(me^{-2c\beta J}+7p\right)^{d_e},
\eeq
where $d_e=(d+1)/2$.
A plot of this estimate is shown in Fig.~\ref{fig:logical_error_rate}, for various values of $c\beta J$ , $p$ and $m$.
\begin{figure}[!ht]
\centering\includegraphics[height=60mm,width=90mm]{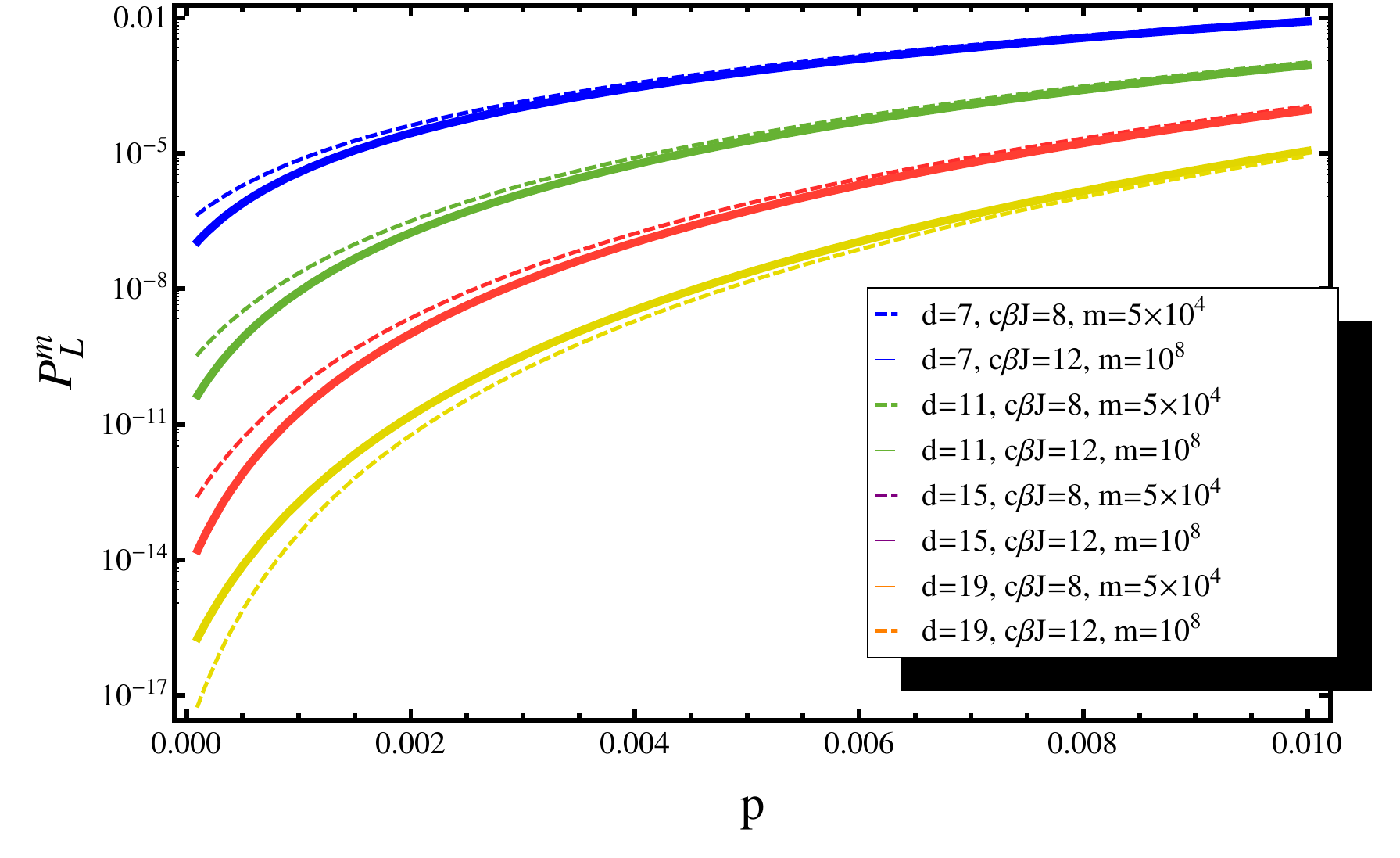}
\caption{\label{fig:logical_error_rate} (Color online) Logical error rate per $m$ time steps for various values of $m$ and $d$. The dashed lines are for $c\beta J=8$ and solid lines for $c\beta J=12$. The blue (top), green (second top), red (third) and yellow (bottom) lines are for $d=7$, $d=11$, $d=15$ and $d=19$, respectively.
}
\end{figure}
We can use these scaling relations to estimate the number of qubits needed to obtain a desired error rate after error correction.
Our goal is that the error rate after the whole computer procedure is bounded by some particular value $\delta \ll 1$. Denoted by $M$ the product of number of logical operation and the number of logical qubits used in an algorithm. We need to have:
\beq\label{eq:logical_error_condition}
P_L^m\lesssim\frac{m\delta}{dM},
\eeq
since each logical operation needs about $d$ time steps in our scheme. For a particular computation like Shor's algorithm implemented on surface codes, $M$ is of the order larger than $10^{14}$~\cite{Raussendorf:2007:199}. We can choose  $p=0.001$ and $c\beta J=12$, which may be achievable in current experiments. Also, set $\delta=0.1$, $d=11$ and $m=10^{8}$, then we have $P_L^m\approx 10^{-8}$, which satisfies the condition of Eq.~(\ref{eq:logical_error_condition}). This requires a number of data and measurement qubits $n_{\text{tot}}=(2d-1)^2\approx 450$ to protect a logical qubit, and perform Shor's algorithm with reasonable success probability. We can see that if large $c\beta J$ is not achievable,
one can always choose a code with larger distance and more frequent error correction to compensate for the small $c\beta J$.
However, if the $c\beta J$ can increase to $15$, we can even reduce $d$ to 7 and $n_q$ to about 170, with $m=10^{10}$ and same value of $\delta$, making it more efficient to build a scalable QC in the near future.

%\beq
%\tau\sim\mathcal{O}\left(e^{-2J/T}
%\left(\frac{\delta(d_e-1)!d_e!}{(d+1)!}\right)^{1/d_e}\right).
%\eeq

%\subsection{Error model}
%\cite{Leggett:1987:1}\cite{DivincenzoPhysRevB.71.035318}
%\beq
%\gamma(\omega)=2\kappa_n\bigg|\frac{\omega^n}{1-e^{-\beta\omega}}\bigg|e^{-|\omega|/\omega_c},
%\eeq
%Here, $\beta=1/T$, with $T$ being the temperature of the bath (we set Boltzmann's constant to one). For simplicity we assume in the following a large cut-off energy $\omega_c\rightarrow\infty$. For $n=1$, the bath is called ``Ohmic", whereas for $n\neq2$ it is called ``super-Ohmic".

\subsection{Local perturbation}
Perturbations will split the degeneracy of the ground space and cause stochastic phase errors between different logical states. This is one of the main obstacles to realizing non-Abelian holonomic quantum gates on system with a small number of qubits. However, for surface codes, the splitting of the ground space (and any other error space) caused by local perturbations will decay exponentially with the distance of the surface code, as shown by Kitaev in Ref.~\cite{Kitaev:2003:2}. Actually, any system with quantum topological order is in general stable under local perturbations~\cite{Bravyi:2010:093512}. This might suggest that holonomic QC is more naturally suitable with systems with topological order than systems with small number of qubits. Consider a local perturbation of the general form:
\beq
V_{\text{local}}=-\sum_jh_j\vec{\sigma}_j-\sum_{j <p}J_{jp}(\vec{\sigma}_j,\vec{\sigma}_p),
\eeq
which includes all one-qubit and two-qubit interactions. The effect of $V_{\text{local}}$ only occurs in the $d/2-$th order of perturbation theory, and the energy splitting vanishes as
\beq
\Delta_{\text{split}}\sim O\left(Je^{-vd/2}\right),
\eeq
where $v=\min_{ij} \{\ln(J/|h_i|), \ln(J/\|J_{ij}\|_1)\}$, which decreases quickly with growth of the code distance. Consider the case when $d=11$, $J=1$. To achieve an error rate of order $10^{-15}$, we must to control the values of $|h_i|/J$ and $\|J_{ij}\|_1/J$ so that they are less than $10^{-3}$, which is practically achievable for current or near future technology.

\subsection{Adiabatic error}
Another type of error corresponds to imperfect adiabatic evolution. We now discuss adiabatic theorem briefly and their application to bound the corresponding error. The traditional version of the adiabatic theorem stated in ~\cite{Messiah:1965:North} says that the adiabatic approximation is satisfied with precision $\delta\leq\epsilon^2$ during adiabatic evolution if the condition
\beq
\frac{\sup_{t\in[0,T]}\parallel P_{\textbf{s}_\alpha}(t)\frac{\partial}{\partial t}H(t)P_{\textbf{s}_\beta}(t)\parallel_1}{\inf_{t\in[0,T]}K\left(\varepsilon_{\textbf{s}_\alpha}(t)
-\varepsilon_{\textbf{s}_\beta}(t)\right)^2}\leq \epsilon, \ \text{for any}\ \alpha\neq\beta,
\eeq
is satisfied ($K$ is the dimension of the code space). In the case of our adiabatic process, this is equivalent to
\beq
\sup_{q,t\in[t_{q-1},t_q]}\frac{\pi|\partial_t{f}_{q}(t)|}{4}\leq\epsilon
\eeq
for the $q$th time segment.
However, it is known that this statement is neither sufficient nor necessary, and we can obtain better results~\cite{Hagedorn:2002:235,lidarAdiabaticaccuracy:102106}.
Here we apply the result in~\cite{lidarAdiabaticaccuracy:102106} to our piecewise adiabatic evolution, serial or parallel, as described in Sec.~\ref{sec:sketch_scheme}, for the $q$th time segment. We can set $T_q=t_{q}-t_{q-1}$, for a Hamiltonian $H(\vartheta)$($\vartheta=t/T_q$) that is analytic near the region $[0,1]$ in the complex plane, with the absolute value of the imaginary part of the nearest pole being $\gamma$, and the first $\mathcal{N}\geq1$ derivatives at boundaries equal to zero, i.e., $H^{(l)}(0)=H^{(l)}(1)=0$ for $l\leq\mathcal{N}$. If we set
\beq\label{eq:evolution_time}
T_q=\frac{e}{\gamma}\mathcal{N}\frac{\xi_q^2}{\Delta_{\text{min}}^3},
\eeq
with $\xi_q=\sup_{\vartheta\in[0,1]}\parallel\text{d}H/\text{d}\vartheta\parallel_\infty$
(where the $\parallel\cdot\parallel_\infty$ is standard operator norm, and $\Delta_{\text{min}}=2J$), then the adiabatic approximation error satisfies
\beq
\delta_{\text{ad}}\leq(\mathcal{N}+1)^{\gamma+1}e^{-\mathcal{N}},
\eeq
or equivalently,
\beq
\delta_{\text{ad}}\lesssim (c_qT_q+1)^{\gamma+1}e^{-c_qT_q},
\eeq
with $c_q=\frac{\gamma \Delta_{\text{min}}^3}{e\xi_q^2}$. In other words, we can decrease the adiabatic error exponentially with evolution time $T_q$, if it is carefully set to be proportional to $\mathcal{N}$ and $f_q(t)$ is chosen such that a) the boundary condition mentioned above is satisfied, and b) $H(\vartheta)$ is analytic near region $[0,1]$ on the complex plane. The adiabatic error for typical processes listed in Table.~\ref{table_summary} can then be bounded by
\beq
\delta_{\text{ad}}\sim O\left( d\cdot\sup_ q(c_qT_q+1)^{\gamma+1}e^{-c_qT_q}\right).
\eeq
So in principle, we can make adiabatic process arbitrarily small with careful chosen $\{T_q\}$ and $\{f_q\}$. Note that the thermal error rate decreases exponentially with $J$, while the during of each adiabatic time segment decreases as the cube of $J$ at fixed temperature, so the processing time overhead of an adiabatic process can be small if $J$ is large.
\begin{remark}
We've analyzed that it is possible to use on the order of $10^2$ physical qubits to protect a single logical qubit in practical quantum computation with protection by a constant gap enabling fault-tolerant QC in surface codes. This is quite efficient compared to the existing QC scheme in surface codes~\cite{Folwer2012PhysRevA.86.032324}. However, the assumption here is that the thermal error model is local, and the stabilizer Hamiltonian is fundamental, given by Nature. Such 4-body $X_s$ and $Z_p$ interactions are hard to build directly, and usually needs certain techniques, like quantum gadgets~\cite{KempeQuantumGadget,Oliveira:2008:CQS:2016985.2016987}, digital quantum simulator~\cite{RydbergsimulatorNatPhysics,weimer2011digital_rydberg_simulator}
, the low energy approximation from Kitaev's honey-comb model~\cite{kitaev2006anyons} or to be generated dynamically~\cite{becker2013dynamic_self_correction}. If the Hamiltonian is effective, rather than being fundamental, it may dramatically change the local thermal error model we have assumed, and cause nonlocal errors. This possibility calls for future investigation.
\end{remark}

%\cite{lidarAdiabaticaccuracy:102106}
%\cite{Yi-Cong_PhysRevA.89.032317}
%
%\beq
%T=\frac{e}{\gamma}\mathcal{N}\frac{\xi^2}{\Delta_{\min}^3},
%\eeq
%
%\beq
%\delta_{\text{ad}}\lesssim (c_2T+1)^{\gamma+1}e^{-c_2 T}
%\eeq
%
%adiabatic error may be corrected by usual error correction procedure.

\section{Summary and conclusion}\label{sec:summary}
We have outlined a scheme for fault-tolerant universal HQC based on surface codes, with stabilizer Hamiltonian to protect quantum information encoded in the degenerate ground space, from both thermal errors and small perturbations. We explicitly constructed all necessary processes with energy gap protection and parallel operations. These processes include logical state creation, a logical universal gate set, and logical state measurement. Logical state initialization and measurement are realized by open-loop adiabatic evolution and measurements on single qubits compatible with system Hamiltonian, while the logical CNOT is implemented by a closed-loop holonomic operation. All other logical gates can be implemented using the logical CNOT, logical state preparation, and logical state measurement. It is worth mentioning that if a
twist is allowed to exist on the surface, the logical Hadamard can be done much more efficiently. Conditions for active error correction are also discussed. The number of physical qubits needed to protect a logical qubit for fault-tolerant QC can reduce to the order of $10^2$, if large coupling constant $J$ and low temperature are achievable in experiment.

Theoretical and experimental progress in non-Abelian HQC for single-qubit operations has been made recently, through both
adiabatic~\cite{HQC_adiabatic_realization2013} and
non-adiabatic evolution~\cite{NJPNonAdiabaticHolonomic,NonAbelianNonAdiabaticHolonomicNature, FengGuanruPhysRevLett.110.190501,Luming2014HQCexperimental} on various of physical systems. Applying our scheme to an actual physical system needs local 4-body interactions. Several theoretical proposals have
been proposed to build such interactions effectively, which include low energy perturbations~\cite{kitaev2006anyons,KempeQuantumGadget,Oliveira:2008:CQS:2016985.2016987}
of systems with strong two body interactions, and dynamic simulation
~\cite{RydbergsimulatorNatPhysics,weimer2011digital_rydberg_simulator,becker2013dynamic_self_correction}.
As argued in Sec.~\ref{sec:fault-tolerance}, the effect of such
effective interaction on local error models needs further study. It is important to find out under what conditions these effective Hamiltonians behave like the ideal ones in \emph{open} quantum systems.

We concentrated on surface codes
in this paper, but we hope the methods can be extended to fault-tolerant QC schemes with constant gap protection on other topological codes, including color codes~\cite{Bombin:2006:180501,landahl2011fault_color} and Turaev-Viro codes~\cite{koenig2010quantumUTQC}.

Another interesting question is, could it be possible to do QC fault-tolerantly on an arbitrarily large scale without any active error correction? It has been shown that it is possible to do so with  6\textsc{D} topological color codes~\cite{bombin2013self-correctingQC}. In our scheme on a 2D lattice, if $J$ is very large and the temperature is sufficiently low (which is certainly a challenging engineering problem), then for practical algorithm, it may not be necessary to do active error correction. It has also been shown that a self-correcting quantum memory to store quantum information for a polynomially (or even exponentially) long time in the lattice size exists, if long range interactions between anyons is allowed~\cite{Hamma_self_correction_PhysRevB.79.245122,chesi2010self-correciton_memory,
PedrocchiMemorycoupledtocavitymodePhysRevB.83.115415,hutter2012self-correction_memory,
wootton2013topological_self_correcting}. Theoretical work
to realize such a long range interaction was also proposed
in~\cite{self_correction_in3DPhysRevA.88.062313,becker2013dynamic_self_correction}.
Long range interaction can freeze the density of excited anyons on the lattice for such a long time that logical errors are quite unlikely to happen. One may ask whether such interactions can be allowed when we adiabatically deform the stabilizer Hamiltonian in our scheme. One difficulty here is that, when enlarging or moving the holes, it is hard to define the concept of anyons on the boundaries of the holes. How to introduce similar long range interactions during hole movement and enlargement is an interesting problem, and if it is possible, one may be able to implement self-correcting QC on a 2D lattice.

\vspace{2mm}
\noindent\emph{Addendum:} When writing this manuscript, we note that Cesare, Landahl, Bacon, Flammia and Neels have published a manuscript~\cite{adiabatic_topologic_qc1406.2690} with the idea of implementing adiabatic TQC. There is a similarity of underlying spirit for both schemes: protecting quantum information with a constant energy gap during the process of quantum computation on topological codes. However, they differ a great deal in how they implement logical state preparation, measurement, state injection and uses of logical ancilla states. Also, we don't restrict ourselves to adiabatic process.
Finally, we analyze the errors carefully to establish the fault-tolerance of our scheme.

%In addition, we will show that during the entire information process procedure, including logical state initialization, logical state measurement, logical gate, state injection and distillation, quantum information is protected from local thermal excitations by a constant gap and the weight of Hamiltonian can be bounded by 4 during the whole adiabatic code deformation process.  To deal with unwanted excitation caused by error (creation of anyons) during adiabatic code deformation, we analyze errors propagation and give conditions when turning off stabilizer Hamiltonian is necessary to do syndrome measurement and error correction. It can be shown that with with gap protection, frequency of error correction and physical resource needed are greatly reduced. We conclude that all computation procedure is scalable and the scheme is fault tolerant.

%\section{Physical implementation}

%\cite{dynamic_generator_hamiltonianPhysRevA.87.042340}

\section*{ACKNOWLDEGEMENT}
We would like to thank Ben Reichardt and Ching-Yi Lai for fruitful discussion of surface code and fault-tolerant quantum computation. This research was supported in part by ARO MURI Grant No. W911NF-11-0268, and by NSF Grants No. EMT-0829870 and No. TF-0830801.

\appendix
\numberwithin{equation}{section}
\section{Geometric Formulation of HQC}\label{sec:geometric_holonomic}
In this section, we introduce a more abstract geometric setting of holonomic problem which is useful to prove the results in Sec.~\ref{sec:sketch_scheme}. We focus on the ground space for simplicity, however, the formalism is general and can be applied to any eigenspace of system Hamiltonian.

Suppose we have a family of Hamiltonians acting on the Hilbert space $\mathbb{C}^N$, and the ground state of each Hamiltonian is $K$-fold degenerate ($K<N$). The natural mathematical setting to describe this system is the principal bundle $(S_{N,K}(\mathbb{C}), G_{N,K}(\mathbb{C}),\pi, {\rm U}(K))$, which consists of the Stiefel manifold $S_{N,K}(\mathbb{C})$, the Grassmann manifold $G_{N,K}(\mathbb{C})$, the projection map $\pi:S_{N,K}(\mathbb{C})\rightarrow G_{N,K}(\mathbb{C})$, and the unitary structure group ${\rm U}(K)$. We will explain the meaning of these mathematical objects in details below.

The Stiefel manifold is defined as:
\beq
S_{N,K}(\mathbb{C})=\{V\in M(N,K;\mathbb{C})|V^\dag V=I_K\},
\eeq
where $M(N,K;\mathbb{C})$ is the set of $N\times K$ complex matrices and $I_K$ is the $K-$dimensional unit matrix. Physically, each column of $V\in S_{N,K}(\mathbb{C})$ can be regarded as a normalized state in $\mathbb{C}^N$, and $V$ can be viewed as an orthonormal set of $K$ basis of the ground space of Hamiltonian:
\beq
V=\big\{|\varphi_1\>,|\varphi_2\>,\ldots , |\varphi_K\>\big\}.
\eeq
Note that we have freedom to transfer from one orthnormal basis of to another through unitary transformation, we can define a unitary group $\textrm{U}(K)$ that acts on $S_{N,K}(\mathbb{C})$ from the right:
\beq\label{eq:gauge_transform}
S_{N,K}(\mathbb{C})\times\textrm{U}(K)\rightarrow S_{N,K}(\mathbb{C}), \quad (V,h)\mapsto Vh,
\eeq
by the matrix product of $V$ and $h$. $V$ and $Vh$ can be regarded as two
different orthonormal basis corresponding to the same ground space.
%This right action is free, which means if there exists a point $V\in S_{N,K}(\mathbb{C})$ such that $Vh=V$, then $h=I_K$.

During the adiabatic evolution, the ground space of the Hamiltonian will change. The ground space can be represented as a $K$-dimensional hyperplane in $\mathbb{C}^N$. So we introduce the Grassmann manifold in $\mathbb{C}^N$:
\beq
G_{N,K}(\mathbb{C})=\{P\in M(N,N;\mathbb{C})|P^2=P,P^\dag=P,\text{Tr}P=K\},
\eeq
where $P$ is a projection operator onto the hyperplane in $\mathbb{C}^N$, and the condition $\text{Tr}P=K$ indicates that the dimension of the hyperplane is $K$. In our scenario, $P\in G_{N,K}(\mathbb{C})$ can be regarded as the projector onto the $K$-dimensional ground space of the Hamiltonian.

The relationship between the orthonormal basis $V$ and ground space $P$ can be seen as follows. We define the projection map $\pi:S_{N,K}(\mathbb{C})\rightarrow G_{N,K}(\mathbb{C})$ as
\beq
\pi:V\mapsto P:=VV^\dag.
\eeq
%It is easy to prove that the map $\pi$ is surjective, which means, for any $P\in G_{N,K}(\mathbb{C})$, we can always construct a orthonormal $K$-frame $V\in S_{N,K}(\mathbb{C})$ such that $\pi(V)=P$.
The corresponding ground space projector can be obtained when the orthonormal basis is given. We can see that the basis $V$ and basis $Vh$ with $h\in \text{U}(K)$ belong to the same ground space, since
\beq\label{eq:projection}
\pi(Vh)=(Vh)(Vh)^\dag=Vhh^\dag  V^\dag = VV^\dag = \pi(V).
\eeq
%Thus the Stiefel manifold $S_{N,K}(\mathbb{C})$ is a principal bundle over $G_{N,K}(\mathbb{C})$ with the structure group ${\rm U}(K)$.

%Physically, $V\in S_{N,K}(\mathbb{C})$ is corresponding a specified orthonormal  set of basis in the subspace $P\in G_{N,K}(\mathbb{C})$, and $h\in \text{U}(K)$ is corresponding to a unitary transformation between different sets of basis. Eq.~(\ref{eq:projection}) represents the fact that different basis connected by $h$ is in the same ground space.

For the purpose of the paper, we want to transform the ground space adiabatically during the procedure. To formulate such a process, we need also define the left action of the unitary group ${\rm U}(N)$ on both $S_{N,K}(\mathbb{C})$ and $G_{N,K}(\mathbb{C})$ by the matrix product:
\beq\label{eq:basis_transform}
{\rm U}(N)\times S_{N,K}(\mathbb{C})\rightarrow S_{N,K}(\mathbb{C}),\quad (g,V)\mapsto gV,
\eeq
and
\beq\label{eq:space_transform}
{\rm U}(N)\times G_{N,K}(\mathbb{C})\rightarrow G_{N,K}(\mathbb{C}),\quad (g,P)\mapsto gPg^\dag.
\eeq
It is easy to check that $\pi(gV)=g\pi(V)g^\dag$.
This action is transitive: there is a $g\in{\rm U}(N)$
for any $V,V^\prime\in S_{N,K}(\mathbb{C})$ such that $V^\prime=gV$.
There is also a $g\in{\rm U}(N)$ for
any $P,P^\prime\in G_{N,K}(\mathbb{C})$ such that $P^\prime=gPg^\dag$.
So this action is sufficient to describe any ground space transformation. This is why we choose to use the form of Hamiltonian deformation in Eq.~(\ref{eq:Hamiltonian_change}).

%The isotropy group of each point $V$ in $S_{N,K}(\mathbb{C})$ is defined as
%\beq
%I_S(V)=\{g\in {\rm U}(N)|gV=V\},
%\eeq
%which is easily verified to be isomorphic to ${\rm U}(N-K)$ for all $V\in S_{N,K}(\mathbb{C})$. Similarly, for each point $P$ in $G_{N,K}(\mathbb{C})$, the isotropy group is defined as
%\beq
%I_G(P)=\{g\in {\rm U}(N)|gPg^\dag=P\},
%\eeq
%which is isomorphic to ${\rm U}(K)\times {\rm U}(N-K)$ for all $P\in G_{N,K}(\mathbb{C})$.
%Thus, $S_{N,K}(\mathbb{C})\cong {\rm U}(N)/{\rm U}(N-K)$ and
%$G_{N,K}(\mathbb{C})\cong {\rm U}(N)/({\rm U}(K)\times {\rm U}(N-K))$ are homogeneous spaces~\cite{Nakahara:2003:IOP}.
%\beq
%\begin{split}
%\pi:&S_{N,k}(\mathbb{C})\cong{\rm U}(N)/{\rm U}(N-k)\rightarrow\\ &G_{N,k}(\mathbb{C})\cong{\rm U}(N)/({\rm U}(k)\times U(N-k))
%\end{split}
%\eeq
%is called a homogeneous bundle.
We can further study the topological structure of $S_{N,K}(\mathbb{C})$ and $G_{N,K}(\mathbb{C})$ for completeness. For each point $V$ in $S_{N,K}(\mathbb{C})$, we can define an isotropy group:
\beq
I_{S}(V)=\{g\in \text{U}(N)|gV=V\},
\eeq
which is isomorphic to $\text{U}(N-K)$ for all $V\in S_{N,K}(\mathbb{C})$. Similarly, we can define an isotropy group for each $P\in G_{N,K}(\mathbb{C})$:
\beq
I_G(P)=\{g\in \text{U}(N)|gPg^\dag = P\},
\eeq
which is isomorphic to $\text{U}(K)\times \text{U}(N-K)$ for all $P\in G_{N,K}(\mathbb{C})$. Thus, $S_{N,K}(\mathbb{C})\cong \text{U}(N)/\text{U}(N-K)$ and $G_{N,K}(\mathbb{C})\cong \text{U}(N)/\big(\text{U}(K)\times \text{U}(N-K)\big)$~\cite{Nakahara:2003:IOP}.

The canonical connection form on $S_{N,K}(\mathbb{C})$
is defined as a $\mathbbm{u}(K)$-valued one-form on $G_{N,K}(\mathbb{C})$:
%Let $\{\mathcal{U}_j\}$ be a an open covering of $G_{N,K}(\mathbb{C})$,
%and let $\sigma_j:\mathcal{U}_j\rightarrow S_{N,K}(\mathbb{C})$ be
%a local cross section:
%
%\beq
%\sigma_j:P\mapsto V^j(P),
%\eeq
%and the connection corresponding to is defined to be:
\beq
A=V(P)^\dag \text{d}V(P),
\eeq
which is a generalization of the WZ connection in Eq.~(\ref{eq:WZ_connection}). This is the unique connection that is invariant under the transformation in Eq.~(\ref{eq:gauge_transform}):
\beq
\begin{split}
\tilde{A}=&h^\dag V(P)^\dag \text{d}\left(V(P)h\right)\\
=&h^\dag A h + h^\dag \text{d}h.
\end{split}
\eeq

We apply this formalism to the system dynamic of HQC. The state vector $|\psi(t)\>\in\mathbb{C}^N$ evolves according to the Schr\"{o}dinger equation:
\beq\label{eq:shrodinger_2}
i\frac{\text{d}}{\text{d}t}|\psi(t)\>=H(t)|\psi(t)\>.
\eeq
The Hamiltonian has a spectral decomposition,
\beq
H(t)=\sum_{l=0}^L\varepsilon_l(t)P_l(t),
\eeq
with projection operators $P_l(t)$. Therefore, the set of energy eigenvalues $(\varepsilon_0(t),\ldots,\varepsilon_L(t))$ and orthogonal projectors $(P_0(t),\ldots, P_l(t))$ encodes the information of the control parameters of the system. For the ground space, we write $P_0(t)$ as $P(t)$ for simplicity. Suppose the degeneracy $K=\text{Tr}\{P(t)\}$ is constant. For all $t$, there exists $V(t)\in S_{N,K}(\mathbb{C})$ such that $P(t)=V(t)V^\dag(t)$. By the adiabatic approximation, we can substitute for $|\psi(t)\>\in\mathbb{C}^N$ a reduced state vector $\phi(t)\in \mathbb{C}^K$:
\beq
|\psi(t)\>=V(t)\phi(t).
\eeq
Since $H(t)|\psi(t)\>=\varepsilon_0(t)|\psi(t)\>$, the Schr\"{o}dinger equation ~(\ref{eq:shrodinger_2}) becomes
\beq
\frac{\text{d}\phi}{\text{\text{d}}t}+V^\dag\frac{\text{d}V}{\text{d}t}\phi(t)=\varepsilon_0(t)V(t)\phi(t),
\eeq
and the solution can be represented formally as
\beq
\phi(t)=e^{-i\int_0^t\varepsilon_0(\tau)\text{d}\tau}\mathcal {P}\exp\left(-\int V^\dag \text{d}V\right)\phi(0).
\eeq
Therefore, $\psi(t)$ can be written
\beq\label{eq:adiabatic_evolution}
|\psi(t)\>=e^{-i\int_0^t\varepsilon_0(\tau)\text{d}\tau}V(t)\mathcal{P}\exp\left(-\int V^\dag \text{d}V\right)V^\dag(0)|\psi(0)\>.
\eeq
In particular, if the system comes back to its initial point, as $P(T)=P(0)$, the holonomy $\Gamma\in\text{U}(K)$ is defined as
\beq\label{eq:holonomy}
\Gamma=V^\dag(0) V(T)\mathcal{P}\exp\left(-\int V^\dag \text{d}V
\right),
\eeq
and the final state is
\beq
\begin{split}
|\psi(T)\>&=e^{-i\int_0^t\varepsilon_0(\tau)\text{d}\tau}V(0)\Gamma\phi(0).\\
\end{split}
\eeq
According to the formula above, an operation $\Gamma\in {\rm U}(K)$ is applied to the ground space.

If the condition
\beq\label{eq:horizontal_condition}
V^\dag\cdot\frac{\text{d}V}{\text{d}t}=0,
\eeq
is satisfied for all $t$, the curve $V(t)$ in $S_{N,K}(\mathbb{C})$ is called a horizontal lift of the curve $P(t)=\pi(V(t))$ in $G_{N,K}(\mathbb{C})$.Then the holonomy ~(\ref{eq:holonomy}) is greatly simplified to
\beq\label{eq:horizontal_holonomy}
\Gamma=V^\dag(0)\cdot V(T)\in\text{U}(K).
\eeq

%Now we are ready to reformulate our problem stated in the end of Sec.~\ref{sec:hc_general}.

For closed-loop HQC, given a desired unitary operation
$U_{\text{op}}\in \text{U}(K)$ and
a fixed initial point $P(0)\in G_{N,K}(\mathbb{C})$,
we want to find a loop $P(t)\in G_{N,K}(\mathbb{C})$ with base points
$P(0)=P(T)$ whose horizontal lift $V(t)\in S_{N,K}(\mathbb{C})$
produces holonomy $\G = U_{\text{op}}$
according to Eq.~(\ref{eq:horizontal_holonomy}). For open-loop adiabatic code deformation, Eq.~(\ref{eq:adiabatic_evolution}) is general to obtain the state evolution when the adiabatic condition is satisfied.

%
%In Sec.~\ref{sec:scheme}, we will discuss in detail how to find such loop of $P(t)$ whose horizontal lift gives the holonomy we desired in the code space.

%A visualization of horizontal lift is shown in Fig.~\ref{Fig:horizontal_lift}.

Without loss of generality, we can always restrict ourselves to the case such that $P(t)$ has the form:
\beq
P(t)=U(t,0)P(0)U^{\dag}(t,0)=U(t,0)v_0v_0^\dag U^{\dag}(t,0),
\eeq
for some smooth $U(t,0)\in {\rm U}(N)$ according to Eq.~(\ref{eq:space_transform}). Note here, $U(t,0)$ should be chosen such that in general, at any time $t$,
\beq
U(t+\tau,t)P(t)U^{\dag}(t+\tau,t)\neq P(t),
\eeq
for some neighborhood of $t$. In other word, $U(t)$ must not be in the isotropy group of $P(t)$. This condition can also stated as
\beq\label{eq:eff_evolution}
\left[\frac{\partial}{\partial\tau}U(t+\tau,t)|_{\tau=0},P(t)\right]\neq 0.
\eeq
The case where Eq.~(\ref{eq:eff_evolution}) equals 0 is allowed only at a finite number of points in $[0,T]$. The horizontal curve should satisfy the following set of equations:
\beq\label{eq:horizontal_equation_group}
\begin{split}
    V^\dag\cdot\frac{\text{d}V}{\text{d}t}&=0, \\
    P(t)=V(t)V^\dag(t)&= U(t,0)v_0v_0^\dag U^{\dag}(t,0).
\end{split}
\eeq
The general solution to these equations can
be written as:
\beq\label{eq:general_curve}
V(t)=U(t,0)v_0h(t,0)
\eeq
for some $h(t,0)\in {\rm U}(K)$. Substituting Eq.~(\ref{eq:general_curve}) into Eq.~(\ref{eq:horizontal_equation_group})
we get:
\beq\label{eq:ht_evolution}
\dot{h}(t,0)=-v_0^\dag U^{\dag}(t,0)\dot{U}(t,0)v_0 h(t,0),
\eeq
which completely determines the $h(t)$, horizontal lift, and state evolution
for a given adiabatic process.

\section{Proof of Lemma~\ref{lemma:state_evolution},\ \  \ref{lemma:error_propagation},\ \ \ref{lemma:parrallelism}}
We first prove a lemma which will be used to prove other lemmas:
\begin{mylemma}\label{lemma:operator}
$\forall g_q\in\mathcal{G}$ is in the normalizer of $G_n$.
\end{mylemma}
\begin{proof}
For any $M\in G_n$, either $[M,Q_q]=0$ or $\{M,Q_q\}=0$. In the second case, we have $[Q_q,M]=2 Q_qM = 2 M^\prime$, with $M^\prime\in G_n$.
\beq
\begin{split}
g_q M g_q^\dag =& \exp\left(i\frac{\pi}{4} Q_q\right) M \exp\left(-i\frac{\pi}{4} Q_q\right)\\
=& M + i\frac{\pi}{4}[Q_q, M] - \frac{\pi^2}{16\cdot 2!}[Q_q,[Q_q, M]]\ldots\\
=&\cos(\pi/2) M + i\sin(\pi/2) M^\prime \\
=&i M^\prime.
\end{split}
\eeq
Further, if $M$, $Q_q$ are Hermitian, $M^\prime$ is anti-Hermitian and $g_qMg_q^\dag$ is Hermitian.
\end{proof}

\subsection{Lemma~\ref{lemma:state_evolution}}\label{sec:proof_lemma_state_evolution}
The deformation of the Hamiltonian is isospectral, so the number of logical qubits encoded in the ground space is constant, say $k$. The horizontal lift $V_0(t)$ for $P_0(t)$ in general can be written as $V_0(t)=U_q(t,t_{q-1})V_0(t_{q-1})h(t,t_{q-1})$. From Eq.~(\ref{eq:ht_evolution}), $U_q^\dag(t,t_{q-1})\partial_tU(t,t_{q-1})=i\partial_tf_q(t)Q_q$,
\beq
\frac{\partial h}{\partial t}=iV_0^\dag(t_{q-1})\partial_tf_q(t)Q_q V_0(t_{q-1})
\eeq
for $t\in[t_{q-1},t_q]$, and
\beq\label{eq:ht_theom}
V_0(t_{q-1})\partial_t{h}(t,0)V_0^\dag(t_{q-1})=iP_0(t_{q-1})\partial_tf_q(t)Q_qP_0(t_{q-1}).
\eeq
Since $S_j(t_0)\in G_n$ for all $j$, $g_l\in G_n$, for all $l$.
\beq
\begin{split}
P_0(t_{q-1})&=\left(\prod_{l=1}^{q-1}g_{q}\right)P(0)\left(\prod_{l=1}^{q-1}g_{l}\right)^\dag\\
&=\left(\prod_{l=1}^{q-1}g_{q}\right) \prod_{j=0}^{n-k}\frac{I+S_j(0)}{2}\left(\prod_{l=1}^{q-1}g_{q}\right)^\dag\\
&=\prod_{j=1}^{n-k}\frac{I+S_j(t_{q-1})}{2},
\end{split}
\eeq
where $S_j(t_{q-1})=\left(\prod_{l=1}^{q-1}g_{l}\right)S_j(t_0)\left(\prod_{l=1}^{q-1}g_{l}\right)^\dag$ is in $G_n$ because $\{g_q\}$ are all in the normalizer of $G_n$ (Lemma.~\ref{lemma:operator}). Since $[Q_q,H(t_{q-1})]\neq0$, so there exists at least one $S_j(t_{q-1})$ such that $\{Q_q, S_j(t_{q-1})\}=0$. According to Eq.~(\ref{eq:ht_theom}), $V_0(t_{q-1})\partial_th(t,0)V_0^\dag(t_{q-1})=0$ and $h(t,t_q)=I$. Thus $V_0(t)=U_q(t,t_{q-1})V_0(t_{q-1})$ and $V_0(t)=U_q(t,t_{q-1})\left(\prod_{l=1}^{q-1}g_{l}\right)V_0(t_0)$.
From Eq.~(\ref{eq:adiabatic_evolution}).
\beq
|\psi(t)\>=e^{-i\varepsilon_0(t-t_{q-1})}U_q(t,t_{q-1})\left(\prod_{l=1}^{q-1}g_{l}\right)|\psi(t_0)\>.
\eeq
Setting $q=p$ and $t=t_p$, we get
\beq
|\psi(t_p)\>=e^{-i\varepsilon_0(t_p-t_{q-1})}\Omega_p|\psi(t_0)\>.
\eeq

\subsection{Lemma~\ref{lemma:error_propagation}}\label{sec:proof_error_propagation}
First, we show that for any $\alpha\neq\beta$, the adiabatic condition for $P_{\textbf{s}_\alpha}$ and $P_{\textbf{s}_\beta}$ is satisfied. We have $S_j(t_l) \in G_n$ according to Lemma.~\ref{lemma:operator} for $1\leq l\leq q$. Consider the time segment $[t_q,t_{q+1}]$ first. Define the index set $\mathscr{I}=\{1,2,\ldots, n-k
\}$ to be the number of terms in the Hamiltonian $H(t_q)$ with sets $\mathscr{A}_\alpha=\{j\in\mathscr{I}|\{S_j(t_q),F_\alpha\}=0\}$, $\mathscr{B}_\alpha=\mathscr{I}\backslash \mathscr{A}_\alpha$, $\mathscr{C}_{Q_l}=\{j\in\mathscr{I}|\{S_j(t_q),Q_l\}=0\}$ and $\mathscr{D}_{Q_l}=\mathscr{I}\backslash \mathscr{C}_{Q_l}$. Since $F_\alpha\in G_n$,
\beq
\begin{split}
&F^q_\alpha P_0(t_q) \left(F^q_\alpha\right)^\dag\\
=&\mathlarger{\prod}_{j\in \mathscr{A}_\alpha}\frac{I+S_j(t_q)}{2}\mathlarger{\prod}_{j^\prime\in \mathscr{B}_\alpha}\frac{I-S_j^\prime(t_q)}{2}\\
=&\underbrace{\mathlarger{\prod}_{m\in\mathscr{C}_{Q_{q+1}}}\frac{I+s_{\alpha_m} S_m(t_q)}{2}}_{P_{\textbf{s}_\alpha}^\mathscr{C}}
\cdot\underbrace{\mathlarger{\prod}_{m^\prime\in\mathscr{D}_{Q_{q+1}}}\frac{I+s_{\alpha_{m^\prime}}S_{m^\prime}(t_q)}{2}}
_{P_{\textbf{s}_\alpha}^\mathscr{D}}\\
=&P_{\textbf{s}_\alpha}(t_q).
\end{split}
\eeq
Here, $P_{\textbf{s}_\alpha}^\mathscr{C}$ and $P_{\textbf{s}_\alpha}^\mathscr{D}$ are short for $P_{\textbf{s}_\alpha}^{\mathscr{C}_{Q_{q+1}}}$ and $P_{\textbf{s}_\alpha}^{\mathscr{D}_{Q_{q+1}}}$.  For any $\beta\neq \alpha$,
\beq\label{eq:ad_cond_calculation}
\begin{split}
&P_{\textbf{s}_\alpha}(t) \frac{\partial H(t)}{\partial t} P_{\textbf{s}_\beta}(t)=\\
&-i\big(\partial_t f_{q}(t)\big)U_{q+1}(t,t_q)
\Big[P_{\textbf{s}_\alpha}(t_q)Q_{q+1}H(t_q)P_{\textbf{s}_\beta}(t_q)-\\
&P_{\textbf{s}_\alpha}(t_q)H(t_q)Q_{q+1}P_{\textbf{s}_\beta}(t_q)\Big]U_{q+1}^\dag(t,t_q),
\end{split}
\eeq
where $U_{q+1}(t,t_{q})=\exp\left(if_{q+1}(t)Q_{q+1}\right)$. We examine the two terms in the square brackets:
\beq
\begin{split}
&P_{\textbf{s}_\alpha}(t_q)Q_{q+1}H(t_{q})P_{s_\beta}(t_q)\\
=&\varepsilon_{\textbf{s}_\beta}(t_q)Q_{q+1}\mathlarger{\prod}_{m\in\mathscr{C}_{Q_{q+1}}}\frac{I-s_{\alpha_m}S_m(t_q)}{2}P_{\textbf{s}_\alpha}^{\mathscr{D}}(t_q)
P_{\textbf{s}_\beta}^{\mathscr{C}}(t_q)P_{\textbf{s}_\beta}^{\mathscr{D}}(t_q),
\end{split}
\eeq
and
\beq
\begin{split}
&P_{\textbf{s}_\alpha}(t_q)H(t_{q})Q_{q+1}P_{s_\beta}(t_q)\\
=&\varepsilon_{\textbf{s}_\alpha}(t_q)P_{\textbf{s}_\alpha}^{\mathscr{
C}}(t_q)P_{\textbf{s}_\alpha}^{\mathscr{D}}(t_q)P_{\textbf{s}_\beta}^{\mathscr{D}}(t_q)\mathlarger{\prod}_{m\in\mathscr{C}_{Q_{q+1}}}\frac{I-s_{\beta_m}S_m(t_q)}{2}Q_{q+1}.
\end{split}
\eeq
For those $\textbf{s}_\beta$ such that $s_{\alpha_m}\neq s_{\beta_m}$ for any $m\in \mathscr{D}_{Q_{q+1}}$, Eq.~(\ref{eq:ad_cond_calculation}) will be zero, and the adiabatic condition will be satisfied automatically. For those $\textbf{s}_\beta$ such that $s_{\alpha_m}=s_{\beta_m}$ for all $m\in \mathscr{D}_{Q_{q+1}}$, it's easy to check the above two expression are not equal to zero only if  $s_{\beta_m}=-s_{\alpha_m}$ for all $m\in \mathscr{C}_{Q_{q+1}}$. Therefore, there is only one $\beta$ such that $P_{\textbf{s}_\alpha}(t)\partial_t\big(H(t)\big)P_{\textbf{s}_\beta}(t)\neq0$ and hence that needs further checking. For that specific $\beta$, we have a simple relation:
\beq
Q_{q+1}P_{\textbf{s}_{\alpha}}(t_q)Q_{q+1}^\dag=P_{\textbf{s}_\beta}(t_q),
\eeq
and
\beq
\begin{split}
&\big \| P_{\textbf{s}_\alpha}(t)\frac{\partial H(t)}{\partial t}P_{\textbf{s}_\beta}(t_q)  \big\|_1\\
=&\partial_tf_{q+1}(t)\big|\varepsilon_{\textbf{s}_\alpha}(t_q)-\varepsilon_{\textbf{s}_\beta}(t_q) \big| \cdot \big\| P_{\textbf{s}_\alpha}(t_q)Q_{q+1} \big \|_1\\
=&K\partial_t{f}_{q+1}(t)\big|\varepsilon_{\textbf{s}_\alpha}(t_q)-\varepsilon_{\textbf{s}_\beta}(t_q) \big|.
\end{split}
\eeq
The left hand side of Eq.~(\ref{eq:adiabatic_condition_general}) reduces to
\beq
\frac{|\partial_t{f}_{q+1}(t)|}{\big|\varepsilon_{\textbf{s}_\alpha}(t_q)-\varepsilon_{\textbf{s}_\beta}(t_q) \big|},
\eeq
since $|\mathscr{C}_{Q_{q+1}}|$ is odd. We have
\beq
\big|\varepsilon_{\textbf{s}_\alpha}(t_q)-\varepsilon_{\textbf{s}_\beta}(t_q) \big| =\Big|\mathlarger{\sum}_{m\in \mathscr{C}_{Q_{q+1}}}2s_{\alpha_m}\Big|\geq 2.
\eeq
If $\partial_t{f}_{q+1}(t) \ll 1$ is satisfied (which is always possible by setting appropriate controls), then $P_{\textbf{s}_\alpha}(t)$ satisfies the adiabatic condition for time segment $t\in[t_q,t_{q+1}]$. The same argument can be applied to the time segments $l > q$ to show that the adiabatic condition can be satisfied between $P_{\textbf{s}_\beta}(t)$ and $P_{\textbf{s}_\alpha}(t)$ for any $\beta$. According to Eq.~(\ref{eq:adiabatic_evolution}),
\beq
\begin{split}
|\psi(t_p)\>&\propto V_{\textbf{s}_\alpha}(t_p)\left(F_\alpha^q V_0(t_q)\right)^\dag F^q_\alpha V_0(t_q)V_0^\dag(0)|\psi(0)\>\\
&= V_{\textbf{s}_\alpha}(t_p)V_0^\dag(t_q)V_0(t_q)V_0^\dag(0)|\psi(0)\>\\
&= V_{\textbf{s}_\alpha}(t_p)V_0^\dag(0)|\psi(0)\>,
\end{split}
\eeq
where $V_{\textbf{s}_\alpha}(t)$ is defined as
\beq
V_{\textbf{s}_\alpha}(t)= U(t,t_q)F^q_{\alpha}V_0(t_q)h(t,t_q), \ \ t>t_q,
\eeq
and is the horizontal lift of $P_{\textbf{s}_\alpha}(t)$ given the initial condition $F_{\alpha}V_0(t_q)$. From the same argument in the proof of Lemma~\ref{lemma:state_evolution}, we get
\beq
\begin{split}
|\psi(t_p)\>=& \mathlarger{\sum}_\alpha c_\alpha e^{-i\varepsilon_{\textbf{s}_\alpha}(t_p-t_q)}\left(\prod_{l=q+1}^{p}g_{l}\right)F^q_\alpha|\psi(t_q)\>\\
=&\mathlarger{\sum}_\alpha c_\alpha e^{-i\varepsilon_{\textbf{s}_\alpha}(t_p-t_q)}F_\alpha^{pq}\left(\prod_{l=1}^{p}g_{l}\right)|\psi(t_0)\>.
\end{split}
\eeq

\subsection{Lemma~\ref{lemma:parrallelism}}\label{sec:proof_parallesim}
For part 1, according to condition~1,
\beq
U_{q+1}(t,t_q)=\exp\left(i\sum_{r=q+1}^{q+M}f(t)Q_r\right),
\eeq
 for $t\in[t_q,
t_{q+1}]$. From the procedure in the proof of Lemma~\ref{lemma:state_evolution},
\beq
\begin{split}
\frac{\partial h}{\partial t}=&if(t)V_{\textbf{s}_\alpha}^\dag(t_q)P_{\textbf{s}_\alpha}(t_q)\left(\mathlarger{\sum}_{Q_r\in \mathscr{P}_q} Q_r\right)P_{\textbf{s}_\alpha}(t_q)V_{\textbf{s}_\alpha}(t_q)\\
=&0,
\end{split}
\eeq
according to $Q_r\in G_n$
and
\beq
|\psi(t)\>=e^{-i\varepsilon_0(t-t_{q-1})}U_{q+1}(t,t_q)|\psi(t_q)\>.
\eeq
When $t=t_{q+1}$, when $f(t_{q+1})=\pi/4$, and
\beq
|\psi(t_{q+1})\>=e^{-i\varepsilon_0(t_{q+1}-t_{q})}\left(\mathlarger{\prod}_{l=q+1}^{q+M}g_{l}\right) |\psi(t_q)\>,
\eeq
under the adiabatic approximation.

For part 2, suppose $F_\alpha^q$ takes the system from the ground space to $P_{\textbf{s}_\alpha}$. Then for any $\beta\neq\alpha$,
\beq
\begin{split}
&P_{\textbf{s}_\alpha}(t)\frac{\partial H(t)}{\partial t}P_{\textbf{s}_\beta}(t) =\\
&i\left(\partial_tf(t)\right)U_{q+1}(t,t_q)\sum_{Q_r\in\mathscr{P}_q}\Big[ P_{\textbf{s}_\alpha}(t_q)Q_r H(t_q)P_{\textbf{s}_\beta}(t_q)-\\
&P_{\textbf{s}_\alpha}(t_q)H(t_q)Q_r P_{\textbf{s}_\beta}(t_q)\Big]U_{q+1}^\dag(t,t_q).
\end{split}
\eeq
By the same argument as in the proof of Lemma~\ref{lemma:error_propagation}, for each $Q_r$, there is only one $\beta_r$ such that $P_{\textbf{s}_\alpha}(t_q)Q_r H(t_q)P_{\textbf{s}_{\beta_r}}(t_q)$ and $P_{\textbf{s}_\alpha}(t_q)H(t_q) Q_r P_{\textbf{s}_{\beta_r}}(t_q)$ do not equal 0. Since $\mathscr{C}_{Q_{r}} \bigcap \mathscr{C}_{Q_{m}}=\emptyset$ for any $Q_r,Q_m \in \mathscr{P}_q$, $\beta_r\neq \beta_m$ when $r\neq m$. Then, for any such $\beta_r$,
\beq
\big \| P_{\textbf{s}_\alpha}(t)\frac{\partial H(t)}{\partial t}P_{\textbf{s}_{\beta_r}}(t) \big \|_1
=K\partial_t{f}(t)\  \big|\ \varepsilon_{\textbf{s}_\alpha}(t_q)-\varepsilon_{\textbf{s}_{\beta_r}}(t_q)\big|.
\eeq
Since $|\mathscr{C}_{Q_{r}}|$ is odd, then $|\varepsilon_{\textbf{s}_\alpha} (t_q)-\varepsilon_{\textbf{s}_{\beta_r}}(t_q)|\geq 2$, the adiabatic condition Eq.~(\ref{eq:adiabatic_condition_general}) holds for arbitrary $\beta$, and we get
\beq
|\psi(t_{q+1})\>=\mathlarger{\sum}_\alpha c_\alpha e^{-i\varepsilon_{\textbf{s}_\alpha}(t_{q+1}-t_q)}F_\alpha^{q+1,q}
\left(\mathlarger{\prod}_{l=q+1}^{q+M}g_{l}\right) |\psi(t_q)\>.
\eeq

\newpage
%\bibliographystyle{apsrev4-1}
%\bibliography{refs}
%

\end{document}